%% file: GPBEvPD.tex
\documentclass[english]{article}
\usepackage[T1]{fontenc}
\usepackage[latin9]{inputenc}
\usepackage{color}
\usepackage{babel}
\usepackage{array}
\usepackage{float}
\usepackage{units}
\usepackage{mathtools}
\usepackage{bm}
\usepackage{varwidth}
\usepackage{amsmath}
\usepackage{amsthm}
\usepackage{amssymb}
\usepackage{graphicx}
\usepackage{geometry}
\usepackage[pdfusetitle,
 bookmarks=true,bookmarksnumbered=false,bookmarksopen=false,
 breaklinks=false,pdfborder={0 0 1},backref=false,colorlinks=false]
 {hyperref}

\makeatletter

\newcommand{\noun}[1]{\textsc{#1}}
\providecommand{\tabularnewline}{\\}
\newenvironment{cellvarwidth}[1][t]
    {\begin{varwidth}[#1]{\linewidth}}
    {\@finalstrut\@arstrutbox\end{varwidth}}
\floatstyle{ruled}
\newfloat{algorithm}{tbp}{loa}
\providecommand{\algorithmname}{Algorithm}
\floatname{algorithm}{\protect\algorithmname}

\theoremstyle{plain}
\newtheorem{thm}{\protect\theoremname}
\theoremstyle{definition}
\newtheorem{defn}[thm]{\protect\definitionname}
\theoremstyle{remark}
\newtheorem{rem}[thm]{\protect\remarkname}
\theoremstyle{plain}
\newtheorem{fact}[thm]{\protect\factname}
\theoremstyle{definition}
\newtheorem{example}[thm]{\protect\examplename}
\theoremstyle{plain}
\newtheorem{lem}[thm]{\protect\lemmaname}
\theoremstyle{definition}
\newtheorem{problem}[thm]{\protect\problemname}
\theoremstyle{plain}
\newtheorem{prop}[thm]{\protect\propositionname}

\usepackage{algorithmic}
\usepackage{qcircuit}
\usepackage{array}
\usepackage{tkz-graph}
\usepackage{placeins}

\usepackage{tikz}
\usetikzlibrary{positioning}
\usetikzlibrary{arrows.meta}
\usetikzlibrary{shapes.geometric}
\usetikzlibrary{calc}
\usetikzlibrary{backgrounds}
\usetikzlibrary{fit}
\usetikzlibrary{shadows}

\usetikzlibrary{arrows.spaced}

\usetikzlibrary{matrix}

\usetikzlibrary{shapes.misc}

\usepackage{amsmath}
\usepackage{amssymb}

\usepackage{url}

\makeatother

\providecommand{\definitionname}{Definition}
\providecommand{\examplename}{Example}
\providecommand{\factname}{Fact}
\providecommand{\lemmaname}{Lemma}
\providecommand{\problemname}{Problem}
\providecommand{\propositionname}{Proposition}
\providecommand{\remarkname}{Remark}
\providecommand{\theoremname}{Theorem}

\begin{document}
\include{lyxmacros}

\title{On Encoding Matrices using Quantum Circuits}
\author{Liron Mor Yosef\thanks{Tel Aviv University, lm2@mail.tau.ac.il} \\
 \and Haim Avron\thanks{Tel Aviv University, haimav@tauex.tau.ac.il}}
\maketitle
\begin{abstract}
Over a decade ago, it was demonstrated that quantum computing has
the potential to revolutionize numerical linear algebra by enabling
algorithms with complexity superior to what is classically achievable,
e.g., the seminal HHL algorithm for solving linear systems. Efficient
execution of such algorithms critically depends on representing inputs
(matrices and vectors) as quantum circuits that encode or implement
these inputs. For that task, two common circuit representations emerged
in the literature: block encodings and state preparation circuits.
In this paper, we systematically study encodings matrices in the form
of block encodings and state preparation circuits. We examine methods
for constructing these representations from matrices given in classical
form, as well as quantum two-way conversions between circuit representations.
Two key results we establish (among others) are: (a) a general method
for efficiently constructing a block encoding of an arbitrary matrix
given in classical form (entries stored in classical random access
memory); and (b) low-overhead, bidirectional conversion algorithms
between block encodings and state preparation circuits, showing that
these models are essentially equivalent. From a technical perspective,
two central components of our constructions are: (i) a special constant-depth
multiplexer that simultaneously multiplexes all higher-order Pauli
matrices of a given size, and (ii) an algorithm for performing a quantum
conversion between a matrix\textquoteright s expansion in the standard
basis and its expansion in the basis of higher-order Pauli matrices.
\end{abstract}

\section{Introduction}

Following breakthroughs such as the Harrow--Hassidim--Lloyd (HHL)
algorithm\textcolor{black}{~\cite{harrow2009quantum}} and subsequent
works\textcolor{black}{~\cite{childs2017quantum,gilyen2019quantum,subacsi2019quantum,an2022quantum,lin2020optimal}},
numerical linear algebra has emerged as an area in which quantum computing
might deliver substantial improvements over classical computing. Such
quantum algorithms, which we refer to as quantum numerical linear
algebra algorithms, exploit the ability of quantum computers to efficiently
represent and manipulate high-dimensional vectors and matrices using
exponentially fewer physical resources than is possible in classical
numerical linear algebra. However, the success of quantum numerical
linear algebra algorithms crucially depends on how input numerical
arrays (such as matrices and vectors) are encoded into quantum circuits.
Efficient representation of matrices in quantum computing is therefore
not a secondary concern, but a central component that determines both
the theoretical and practical feasibility of quantum numerical linear
algebra.

Two main paradigms have been developed for representing numerical
arrays within quantum circuits. The first, block encoding~\cite{gilyen2019quantum,clader2022quantum,camps2022fable},
embeds a scaled version of a matrix as a sub-block of a larger unitary
operator. Block encodings forms the standard input model for many
modern quantum numerical linear algebra algorithms~\cite{chakraborty2018power,martyn2021grand,takahira2021quantum}.
The second paradigm is the use of state preparation circuits (aka
amplitude encoding)~\cite{mottonen2004transformation,shende2005synthesis,plesch2011quantum,low2024trading,araujo2021divide,araujo2023configurable,gleinig2021efficient,zhang2022quantum},
which encode the entries of a matrix as probability amplitudes of
quantum states. Typically, state preparation circuits are used to
represent input vectors, e.g., the right hand side in HHL\textcolor{black}{~\cite{harrow2009quantum}.
Nevertheless, in our previous work}~\cite{yosef2024multivariate},
we explored the use state preparation circuits to encode matrices,
developed a limited yet useful library of matrix computational kernels
on such circuits, and explored applications in the context of estimating
multivariate traces and spectral sums. Although both models can represent
matrices for quantum computation, state preparation circuits are often
simpler to construct than block encodings, which typically require
both PREPARE and SELECT oracles and can incur data loading costs that
dominate the runtime\textcolor{black}{~\cite{low2019hamiltonian,sunderhauf2024block,clader2022quantum,yuan2023optimal}.}

Despite close conceptual relationship between the aforementioned matrix
encoding methods, no general and efficient algorithm has previously
existed for translating between these two representations. This limitation
prevents quantum algorithms from flexibly combining algebraic and
computational tools associated with each model. Establishing such
an algorithm serves not only clarify the theoretical connection between
them, but also enables hybrid algorithmic designs that can exploit
whichever representation is most natural for a given problem instance.

In this paper, we systematically study encoding matrices in the form
of block encodings and state preparation circuits. We examine methods
for constructing these representations from matrices given in classical
form, as well as quantum two-way conversions between circuit representations.
Two key results we establish (among others) are: (a) a general method
for efficiently constructing a block encoding of an arbitrary matrix
given in classical form (entries stored in classical random access
memory); and (b) a low-overhead, bidirectional conversion between
block encodings and state preparation circuits, showing that these
models are essentially equivalent. Our previous work~\cite[Section 3.2.2]{yosef2024multivariate}
partially addressed conversion from block encoding to matrix state
preparation, establishing that exact block encodings can be converted
via a simple circuit transformation. In this paper we complete the
analysis by addressing arbitrary block encodings and introducing the
dual conversion: from state preparation to block encoding. The forward
conversion keeps the conceptual simplicity of the exact block encoding
conversion in~\cite[Section 3.2.2]{yosef2024multivariate}, while
the reverse conversion requires substantially more complex circuits
and additional qubits.

From a technical perspective, two central components of our constructions
are: (i) a special constant-depth multiplexer that simultaneously
multiplexes all higher-order Pauli matrices of a given size, and (ii)
an algorithm for performing a quantum conversion between a matrix\textquoteright s
expansion in the standard basis and its expansion in the basis of
higher-order Pauli matrices.

The remainder of the paper is organized as follows. Section~\ref{sec:preliminaries}
introduces mathematical preliminaries and notation. Section~\ref{subsec:rotation-multiplexers}
discuss both general construction of mulitplexers, and ones that are
specific to rotation mulitplexers. Section~\ref{sec:Converting-between-standard}
presents classical and quantum conversion techniques between standard
and Pauli bases. Section~\ref{sec:multiplexing-pauli-coefficient}
discusses how to multiplex higher order Pauli matrices. Section~\ref{sec:From-dense-classical}
presents algorithms for directly constructing block-encoding circuits
from classical matrix representations, and Section~\ref{sec:algorithms-encoding}
discuss algorithms for moving between representations, both from classical
to quantum and from quantum to quantum (of different types).

\subsection{Summary of contributions}

Since we systematically study the two aforementioned matrix encodings
schemes, our paper has many contributions, some of which we highlighted
previously. Here we summarize all the contributions of the paper.
See Table~\ref{tab:algorithm-comparison} for a summary of algorithmic
(classical and quantum) complexities.
\begin{enumerate}
\item Efficient conversion between the standard and Pauli basis representation
of a matrix, both in a classical algorithm (Subsection~\ref{subsec:Pauli-decomposition-classiq})
and quantumly in the form of converting an encoding in one basis to
the other (Subsection~\ref{subsec:ua-to-uap}, Theorem~\ref{thm:ua-to-uap-and-back}).
Conversions are two-way.
\item Efficient algorithms for building block encodings of matrices given
their classical representation either in the standard (Subsection~\ref{subsubsec:Matrix2Be})
or the Pauli (Subsection~\ref{subsubsec:PauliMatrix2Be}) basis.
Our algorithms produce circuits with lower depth than all previously
published algorithms with number of qubits that is logarithmic in
the matrix size. Is some cases, it produces block encoding with superior
scale factor than all the one produced by previous algorithms.
\item Efficient algorithms for building of block encodings of matrices given
as state preparation circuits either in the standard or the Pauli
basis. Both conversion algorithms are reported in Subsection~\ref{subsubsec:PauliMsp2Be_Msp2Be}.
\item Efficient algorithm for building a matrix state preparation circuit
given a block encoding of the matrix (Subsection~\ref{subsubsec:Be2Msp_Be2PauliMsp}).
\item Pauli Multiplexer Circuit (${\cal PMX}_{n}$): we design and analyze
an ultra-efficient quantum circuit that multiplexes the complete set
of $n$-qubit higher-order Pauli operators. The ${\cal PMX}_{n}$
circuit achieves constant gate depth and a T-depth of one, independent
of $n$, with a logarithmic gate count that scales as $O(n)$. This
efficiency arises from an ultra-sparse structure in the multiplexer\textquoteright s
rotation angles after a Walsh--Hadamard transform. See Section~\ref{sec:multiplexing-pauli-coefficient}.
This technical contribution is necessary for our conversion algorithms,
but might be of independent interest.
\end{enumerate}

\subsection{Related work}

Many quantum algorithms require encoding non-unitary matrices in quantum
circuits. Block encoding addresses this challenge by embedding arbitrary
matrices within larger unitary operators, enabling their use in quantum
circuits. This technique, combined with quantum singular value transformation
(QSVT), provides a framework for performing various linear algebra
operations on quantum computers~\cite{chakraborty2018power,gilyen2019quantum,takahira2021quantum,martyn2021grand}.

Block encoding represents a scaled matrix $\matA/\alpha\in\mathbb{C}^{N\times N}$
by embedding it within a larger unitary matrix $\matU\in\mathbb{C}^{2^{p}N\times2^{p}N}$,
formally expressed as $\matA=(\bra 0\otimes\matI_{2^{p}})\matU(\ket 0\otimes\matI_{2^{p}})$
where $\matI_{2^{p}}$ is the $2^{p}\times2^{p}$ identity matrix
for $p$ ancillary qubits. However, efficient construction of block
encodings is challenging, with works advocating the use of resources
like QRAM~\cite{giovannetti2008quantum,clader2022quantum}. Related
is the Linear Combination of Unitaries (LCU) framework which is a
central and widely-used strategy for synthesizing block encodings~\cite{kothari2014efficient,berry2015hamiltonian}.

Early works on synthesizing block encodings mainly addressed construction
of block encodings for sparse matrices. Childs and Kothari developed
a pioneering technique for simulating sparse Hamiltonians, which was
later adapted for block encoding applications~\cite{childs2011simulating}.
Subsequent works introduced powerful techniques based on quantum walks
and qubitization~\cite{berry2015hamiltonian,low2019hamiltonian}.
Later, the introduction of QSVT provided a unifying framework that
uses block encodings to perform advanced matrix arithmetic, such as
inversion and filtering~\cite{gilyen2019quantum}. Crucially, these
foundational algorithms are typically analyzed in the oracle model,
and their efficiency is measured in query complexity which is the
number of calls to an oracle that provides information about the matrix's
non-zero entries. This abstract complexity model does not directly
translate to the explicit gate-level circuit costs that are the focus
of our work.

As for gate-level circuit constructions, Camps et al. show how to
build block encoding circuits for structured sparse matrices~\cite{camps2024explicit}.
For a scaled $s$-sparse matrix $\matA$ ($\TNorm{\matA}\le1$), they
block-encode $\nicefrac{\matA}{s}$ using two pieces: a structured
unitary that points to each nonzero entry and a value unitary that
writes that entry's magnitude. For symmetric stochastic matrices $\matP$,
they avoid the $\nicefrac{1}{s}$ scaling by constructing a symmetric
block-encoding of $\matP$ itself (rather than $\nicefrac{\matP}{s}$).

Several works considered block encodings of dense matrices. Camps
et al. introduced FABLE: an efficient method for constructing block
encoding quantum circuits with circuit depth $O(N^{2})$~\cite{camps2022fable}.
Their approach revolves around encoding matrix entries through angles
derived from their inverse cosine, implementing these angles using
${\cal R}_{y}$ rotations with a one-dimensional multiplexer. The
method applies Hadamard gates on both sides (which can be viewed as
multiplication by all-ones vectors from both sides), resulting in
additional scaling factor of $O(\nicefrac{1}{N})$ in the block encoding's
prefactor. While this scaling factor potentially requires additional
measurement shots to achieve desired precision, the method's efficient
circuit depth makes it attractive. Recent variants extend FABLE to
sparse inputs: the S-FABLE and \textquotedblleft Lazy\textquotedblright{}
S-FABLE (LS-FABLE) schemes adapt FABLE to unstructured sparse matrices,
empirically reducing resources (for matrices with with $O(N)$ nonzero
entries they report roughly $O(N)$ rotations and $O(N\log N)$ CNOTs)
while preserving accuracy-vs-compression trade-offs~\cite{kuklinski2024s}.
Complementary work shows how to block-encode arithmetically structured
matrices directly from compact descriptions (e.g., Toeplitz, tridiagonal),
reducing data-loading overheads when structure is available~\cite{sunderhauf2024block}.
Chakraborty et. al.~\cite{chakraborty2018power} developed techniques
for block encoding dense matrices using QSVT, achieving query complexity
$O(\kappa\text{poly(\ensuremath{\log N/\epsilon})})$ for matrices
with condition number $\kappa$. Dong et al.~\cite{dong2021efficient}
formalized a general construction for block encoding an arbitrary
matrix, assuming the ability to implement multiplexers that control
rotations based on the matrix entries. Clader et al.~\cite{clader2022quantum}
systematically analyze the circuit resources needed to block-encode
dense classical matrices using QRAM, showing a minimal-depth construction
with $O(N^{2})$ qubits and T-depth $O(\log(N/\epsilon))$ versus
a minimal-count construction with $O(N\log(1/\epsilon))$ qubits and
T-depth $O(N)$. Additionally, they compare select-swap vs. bucket-brigade
QRAM constructions and introduce a state preparation routine improving
T-depth from $O(\log^{2}(N/\epsilon))$ to $O(\log(N/\epsilon))$
by utilizing ancillary qubits. As is evident from these complexities,
their low depth algorithms requires a large number ancillary qubits.

State preparation (aka amplitude encoding) is another well-established
method for embedding classical data vectors into quantum systems through
circuits that operate on the ground state and encode vectors as probability
amplitudes. Various state preparation techniques~\cite{mottonen2004transformation,shende2005synthesis,bergholm2005quantum,plesch2011quantum,iten2016quantum,malvetti2021quantum,gleinig2021efficient,araujo2023configurable}
have been developed over the years and designing new ones is an active
research topic. For a discussion of the state preparation literature,
including the derivation and analysis of quantum circuit costs, see
Section~\ref{subsec:Matrix2Msp_Matrix2PauliMsp}. In our recent work~\cite{yosef2024multivariate},
we extend the concept of state preparation from vectors to matrices
through matrix-state preparation, which encodes a matrix's vectorization
into a quantum state. This approach enables us to develop a limited
yet practical library of quantum circuits that implement fundamental
linear algebra primitives, and demonstrate its utility for computing
multivariate traces and spectral sums.

\section{\label{sec:preliminaries}Preliminaries}

This section presents the mathematical framework underlying our algorithms
and their analysis. Although several definitions build on our previous
work~\cite{yosef2024multivariate}, we extend them here to meet the
specific requirements of this paper.

\subsection{Linear algebra notation}

We denote scalars using Greek or Latin letters (e.g., $\alpha,\beta,\dots$
or $x,y,\dots$). Vectors are written in bold lowercase letters such
as $\mathbf{x},\mathbf{y},\dots$, matrices in bold uppercase, such
as $\mathbf{A},\mathbf{B},\dots$, and hypermatrices (higher dimensional
arrays)\footnote{In numerical analysis literature, hypermatrices are usually referred
to as ``tensors''. However, tensors are, in modern definitions,
elements of a tensor product of vector spaces or multilinear functionals~\cite{lim21tensors},
and this definition more aligned with the use of the term ``tensor''
in quantum physics. Thus, we opted for the term ``hypermatrix''
to describe a multi-dimensional array of numbers.} in calligraphic notation (e.g., $\mathcal{A},\mathcal{B},\dots$)\footnote{There is some ambiguity in the notation between hypermatrices, and
subsequent notation for quantum circuits (which actually describes
a tensor). We rely on context for distinguish between the two cases.}. The dimension, or order, of a hypermatrix $\tenA$ is denoted by
$\order{\tenA}$. We assume $0$-based indexing for vectors, matrices
and hypermatrices. This is less common in the numerical linear algebra
literature, but is more convenient in the context of quantum computing.
We use the convention that vectors are column vectors, unless otherwise
stated.

Vectorization of a hypermatrix ${\cal A}\in\C^{I_{0}\times\cdots\times I_{d-1}}$
is defined recursively by 
\[
\vec{\tenA}\coloneqq\begin{cases}
\tenA, & \order{\tenA}=1\\
\left[\begin{array}{c}
\vec{\tenA(1,:,\dots,:,:)}\\
\vec{\tenA(2,:,\dots,:,:)}\\
\vdots
\end{array}\right], & \text{otherwise}
\end{cases}
\]
Note that for a matrix (order-2), this corresponds to row-major vectorization.

The Frobenius norm of a hypermatrix ${\cal A}\in\C^{I_{0}\times\cdots\times I_{d-1}}$
is defined as 
\[
\FNorm{\tenA}\coloneqq\sqrt{\sum_{i_{0}=0}^{I_{0}-1}\sum_{i_{1}=0}^{I_{1}-1}\cdots\sum_{i_{d-1}=0}^{I_{d-1}-1}\abs{\tenA(i_{0},\dots,i_{d-1})}^{2}}
\]
The definition naturally applies to matrices and vectors (order-1
hypermatrices), where it coincides with the (matrix) Frobenius norm
and the 2-norm (respectively). We say that a hypermatrix $\tenA$
(respectively, matrix and vector) is normalized if $\FNorm{\mathcal{A}}=1$.
The inner product between two hypermatrices $\mathcal{A}$ and $\mathcal{B}$,
both of the same size, is given by $\langle\mathcal{A},\mathcal{B}\rangle\coloneqq\vec{\mathcal{A}}^{\H}\vec{\mathcal{B}}$
where the superscript $\H$ denotes the Hermitian transpose of vectors
and matrices. This is consistent with the Frobenius norm definition,
i.e., $\|\mathcal{A}\|_{F}=\sqrt{\langle\mathcal{A},\mathcal{A}\rangle}.$

The set of square complex matrices of size $M$, is a \emph{complex}
vector space of dimension $M^{2}$ . However, it is also real vector
space of dimension $2M^{2}$ . Inside that real vector space, we have
the subspace of Hermitian matrices, denoted by $\HH_{M}\subseteq\C^{M\times M}$,
which consists of all Hermitian matrices of size $M\times M$. The
dimension of $\HH_{M}$ (over $\R$) is $M^{2}$ (half the dimension
of $\C^{M\times M}$ when viewed as a real vector space, and the same
as the dimension when $\C^{M\times M}$ is viewed as a complex vector
space). On $\HH_{M}$, when viewed as a real vector space, we define
the following inner product, which we refer to $\HH_{M}$-inner product:
$\inner{\matA}{\matB}_{\HH_{M}}\coloneqq\Trace{\matA\matB}$. This
inner product induces the Frobenius norm.

We use capital letters to denote dimension sizes, e.g., $M,N,I_{1},I_{2},\dots$
Throughout the paper, whenever we encounter such dimensions we assume
they are power of two, with the corresponding small letter denoting
the base-2 logarithm, e.g., $n=\log_{2}N,m=\log_{2}M$, etc.

\subsection{Matrices, circuits and quantum states: Basic notations and definitions}

We use the standard state vector based formulation of quantum computing,
where the system's state is represented by a unit vector in an Hilbert
space. Without loss of generality, we assume that the Hilbert space
is $\C^{N}$, where $N=2^{n}$ for $n$ qubits. We denote the computational
basis in a $n$-qubit system as $\Ket 0_{n},\Ket 1_{n},\Ket 2_{n},\dots$
. For abstract states in a $n$-qubit system we use the notation $\Ket{\phi}_{n},\Ket{\psi}_{n},\dots$
. We denote by $\Ket{\psi}_{n}\Ket{\phi}_{m}$ the $n+m$ qubit state
obtained by tensoring the two states to obtain a state in $\C^{MN}$
where $N=2^{n}$ and $M=2^{m}$ with $\Ket{\phi}_{n}$ and $\Ket{\psi}_{m}$
being two abstract states on $n$ and $m$ qubits (respectively).
We assume that the MSB is in the lowest index in binary expansions,
e.g., $\Ket i_{n}=\Ket{b_{0}}_{1}\Ket{b_{1}}_{1}\cdots\Ket{b_{n-1}}_{1}\eqqcolon\Ket{b_{0}\cdots b_{n-1}}_{n}$
where $b_{0}\cdots b_{n-1}$ is the binary expansion of $i$; this
aligned with many physics textbooks.

We denote the $n$-qubit system's state whose amplitudes are given
by $\valpha$ after normalization by the ket $\Ket{\valpha}$ notation,
where $\valpha\in\C^{N}$ is the probability amplitude vector. That
is, $\Ket{\valpha}\coloneqq\frac{1}{\TNorm{\valpha}}\sum_{i=0}^{n-1}\alpha_{i}\Ket i_{n}$
where $\alpha_{i}$ is the $i$th entry of $\valpha$. Note that under
these conventions we have $\Ket i_{n}=\Ket{\e_{i}^{(N)}}$ where $\e_{0}^{(N)},\e_{1}^{(N)},\dots$
are the $N$-dimensional identity vectors. Given two amplitude vectors
$\valpha\in\C^{N}$ and $\vbeta\in\C^{M}$, we have $\Ket{\valpha}\Ket{\vbeta}=\Ket{\valpha\otimes\vbeta}$
and this extends naturally to higher number of multiplicands.

We use calligraphic letters to denote both quantum circuits and operators
on the state's Hilbert space (and, as explained in the previous subsection,
hypermatrices), e.g., ${\cal U}$, ${\cal S}$ and $\calT$, also
using the same letter for a circuit and the operator it induces. For
a circuit ${\cal U}$ on $n$ qubits, we use $\matM({\cal U})\in\C^{N\times N}$,
to denote the unique unitary matrix such that for every $\valpha\in\C^{N}$
applying ${\cal U}$ on the state $\Ket{\valpha}$ results in the
state $\Ket{\matM({\cal U})\valpha}$. We say that a circuit ${\cal U}^{\conj}$
is the inverse (or adjoint) of circuit ${\cal U}$ if $\matM({\cal U})=\matM({\cal U})^{\H}$.
Given a quantum system, we denote the application of circuits sequentially
with $\cdot$ or omit altogether. We say two circuits $\cirU$ and
$\cirV$ on the same number of qubits are \emph{equivalent}, denoted
by $\cirU\cong\cirV$, if they induce the same unitary matrix on the
Hilbert space (i.e., $\matM(\cirU)=\matM(\cirV)$).

A quantum register is defined as contiguous subset of qubits of a
multi-qubit system. The size of the register is determined by the
number of qubits it encompasses. Given a system with several quantum
registers, we use the standard tensor product $\otimes$ operator
to denote the application of each circuit on the corresponding register,
but also use $\cdotp_{i}$ to denote the action of a circuit ${\cal U}$
to only register $i$ of the system, e.g., ${\cal U}\cdotp_{i}\Ket{\valpha}$
(where the operation on the other register is the identity). We use
$0$-based indexing for register numbering within a system.

Given a hypermatrix $\tenA\in\C^{I_{0}\times\cdots\times I_{d-1}}$,
we denote by $\kket{\tenA}$ the state $\ket{\vec{\tenA}}$. This
state has $q=\sum_{j=0}^{d-1}i_{j}$ qubits, divided into $d$ registers
of size $i_{0},i_{1},\dots,i_{d-1}$. We also have

\[
\kket{{\cal A}}=\frac{1}{\|{\cal A}\|_{F}}\sum_{j_{0}=0}^{I_{0}-1}\sum_{j_{1}=0}^{I_{1}-1}\cdots\sum_{j_{k-1}=0}^{I_{d-1}-1}\tenA(j_{0},\dots,j_{d-1})\Ket{j_{0}\cdot I_{1}I_{2}\cdots I_{d-1}+j_{1}\cdot I_{2}I_{3}\cdots I_{d-1}+\cdots+j_{d-2}I_{d-1}+j_{d-1}}_{q}.
\]
For matrices, we have
\[
\kket{\matA}=\frac{1}{\|\matA\|_{F}}\sum_{j=0}^{N-1}\sum_{i=0}^{M-1}a_{i,j}\Ket{iN+j}_{n+m}
\]
where $a_{ij}$ denotes the $(i,j)$th entry of $\matA$. We use the
notation $\kket{{\cal A}}$ (instead of $\Ket{\vec{{\cal A}}}$) to
help the reader distinguish state descriptions that are based on hypermatrices
(or matrices) as opposed to vectors.

For a hypermatrix ${\cal A}$ we denote the number of qubits we need
to hold the state $\kket{{\cal A}}$ by $q(\tenA)$. For a circuit
${\cal U}$, we denote the number of gates in ${\cal U}$ by $g(\cirU)$,
and the depth (critical path) of ${\cal U}$ by $d(\cirU)$, and the
number of qubits of $\mathcal{U}$ by $q(\mathcal{U})$.

\subsection{Representing matrices using quantum circuits}

In quantum numerical linear algebra, matrices are represented - typically
as inputs - by quantum circuits. In the literature, the predominant
paradigm for such representations is block encoding. An alternative
approach uses matrix state preparation circuits. In this subsection,
we describe both methods.

A \emph{block encoding} of a matrix $\matA\in\C^{N\times N}$ is any
circuit for which in the corresponding unitary we have a scaled version
of $\matA$ in the top left part. The concept was formally introduced
in~\cite{gilyen2019quantum}, even though the idea can be traced
to earlier works on Linear Combination of Unitaries (LCU)~\cite{childs2011simulating,berry2015hamiltonian,low2019hamiltonian,childs2017quantum}.
Since its introduction, block encodings have been extensively used
in the literature on quantum numerical linear algebra and its applications~\cite{chakraborty2018power,martyn2021grand}.

A formal definition is as follows; the definition is, in one aspect,
slightly more general than the one originally proposed in~\cite{gilyen2019quantum}
in that it allows an arbitrary phase. However, it also slightly less
general as it does not allow an approximation error.
\begin{defn}
[Block Encoding Circuit] \label{def:BE}For $\alpha\geq0$, a circuit
$\cirU$ is a $\alpha$\emph{-Block Encoding} of $\matA\in\C^{N\times N}$,
denoted as $\cirU\in\BE{\matA}{\alpha}$, if there exists a $\theta\in[0,2\pi)$
such that
\[
\alpha e^{\si\theta}\CirMat{\cirU}=\left[\begin{array}{cc}
\matA & *\\*
* & *
\end{array}\right]
\]
We refer to $\alpha$ as the \emph{scale}. Note that the number of
qubits in $\cirU$ must be at least $n$. Surplus qubits are called
\emph{ancillary qubits, }and their number is equal to $q(\cirU)-n$.
We say that a block encoding is \emph{exact} if there is no ancillary
qubits. Note that this possible only if $\matA$ is scaled unitary
matrix.
\end{defn}

Suppose $\cirU\in\BE{\matA}{\alpha}$. Even if $\matA$ is Hermitian
this does not necessarily hold for $\CirMat{{\cal U}}$. However,
in some cases it might be desirable for the embedding matrix to be
Hermitian~\cite{dong2021efficient}. The following defines a more
restricted form of block encoding.
\begin{defn}
[Hermitian Block Encoding Circuit] For $\alpha\in\R$ , a circuit
$\cirU$ is a $\alpha$\emph{-Hermitian Block Encoding} of an Hermitian
matrix $\matA\in\C^{N\times N}$, denoted as $\cirU\in\HBE{\matA}{\alpha}$,
if 
\[
\alpha\CirMat{\cirU}=\left[\begin{array}{cc}
\matA & *\\*
* & *
\end{array}\right]
\]
and $\CirMat{\cirU}$ is Hermitian. We refer to $\left|\alpha\right|$
as the \emph{scale}.
\end{defn}

We denote block encoding circuits of a matrix $\matA$ by $\cirU_{\matA}^{\text{BE}}$,
and $\cirU_{\matA}^{\text{HBE}}$ for a Hermitian block encodings.
Let $a=q(\cirU_{\matA}^{\text{BE}})-n$. The definition of block encoding
implies that for any unit norm vector $\vpsi\in\C^{N}$, there exists
a $\theta$ such that 
\[
\cirU_{\matA}^{\text{BE}}\ket 0_{a}\ket{\vpsi}=\sqrt{p(\cirU_{\matA}^{\text{BE}},\vpsi)}e^{-\si\theta}\ket 0_{a}\ket{\matA\vpsi}=+\sqrt{1-p(\cirU_{\matA}^{\text{BE}},\vpsi)}\ket{\bot}_{q(\cirU)}
\]
where $\ket{\bot}$ satisfies $(\ket{0^{a}}\bra{0^{a}}\otimes\mat I_{N})\ket{\bot}_{q(\cirU_{\matA}^{\text{BE}})}=0$
and 
\begin{equation}
p(\cirU_{\matA}^{\text{BE}},\vpsi)=\frac{\TNormS{\matA\vpsi}}{\alpha^{2}}\label{eq:be-prob}
\end{equation}
We can now measure the ancillary qubits, and if we get $\ket 0_{a}$
we have successfully prepared $\ket{\matA\vpsi}$. The probability
of successfully preparing $\ket{\matA\vpsi}$ is $p(\cirU_{\matA}^{\text{BE}},\vpsi)\geq\sigma_{\min}(\matA)^{2}/\alpha^{2}$.
We see that it is desirable to for the scaling factor $\alpha$ to
be as small possible so to maximize the success probability. The only
lower bound shown in the literature for $\alpha$ is $\alpha\geq\TNorm{\matA}$.
Another trivial to prove bound is $\alpha\ge\frac{\FNorm{\matA}}{\sqrt{\text{M}}}$.
Both bounds are not informative. 

Many algorithms that make extensive use of block encodings assume
the input matrix is given block encodings form. Consequently, having
efficient block encoding circuits for input matrices is crucial for
realizing these algorithms' claimed quantum advantage. However, block
encoding is inherently restricted to matrices. Furthermore, in the
definition above, we defined block encodings only for square matrices.
Though we could have defined it to allow rectangular matrices, such
block encodings are less ubiquitous. Quantum numerical linear algebra
algorithms are mostly defined only for block encoding of square matrices.
For numerical arrays that are not square matrices, such as vectors,
rectangular matrices or higher order tensors, it is more convenient
to work with \emph{state preparation circuits}. For example, algorithms
for solving linear equations assume the right hand side is given as
a state preparation circuit~\cite{harrow2009quantum,childs2017quantum,childs2011simulating}.
\begin{defn}
[Hypermatrix State Preparation Circuit] \label{def:matrix-circ}Let
$I_{0},I_{1},\dots,I_{d-1}$ be powers of 2, and let ${\cal A}\in\mathbb{C}^{I_{0}\times I_{1}\times\cdots\times I_{d-1}}$
be a hypermatrix. For $\alpha\geq0$, a circuit $\cirU$ is a $\alpha$-\emph{Hypermatrix
State Preparation Circuit} of ${\cal A}$, denoted by $\cirU\in\TS{{\cal A}}{\alpha}$,
if there exists a $\theta\in[0,2\pi)$ such that
\[
\alpha e^{\si\theta}\CirMat{\cirU_{{\cal A}}}=\begin{bmatrix}\vdots & * & \cdots & *\\[1mm]
\vec{{\cal A}} & \vdots &  & \vdots\\[1mm]
\vdots & \vdots &  & \vdots\\[1mm]
\vpsi & \vdots &  & \vdots\\[1mm]
\vdots & * & \cdots & *
\end{bmatrix},
\]
for some (possibly empty) vector $\vpsi$. We refer to $\alpha$ as
the \emph{scale}.
\end{defn}

Suppose $\cirU\in\TS{{\cal A}}{\alpha}$. Since $\CirMat{\cirU}$
is, by definition, unitary, $\alpha$ and $\vpsi$ must be such that
first column of $\CirMat{\cirU}$ has unit norm. Thus, we must have
$\alpha\geq\FNorm{{\cal A}}$. While the requirement of $I_{0},I_{1},\dots,I_{d-1}$
to be powers of 2 is formally redundant, it enables a more streamlined
analysis of hypermatrix state preparation circuits. For example, in
the above the number of elements in $\vpsi$ is $2^{q(\cirU)}-I_{0}I_{1}\cdots I_{d-1}=(2^{q(\cirU)-\sum i_{j}}-1)I_{0}I_{1}\cdots I_{d-1}$\textcolor{red}{}.
In the notation $\cirU\in\TS{{\cal A}}{\alpha}$ we do not specify
the number of qubits in ${\cal U}$. It is property of the circuit
itself, not the fact that it is a hypermatrix state preparation circuit.
At the vary least, the number of qubits in ${\cal U}$ must be $\sum i_{j}$.
Thus, the number of \emph{ancillary qubits }in $\cirU\in\TS{{\cal A}}{\alpha}$
is $q(\cirU)-\sum i_{j}$. We denote hypermatrix state preparation
circuits of a hypermatrix $\tenA$ as $\cirU_{\tenA}^{\text{SP}}$.

Suppose that $\cirU_{\tenA}^{\text{SP}}\in\TS{{\cal A}}{\alpha}$.
The definition implies there exists a $\theta$ such that
\[
\cirU_{\tenA}^{\text{SP}}\ket 0_{q(\cirU)}=\sqrt{p(\cirU_{\tenA}^{\text{SP}})}e^{-\si\theta}\ket 0_{q(\cirU_{\tenA}^{\text{SP}})-\sum i_{j}}\kket{{\cal A}}+\sqrt{1-p(\cirU_{\tenA}^{\text{SP}})}\ket{\perp}
\]
where $\ket{\perp}$ satisfies the orthogonality condition $\bra{0^{q(\cirU_{\tenA}^{\text{SP}})-\sum i_{j}}}\otimes{\cal I}{}_{\sum i_{j}})\ket{\perp}=0$,
and $p(\cirU_{\tenA}^{\text{SP}})=\frac{\FNormS{{\cal A}}}{\alpha^{2}}$
is the probability of measuring $\ket 0_{q(\cirU_{\tenA}^{\text{SP}})-\sum i_{j}}$
when measuring only the last $q(\cirU_{\tenA}^{\text{SP}})-\sum i_{j}$
qubits. In other words, $p(\cirU_{\tenA}^{\text{SP}})$ is the probability
of actually succeeding in preparing $\kket{\tenA}$, up to global
phase, in the second register. Thus, it is desirable for $\alpha$
to be as small as possible. We must have $\alpha\geq\FNorm{\tenA}$.
We say that the state preparation circuit is \emph{exact }if $\alpha=\FNorm{\tenA}$
(equivalently, the probability of preparing $\kket{\tenA}$ in the
target register is $1$). A state preparation circuit can be exact
even when there are ancillary qubits, though if there are no ancillary
qubits the circuit must be exact.

Two special cases of hypermatrix state preparation circuits are worth
noting, and giving them a dedicated notation. When $k=1$, we are
dealing with (vector) state preparation circuits with garbage, a construction
extensively used in the literature~\cite{mottonen2004transformation,shende2005synthesis,plesch2011quantum}.
We use the notation $\cirU_{\v}^{\text{SP}}\in\SP{\v}{\alpha}$. When
$k=2$, i.e., $\matA\in\mathbb{C}^{M\times N}$, then $\vec{\matA}$
is the row major vectorization of a matrix. In this case, we denote
the corresponding \emph{Matrix State Preparation} of $\matA$ as $\cirU_{\matA}^{\text{SP}}\in\MS{\matA}{\alpha}$~\cite{yosef2024multivariate}.
See Section~\ref{subsec:Matrix2Msp_Matrix2PauliMsp} for related
work on state preparation circuit synthesis, including a discussion
of resource complexities.

The scale parameter $\alpha$ plays an important role in both block
encodings circuits and matrix state preparation circuits. As we saw,
it is connected to the probability of successfully preparing desired
states (signaled by successfully measuring certain states in the ancillary
qubits). Roughly speaking, when $\alpha$ is smaller, shallower circuits
suffice or less shots are required in downstream uses. The minimum
value is the amount we need to ``stretch'' or ``shrink'' a non-unitary
matrix $\matA$ so it fits within quantum unitary constraints. It
is possible to build an exact block encoding only when $\matA$ is
a scaled unitary matrix. In contrast, for \emph{any} $M\times N$
matrix it is possible to construct a $O(MN)$ depth exact matrix state
preparation circuit in $O(MN\log MN)$ (classical) time \cite{shende2005synthesis}.

\subsection{Permutations of qubits}

Let $S_{q}$ denote the set of permutations on $\{0,\dots,q-1\}$.
For a permutation $\sigma\in S_{q}$, let ${\cal S}_{\sigma}$ denote
the $q$-qubit quantum circuit that permutes the qubits according
to $\sigma$, i.e., qubit $i$ get re-labeled as qubit $\sigma(i)$.
Since we talking about permutation on qubits (and not on basis states),
$\calS_{\sigma}$ can be implemented using SWAP gates for any $\sigma\in S_{q}$.
Note that for the identity permutation $\text{id}_{q}$ we have $\calS_{\text{id}_{q}}={\cal I}_{q}$
where ${\cal I}_{q}$ is the identity circuit (empty circuit).

Given two permutations $\sigma_{1}\in S_{q}$ and $\sigma_{2}\in S_{k}$,
let $\sigma_{1}\oplus\sigma_{2}$ denote the concatenation (aka direct
sum) of the permutations:
\[
(\sigma_{1}\oplus\sigma_{2})(j)=\begin{cases}
\sigma_{1}(j) & 0\leq j<q\\
q+\sigma_{2}(j-q) & q\leq j<q+k
\end{cases}
\]
We have ${\cal S}_{\sigma_{1}\oplus\sigma_{2}}={\cal S}_{\sigma_{1}}\otimes{\cal S}_{\sigma_{2}}$.
The reason that the direct sum translates to a tensor product is because
the permutation is on the qubits, and not the states.

\subsection{Describing the complexity of quantum circuits}

The algorithms we describe in this paper receive classical inputs
and output classical outputs. The output is always a quantum circuit
that is to be executed on quantum hardware. Inputs can be classical
scalar, vectors and matrices or quantum circuits that describe such
objects. Note that the latter are also classical inputs: they are
classical descriptions of quantum circuits. When analyzing such algorithms
(classical-to-classical algorithms that produce quantum circuits),
we need to assess both the classical cost of the algorithm, and the
quality of the output circuits. The former is straightforward, while
the latter is intricate.

There are quite a few ways to understand the quality of a quantum
circuit. One clear metric is the number of qubits used; recall that
we denote this by $q(\cirU$). Two other measures is the overall number
of gates used (known as \emph{gate complexity}) and the length of
the critical path in the quantum circuits (known as \emph{circuit
depth}). We denote the depth and gate complexity of a circuit $\cirU$
(respectively). We are mostly interested in \emph{surplus complexity},
which are costs (depth, gate complexity, ...) in addition to input
circuits.

Both gate complexity and circuit depth depend on a prescribed \emph{gate
library}. Typically, quantum hardware allow only for a certain library
of elementary gates. Quantum computing frameworks decompose quantum
circuits submitted by the user to circuits that use only elementary
gates. This a process is called transpilation or compilation. The
actual circuit depth and gate complexity of the circuit submitted
to hardware (after transpilation) is the real metric of interest.
However, to keep our discussion simple, we make use of high level
gates in circuits description throughout the paper. Specifically,
we allow all single qubit gates, CNOT, fanout-CNOT.

Allowing all single qubit gates gives only a partial view of the overall
quality of the quantum circuit. In analysis of quantum circuits, one
common way to address this gap is to allow (in the gate library) all
Clifford (H, S, CNOT) gates + T gate, and count the complexity in
term of T gates (i.e., T-depth and T-count, denoted by $d_{T}(\cirU)$
and $g_{T}(\cirU)$). The underlying motivation is that in fault\nobreakdash-tolerant
architectures, Clifford operations are practically free, so the real
cost driver are T gates in the circuit, making T\nobreakdash-count
and T\nobreakdash-depth a good proxy for both space and time complexity
of executing a quantum circuit that uses only Clifford+T gates~\cite{amy2014polynomial,selinger2015generators}.
We also make use of this abstraction, making sure to specify the Clifford+T
complexity of single qubit gates we use (see following discussion
for key gates used). Again, we mostly interested in surplus complexity
compared to the input's T-depth and T-count.
\begin{rem}
\label{rem:depth-tdepth-paradox}Logically, the depth of the circuit
is an upper bound on the T-depth. However, we measure both quantities
using a different gate library: regular depth assumes the use of all
single qubit gates and CNOTs, while T-depth allows only Clifford+T.
Thus, it is possible for the T-depth, as we measure it in this paper,
to be larger the the depth (as we measure it in this paper).
\end{rem}

\paragraph{T-costs of rotation gates.}

The T-cost of rotation gates ${\cal R}_{y}(\theta)$ and ${\cal R}_{z}(\theta)$
depends on the angle $\theta$. In general, exact synthesis is only
possible for angles of the specific form $\theta=\nicefrac{k\pi}{2^{n}}$
(for integers $k,n$). For these angels the T-count scales as $O(n)$~\cite{kliuchnikov2012fast,giles2013exact,russell2014exact}.
Certain special angles have reduced T-costs: $\theta\in\{0,\pi/2,\pi\}$
are Clifford gates (T-cost = 0), and $\theta=\pi/4$ is a T-gate (T-cost
= 1). For arbitrary angles, both ${\cal R}_{y}(\theta)$ and ${\cal R}_{z}(\theta)$
must be approximated using Clifford+T decompositions, with T-costs
scaling as $O(\log(1/\epsilon))$ for precision $\epsilon$~\cite{dawson2005solovay,RossSelinger2016Journal}.
For state preparation circuits, the size of the angles depend on the
number of prepared amplitudes and this affects the T-depth of the
rotations. In general, when preparing $N$ amplitudes, you need $O(\log(N/\epsilon))$
T-gates for each rotation, where $\epsilon$ is an accuracy parameter~\cite{low2024trading}.
However, we make the assumption that $\epsilon$ is a fixed small
number, which allows us to ignore the $\log(1/\epsilon)$ term.

\paragraph{Diagonal Gate.}

A diagonal gate ${\cal D}$ on a $m$\nobreakdash-qubit register is
a circuit that applies independent phase factors to each computational
basis state, i.e., 
\[
\matM({\cal D})=\diag{e^{\si\phi_{0}},\dots,e^{\si\phi_{M-1}}}
\]
for some (input) $\phi_{0},\phi_{1},\dots,\phi_{M-1}\in[0,2\pi)$.
Note that while it is typically referred to as a ``gate'', it is
actually a circuit family, and several implementations of it have
been proposed in the literature. In our constructions, we write it
as a gate, even though we actually mean that any implementation can
be used. In complexity analysis, consistent with state-of-the-art
results, we assume that gate and depth complexities of a $m$-qubit
diagonal gate are $O(M/m)$, and the T-costs are $d_{T}({\cal D})=O(M)$
and $g_{T}({\cal D})=O(Mm)$~\cite{sun2023asymptotically,low2024trading}.
We note that structured set of phases, such as multiplexed ${\cal R}_{z}$
rotations for certain rotations, admit significantly lower T\nobreakdash-cost.
\begin{rem}
In reality, a diagonal gate cannot be implemented exactly for arbitrary
phases (though there are phases for which it can be implemented exactly),
and gate constructions implement the operation up to a prescribed
approximation error $\epsilon$. The T-complexities of state-of-the-art
diagonal gate constructions are actually $g_{T}({\cal D})=O(M\log(M/\varepsilon))$,
$d_{T}(D)=O(\frac{M}{m}\log(M/\varepsilon))$. However, we again make
the assumption that $\epsilon$ is a fixed small number, which allows
us to ignore the $\log(1/\epsilon)$ term. For simplicity, we generally
describe the operation of the circuit as if diagonal gates are implemented
exactly.
\end{rem}

\paragraph{Eliminate Permutations.}

In our previous work~\cite[Appendix A.4]{yosef2024multivariate},
we describe a method, called \noun{EliminatePermutations}, to eliminate
trailing and leading SWAP gates in hypermatrix state preparation circuits.
The \noun{EliminatePermutations }algorithm receives a state preparation
circuit ${\cal U}_{\tenA}^{\text{SP}}$ for some hypermatrix $\tenA$,
and assumes the circuit can be written as $\mathcal{U}_{\tenA}^{\text{SP}}={\cal S}_{\sigma_{1}}\cdot{\cal Q}\cdot{\cal S}_{\sigma_{2}}$
for some (possibly empty) circuit ${\cal S}_{\sigma_{1}}$ and ${\cal S}_{\sigma_{2}}$
that are composed only of SWAP gates, so implement qubit permutations.
It outputs a circuit ${\cal V}_{\tenA}^{\text{SP}}$ that also implements
a hypermatrix state preparation circuit for $\tenA$, without leading
and trailing SWAPs. Note that typically $\matM({\cal U}_{\tenA}^{\text{SP}})\neq\matM({\cal V}_{\tenA}^{\text{SP}})$,
however, since both circuits implement the same hypermatrix state
preparation, in many cases ${\cal V}_{\tenA}^{\text{SP}}$ can replace
$\cirU_{\tenA}^{\text{SP}}$ in more elaborate circuits that use $\cirU_{\tenA}^{\text{SP}}$
as a sub-circuit. The classical cost of the algorithm is $O(g({\cal U}_{\tenA}^{\text{SP}}))$.

\section{\label{subsec:rotation-multiplexers}Multiplexers and rotation multiplexers
(aka uniformly controlled rotations)}

Quantum multiplexers (also known as \emph{uniformly controlled circuits})
are central to the various algorithms we propose in this paper. In
this section, we discuss both general construction of mulitplexers,
and ones that are specific to rotation mulitplexers.

Quantum multiplexers are fundamental building blocks in quantum circuits~\cite{bullock2008asymptotically,vartiainen2004efficient,mottonen2004quantum}.
They enable \emph{if-then-else} conditional operations on target qubits
based on the state of select (control) qubits. In circuit diagrams
and figures, quantum multiplexers are typically drawn with two components:
an empty rectangular box on the control lines representing all the
control states, and a target gate on the target qubits representing
the family of circuits. Figure~\ref{fig:mux} (a) illustrates the
graphical notation commonly used to denote multiplexing a parameterized
family of circuits $\{{\cal U}_{\theta}\}$. Mathematically, a multiplexer
${\cal U}$ is usually described via the equation 
\[
{\cal U}=\sum_{i=0}^{M-1}\Ket i\Bra i\otimes{\cal U}_{i}
\]
meaning that the operation ``selects'' which circuit ($\cirU_{0},\dots,\cirU_{M-1}$)
to apply based on the state of the $m=\log_{2}M$ qubits in the control
register. Thus, in the literature on Linear Combinations of Unitaries
(LCU), the operation is referred to as SELECT operation, and is an
important building block in implementing LCUs.

A multiplexer is formally defined as a circuit whose associated unitary
matrix is the direct sum of the unitaries corresponding to the multiplexed
circuits.
\begin{defn}
[Direct sum between matrices] Given $M$ matrices $\matA_{0},\dots,\matA_{M-1}$,
their \emph{direct sum}, denoted by $\bigoplus_{i=0}^{M-1}\matA_{i}$,
is the block diagonal matrix obtained by putting $\matA_{0},\dots,\matA_{M-1}$
on the diagonal: 
\[
\bigoplus_{i=0}^{M-1}\matA_{i}\coloneqq\blkdiag{\matA_{0},\dots,\matA_{M-1}}=\left[\begin{array}{cccc}
\matA_{0}\\
 & \matA_{1}\\
 &  & \ddots\\
 &  &  & \matA_{M-1}
\end{array}\right]
\]
\end{defn}

\begin{defn}
[Multiplexer circuit] Suppose that $\cirU_{0},\dots,\cirU_{M-1}$
are circuits that on the same number of qubits. A circuit that implements
a \emph{multiplexer between} $\cirU_{0},\dots,\cirU_{M-1}$, denoted
by $\cirU\in{\cal MX}({\cal U}_{0},\cirU_{1},\dots,\cirU_{M-1})$,
is any circuit $\cirU$ for which 
\[
\matM\left(\cirU\right)=\bigoplus_{i=0}^{M-1}\matM(\cirU_{i})
\]
Note that $q(\cirU)=n+q(\cirU_{0})$.
\end{defn}

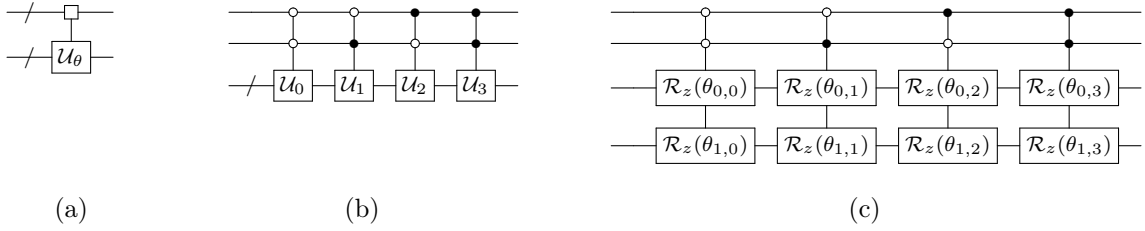
\begin{figure}[t]
\centering{}%
\begin{tabular}{ccccc}
{\footnotesize\Qcircuit @C=1em @R=1em {
 \lstick{} &{/} \qw    & \gate{} 	  		   & \qw 		&  \\
 \lstick{} &{/} \qw &\gate{{\cal U}_{\theta}} \qwx & \qw & 							
 } 
} & {\footnotesize ~} & {\footnotesize\Qcircuit @C=1em @R=1em {
 \lstick{}  & 			& \qw    & \ctrlo{1} & \ctrlo{1} & \ctrl{1} & \ctrl{1} & \qw \\
 \lstick{}  & 			&\qw     & \ctrlo{1} & \ctrl1{1} & \ctrlo{1} & \ctrl{1} & \qw \\
 \lstick{}  & &{/} \qw & \gate{{\cal{U}}_0} & \gate{{\cal{U}}_1} & \gate{{\cal{U}}_2} & \gate{{\cal{U}}_3} & \qw 
 } } & {\footnotesize ~} & {\footnotesize\Qcircuit @C=1em @R=1em {
 \lstick{}  & & \qw & \ctrlo{1} & \ctrlo{1} & \ctrl{1} & \ctrl{1} & \qw \\
 \lstick{}  & & \qw & \ctrlo{1} & \ctrl{1} & \ctrlo{1} & \ctrl{1} & \qw \\
 \lstick{}  & & \qw & \gate{{\cal R}_{z}(\theta_{0,0})} & \gate{{\cal R}_{z}(\theta_{0,1})} & \gate{{\cal R}_{z}(\theta_{0,2})} & \gate{{\cal R}_{z}(\theta_{0,3})} & \qw \\
 \lstick{}  & & \qw & \gate{{\cal R}_{z}(\theta_{1,0})} \qwx & \gate{{\cal R}_{z}(\theta_{1,1})} \qwx & \gate{{\cal R}_{z}(\theta_{1,2})} \qwx & \gate{{\cal R}_{z}(\theta_{1,3})} \qwx & \qw 
}}\tabularnewline
~ &  & ~ &  & ~\tabularnewline
(a) &  & (b) &  & (c)\tabularnewline
\end{tabular}\caption{\label{fig:mux}Quantum multiplexer constructions. (a) Notation for
quantum multiplexer for a family of circuits $\{\protect\cirU_{\theta}\}$.
(b) Generic construction of multiplexer between $\protect\cirU_{0},\protect\cirU_{1},\protect\cirU_{2},\protect\cirU_{3}$.
(c) Generic construction of 2D ${\cal R}_{z}$ multiplexer with angles
in a matrix $\protect\mat{\Theta}$.}
\end{figure}

Figure~\ref{fig:mux} (b) illustrates a generic (and naive) implementation
of a quantum multiplexer that works for aribtrary input circuits,
but uses, as its building block, multi-qubit control versions of input
circuits (the example shown has $M=4$, and thus utilizes two control
qubits).

\subsection{\label{subsec:1d-2d-rotation-multiplexers}Construction of rotation
multiplexers}

Multi-qubit controls are costly, thus the generic implementation is
inefficient. Since in many cases much of the complexity in a quantum
circuit comes from a multiplexer subcircuit, more efficient multiplexers
have been proposed for certain cases. One such case is multiplexers
of rotation gates.
\begin{defn}
Let $M$ be a power of two, and $K\geq1$. Given angles $\theta_{0,0},\dots,\theta_{(K-1),(M-1)}\in[0,2\pi)$,
a $K$-by-$M$ ${\cal R}_{y}$ 2D-multiplexer of these angles is any
circuit $\cirU$ such that
\begin{align*}
\cirU\in{\cal MX}\big( & {\cal R}_{y}(\theta_{0,0})\otimes\cdots\otimes{\cal R}_{y}(\theta_{K-1,0}),\\
 & {\cal R}_{y}(\theta_{0,1})\otimes\cdots\otimes{\cal R}_{y}(\theta_{K-1,1}),\\
 & \quad\quad\quad\quad\quad\quad\quad\vdots\\
 & {\cal R}_{y}(\theta_{0,M-1})\otimes\cdots\otimes{\cal R}_{y}(\theta_{K-1,M-1})\big)
\end{align*}
When $K=1$ we say this is a 1D multiplexer. We also introduce the
notation ${\cal \cirU\in MX}_{y}(\mat{\Theta})$, where $\mat{\Theta}\in\R^{K\times M}$
is the matrix with angles given by $\theta_{0,0},\dots,\theta_{K-1,M-1}$.
The definition of ${\cal R}_{z}$ 1D and 2D multiplexer simply replaces
${\cal R}_{y}$ gates with ${\cal R}_{z}$ gates.
\end{defn}

Figure~\ref{fig:mux} (c) illustrates the construction of a 2-by-4
${\cal R}_{z}$ 2D multiplexer via its naive implementation. Mottonen
et al. showed how to efficiently implement ${\cal R}_{y}$ (respectively
${\cal R}_{z}$) 1D multiplexers via only ${\cal R}_{y}$ (resppectively
${\cal R}_{z}$) gates and CNOT gates~\cite{mottonen2004transformation}.
The technique is illustrated in Figure \ref{fig:recursive-multiplexed}.
Since then, the construction has been repeated as basic building block
in many works~\cite{shende2004smaller,camps2024explicit,shende2005synthesis,camps2022fable}.

We now give more explicit details on the construction of Mottonen
et al. These details are required to for some of our proofs. First
we consider the 1D-multiplexer case ($K=1$). The construction for
${\cal R}_{y}$ and ${\cal R}_{y}$ is essentially the same, with
the difference being the gate used at the bottom of the recursion.
Thus, in the following we use $\xi$ as a generic placeholder for
either $y$ or $z$.

Given a row vector of angles $\vtheta$, the goal is to build circuits
$1d{\cal MXR}_{\xi}(\vtheta)\in{\cal MX}_{\xi}(\vtheta)$. In $1d{\cal MXR}_{\xi}(\vtheta)$
there is one target qubit $t$ and $m$ control qubits labeled $c_{0},\dots,c_{m-1}$.
We first compute a new row vector of angles $\hat{\vtheta}=\frac{1}{\sqrt{M}}\vtheta\mat H_{M}\matG_{M}$
where $\mat H_{M}$ is the $M$-by-$M$ Walsh-Hadamard matrix, and
and $\matG_{M}$ is the permutation matrix that reorders from binary
order to the binary-reflected Gray code (BRGC) order (i.e., index
$j$ moves to index $j\oplus(j\gg1)$ where $j\gg1$ denotes shift
of bits right). From there, the construction is 
\[
1d{\cal MXR}_{\xi}(\vtheta)={\cal F}_{m}^{(0)}(\hat{\vtheta})\cdot\text{CX}(c_{0},t)
\]
where ${\cal F}_{m}^{(0)}(\hat{\vtheta})$ is specified recursively
(with superscript denoting depth of recursion): 
\begin{equation}
{\cal F}_{m}^{(d)}(\vphi)\coloneqq{\cal F}_{m}^{(d+1)}(\vphi_{0:(2^{m-d-1}-1)})\cdot\text{CX}(c_{d},t)\cdot{\cal F}_{m}^{(d+1)}(\vphi_{2^{m-d-1}:(2^{m-d}-1)})\label{eq:recursive-f_m}
\end{equation}
The leaf case is:
\begin{equation}
{\cal F}_{m}^{(m-1)}(\alpha,\beta)\coloneqq{\cal R}_{\xi,t}(\alpha)\cdot\text{CX}(c_{m-1},t)\cdot{\cal R}_{\xi,t}(\beta)\label{eq:base-f_m}
\end{equation}
In the above, $\text{CX}(c,t)$ denotes a CNOT gate with control $c$
and target $t$, and ${\cal R}_{\xi,t}(\gamma)$ denotes adding a
single ${\cal R}_{\xi}$ rotation by angle $\gamma$ on target $t$.
For simplicity, we allow omitting the rotation label when clear from
context, since the circuit construction is identical for both ${\cal R}_{z}$
and ${\cal R}_{y}$ multiplexers, differing only in the leaf rotations.

The technique can be easily extended to 2D rotation multiplexers,
repeating the 1D multiplexer on ``rows'' of the circuit, using the
same control qubits. The CNOTs become fanout-CNOTS\footnote{A sequence of CNOT gates where the same control qubit is used to target
multiple different qubits (sequentially with depth 1).}. See Figure~\ref{fig:2d-multiplexed} for an illustration. In the
following, we denote this construction by $2d{\cal MXR}_{y}(\mat{\Theta})$
and by $2d{\cal MXR}_{z}(\mat{\Theta})$ where $\mat{\Theta}\in\R^{K\times M}$
is a matrix of rotation angles. A key observation is that when implementing
$2d{\cal MXR}_{y}(\mat{\Theta})$ (respectively $2d{\cal MXR}_{z}(\mat{\Theta})$)
we set up an array of ${\cal R}_{y}$ (respectively, ${\cal R}_{z}$)
gates with the same layout as $\mat{\Theta}$, with CNOT gates in
between. The following fact relates the rotation angles to the angles
in the gates of the multiplexers.
\begin{fact}
[\cite{mottonen2004quantum}] \label{fact:GN-fact}The angles in the
rotation gates of $1d{\cal MXR}_{z}(\vtheta)$ and $1d{\cal MXR}_{y}(\vtheta)$
are given by $\hat{\mat{\Theta}}=\frac{1}{\sqrt{M}}\mat{\Theta}\mat H_{M}\matG_{M}$
where $\mat H_{M}$ is the $M$-by-$M$ Walsh-Hadamard matrix, and
and $\matG_{M}$ is the permutation matrix that reorders from binary
order to the binary-reflected Gray code (BRGC) order (i.e., index
$j$ moves to index $j\oplus(j\gg1)$ where $j\gg1$ denotes shift
of bits right).
\end{fact}

The classical cost of computing the $\hat{\mat{\Theta}}$ using (fast)
Welsh-Hadamard transform is $O(KMm)$. The resulting circuit has $KM$
rotations angles translated to gates and $M$ fanout-CNOTs gate (between
each columns). The depth, when counting fanout-CNOTs and single control
rotation gates as depth 1, is $2M$. For T + Clifford analysis, the
T-gate count is $O(KM\log(M/\epsilon))$ (each of the $KM$ arbitrary
rotations requires $O(\log(M/\epsilon))$ T-gates for approximation,
where $\epsilon$ represents the total approximation error when converting
arbitrary rotation angles into discrete T-gates) and the T-depth is
$O(M\log(M/\epsilon))$ (the overall $2M$ depth is multiplied by
the $O(\log(M/\epsilon))$ T-depth of each approximated rotation layer)
for an approximation error $\epsilon$ of the arbitrary angles.

The discussion above assumed that the angles $\hat{\mat{\Theta}}$
that form the array of rotation gates is dense and arbitrary. However,
it might be the case that $\hat{\mat{\Theta}}$ is sparse even if
$\mat{\Theta}$ is dense, and vice versa. The quantum complexity of
the resulting multiplexer circuit depends on the density of $\hat{\mat{\Theta}}$,
though the classical cost of forming $\hat{\mat{\Theta}}$ is unchanged
(so the classical cost of creating the multiplexer remains the same).
In Section~\ref{sec:multiplexing-pauli-coefficient}, we discuss
a case in which $\hat{\mat{\Theta}}$ is not only highly sparse, but
also structured with unique entries. We show that in that case, we
can build an ad-hoc construction for this multiplexer that is highly
efficient, reducing both quantum and classical costs.

\begin{figure}[t]
\begin{centering}
{\footnotesize\begin{tabular}{c c c c c c} 
\Qcircuit @C=0.8em @R=1em { & \qw & \gate{} & \qw \\ & \qw & \gate{} \qwx & \qw \\ & \qw & \gate{{\cal R}_z} \qwx & \qw } 
& 
\raisebox{-2em}{$\cong$}     
& 
\Qcircuit @C=0.8em @R=1em {
& \qw & \qw 			& \ctrl{2} 	& \qw 			  & \ctrl{2}    & \qw \\
& \qw & \gate{}  		& \qw  		& \gate{}         & \qw			& \qw \\
& \qw & \gate{{\cal R}_z} \qwx & \targ 	& \gate{{\cal R}_z} \qwx & \targ 		& \qw 
}  
&
\raisebox{-2em}{$\cong$}     
& 
 \Qcircuit @C=0.8em @R=1em { 
& \qw & \qw           & \qw    		& \qw           & \ctrl{2}    & \qw & \qw           & \qw    		& \qw      & \ctrl{2} & \qw     \\ 
& \qw & \qw           & \ctrl{1}    & \qw           & \qw         & \qw & \qw           & \ctrl{1}    & \qw        & \qw      & \qw  \\ 
& \qw & \gate{{\cal R}_{z}(\hat{\theta}_{0})}    & \targ       & \gate{{\cal R}_{z}(\hat{\theta}_{1})}    & \targ       & \qw & \gate{{\cal R}_{z}(\hat{\theta}_{2})}    & \targ       & \gate{{\cal R}_{z}(\hat{\theta}_{3})} & \targ    & \qw   
} 
\end{tabular}}{\footnotesize\par}
\par\end{centering}
\caption{\label{fig:recursive-multiplexed}Illustration of the recursive construction
of \cite{mottonen2004transformation} for 1D ${\cal R}_{z}$-rotation
multiplexer of angles $\theta_{0},\theta_{1},\theta_{2},\theta_{3}$.
The angles $\hat{\theta}_{0},\hat{\theta}_{1},\hat{\theta}_{2},\hat{\theta}_{3}$
are computed as linear combination of $\theta_{0},\theta_{1},\theta_{2},\theta_{3}$.}
\end{figure}
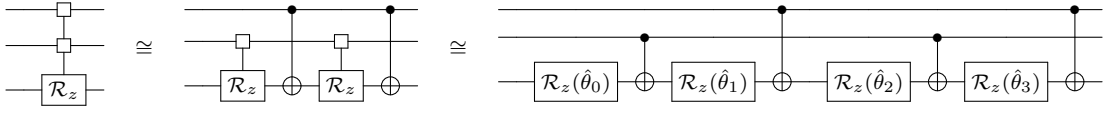

\begin{figure}[t]
\begin{centering}
{\footnotesize\begin{tabular}{c c c c c c} 
\Qcircuit @C=0.8em @R=1em { 
& \qw & \gate{} & \qw \\
& \qw & \gate{} \qwx & \qw \\
& \qw & \gate{{\cal R}_z} \qwx & \qw \\
&  & \raisebox{0.7em} \vdots  &  \\
& \qw & \gate{{\cal R}_z}  & \qw 
} 
& 
\raisebox{-2em}{$\cong$}     
& 
\Qcircuit @C=0.5em @R=1em {
& \qw & \qw 			& \ctrl{2} &\qw     & \ctrl{4}    & \qw 			  & \ctrl{2}&\qw    & \ctrl{4}    & \qw \\
& \qw & \gate{}  		& \qw  	   &\qw     & \qw  		  & \gate{}           & \qw		&\qw    & \qw  			& \qw \\
& \qw & \gate{{\cal R}_z} \qwx & \targ    &\qw     & \qw 	      & \gate{{\cal R}_z} \qwx   & \targ 	&\qw    & \qw 	   		& \qw  \\
&     & \raisebox{0.7em} \vdots  		&  		   &\ddots  &  	  		  & \raisebox{0.7em} \vdots  		  &  	 	&\ddots &  			    &   \\																	 			
& \qw & \gate{{\cal R}_z}      & \qw 	   &\qw     & \targ 	  & \gate{{\cal R}_z}   	  & \qw 	&\qw    & \targ 		& \qw  \\
}  
&
\raisebox{-2em}{$\cong$}     
& 
\scalebox{0.7}{
\Qcircuit @C=0.5em @R=1em {
& \qw & \qw 			        &\qw  		&\qw    & \qw      & \qw 			         & \ctrl{2}  &\qw     & \ctrl{4}  & \qw 			        &\qw  		&\qw    & \qw      & \qw 			           & \ctrl{2} &\qw    & \ctrl{4}  & \qw \\
& \qw & \qw		       			&\ctrl{1}  &\qw     & \ctrl{3} & \qw  		             & \qw  	 &\qw     & \qw  	  & \qw		       			&\ctrl{1}  &\qw     & \ctrl{3} & \qw  		          	  & \qw		&\qw    & \qw  		& \qw \\
& \qw & \gate{{\cal R}_{z}(\hat{\theta}_{0,0})}          	&\targ     &\qw     & \qw 	   & \gate{{\cal R}_{z}(\hat{\theta}_{0,1})}            	 & \targ     &\qw     & \qw 	  & \gate{{\cal R}_{z}(\hat{\theta}_{0,2})}          	&\targ     &\qw     & \qw 	   & \gate{{\cal R}_{z}(\hat{\theta}_{0,3})}            	   & \targ 	&\qw    & \qw 	   	& \qw  \\
&     & \raisebox{0.7em} \vdots & 		    &\ddots &  	  	   & \raisebox{0.7em} \vdots &  		 &\ddots  &  	  	  & \raisebox{0.7em} \vdots & 		    &\ddots &  	  	   & \raisebox{0.7em} \vdots   &  	 	&\ddots &  			&   \\																	 			
& \qw & \gate{{\cal R}_{z}(\hat{\theta}_{K-1,0})}     			&\qw 	    &\qw    & \targ    & \gate{{\cal R}_{z}(\hat{\theta}_{K-1,1})}     		 & \qw 	     &\qw     & \targ 	  & \gate{{\cal R}_{z}(\hat{\theta}_{K-1,2})}     			&\qw 	    &\qw    & \targ    & \gate{{\cal R}_{z}(\hat{\theta}_{K-1,3})}     		   & \qw 		&\qw    & \targ 	& \qw  \\
}

}

\end{tabular}}{\footnotesize\par}
\par\end{centering}
\caption{\label{fig:2d-multiplexed} Illustration of the recursive construction
of 2D ${\cal R}_{z}$rotation multiplexer of angles $\protect\mat{\Theta}\in\protect\R^{K\times4}$.
The angles $\hat{\protect\mat{\Theta}}$ are computed by $\hat{\protect\mat{\Theta}}=\frac{1}{2}\protect\mat{\Theta}\protect\mat H_{4}\protect\matG_{4}$.}
\end{figure}
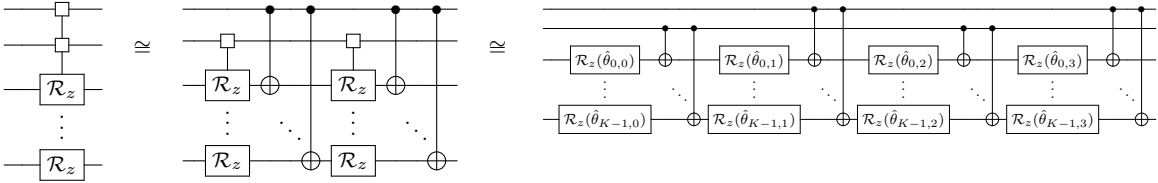

\section{\label{sec:Converting-between-standard}Converting between standard
and Pauli basis}

The Pauli decomposition of matrices plays a fundamental role in quantum
computing and quantum information theory, providing an efficient means
of representing and manipulating quantum operators. We make extensive
use of this decomposition in our constructions. In this section, we
discuss how to convert between the standard and Pauli basis representations
of a matrix, both classically and quantumly. While the classical case
has been addressed in the literature, the quantum case has not.

\subsection{Pauli matrices}

The Pauli matrices, along with the identity matrix, form a orthogonal
basis for $2\times2$ matrices.
\begin{defn}
[Pauli matrices] The Pauli matrices are: 
\begin{align*}
\sigma_{0}\coloneqq\sigma_{I}=\matI & =\left[\begin{array}{cc}
1 & 0\\
0 & 1
\end{array}\right]\\
\sigma_{1}\coloneqq\sigma_{X}=\matX & =\left[\begin{array}{cc}
0 & 1\\
1 & 0
\end{array}\right]\\
\sigma_{2}\coloneqq\sigma_{Y}=\matY & =\left[\begin{array}{cc}
0 & -\si\\
\si & 0
\end{array}\right]\\
\sigma_{3}\coloneqq\sigma_{Z}=\matZ & =\left[\begin{array}{cc}
1 & 0\\
0 & -1
\end{array}\right]
\end{align*}
\end{defn}

\begin{fact}
The following hold:
\begin{itemize}
\item \textbf{Orthogonality with respect to $\inner{\cdot}{\cdot}_{\HH_{2}}$:}
$\Trace{\sigma_{i}\sigma_{j}}=2\delta_{ij}$ where $\delta_{ij}$
is the Kronecker delta.
\item \textbf{Anti-commutation: $\sigma_{i}\sigma_{j}=-\sigma_{j}\sigma_{i}$}
for $i\ne j,\,i,j\in\{1,2,3\}$.
\item \textbf{Multiplication:} $\sigma_{i}^{2}=\matI$ for all $i$.
\item \textbf{Basis Completeness:} Any $2\times2$ Hermitian matrix $\matA$
can be uniquely expressed as $\matA=\sum_{i=0}^{3}\alpha_{i}\sigma_{i},$
where $\alpha_{i}\in\R$, and any $2\times2$ complex matrix $\matB$
can be uniquely expressed as $\matB=\sum_{i=0}^{3}\beta_{i}\sigma_{i}$
where $\beta_{i}\in\C$.
\end{itemize}
\end{fact}

Going to higher order, we have the following.
\begin{defn}
[Higher-order Pauli matrices] Let $\Sigma=\{I,X,Y,Z\}$ and $N=2^{n}$.
For a \emph{word} (i.e., sequence) $\w=(w_{0},w_{2},\dots,w_{n-1})\in\Sigma^{n}$,
the \emph{n-wise Pauli matrix} corresponding to $\w$ is defined as:
\[
\sigma_{\w}\coloneqq\sigma_{w_{0}}\otimes\sigma_{w_{1}}\otimes\cdots\otimes\sigma_{w_{n-1}}\in\mathbb{H}_{N}
\]

In quantum computing and quantum information, such matrices are typically
referred to \emph{multi-qubit Pauli matrices}. However, their description
is completely classical, so we opted not to use the confusing term
``qubit''. It is also more notation efficient and more convenient
to write $\w$ as string $\w=w_{0}w_{1}\cdots w_{n-1}$, and use juxtaposition
to denote word concatenation, i.e., $\w_{1}\w_{2}\coloneqq w_{11}\cdots w_{1,k}w_{21}\cdots w_{2,m}$.
In term of indices, it is useful to identify the symbols $I,X,Y,Z$
as $I=0,\,X=1,\,Y=2,\,Z=3$, so, for example, we view the word $XYZIX$
as synonymous to $(1,2,3,0,1)$.

Assume $N=2^{n}$. Consider the set $\PP_{N}\coloneqq\{\sigma_{\w}\,:\,\w\in\Sigma^{n}\}\subseteq\HH_{N}$.
Clearly, there are $|\Sigma|^{n}=N^{2}$ such matrices. Since they
are the Kronecker product of Hermitian matrices, they are also Hermitian.
Due to the mixed-product property of the Kronecker, and the fact that
trace of Kronecker product is the product of traces, for each $\w,\s\in\Sigma^{q}$
we have
\[
\inner{\sigma_{\w}}{\sigma_{s}}_{\HH_{N}}=\begin{cases}
N & \w=\s\\
0 & \w\neq\s
\end{cases}
\]
Recall that the dimension of $\HH_{N}$ as a real vector space is
$N^{2}$, so the set $\PP_{N}$ forms an orthogonal basis for $\HH_{N}$
over reals. They are also an orthogonal basis for $\C^{N\times N}$
over $\C$ with respect to the Hilbert-Schmidt inner product.\emph{}
\end{defn}

\subsection{\label{subsec:Pauli-decomposition-classiq}Pauli decomposition and
its (classical) computation}

Higher-order Pauli matrices form an orthogonal basis of $\C^{N\times N}$
over $\C$, so every matrix $\matA\in\C^{N\times N}$ can be uniquely
expressed as a linear combination of higher-order Pauli matrices:

\[
\matA=\sum_{\w\in\Sigma^{n}}\alpha_{\matA}(\w)\sigma_{\w}
\]
where $\{\alpha_{\matA}(\w)\}_{\w\in\Sigma^{n}}\subseteq\C$ are the
$N^{2}$ \emph{Pauli coefficients} of $\matA$. When $\matA$ is Hermitian,
we have $\alpha_{\matA}(\w)\in\R$ for all $\w\in\Sigma^{n}$. For
every $\w$, view $\matA\mapsto\alpha_{\matA}(\w)$ as function from
$\C^{N\times N}$ (respectively $\HH_{N}$) to $\C$ (respectively
$\R$). Alternatively, for any matrix $\matA$ we can view $\w\mapsto\alpha_{\matA}(\w)$
as a function from $\Sigma^{n}$ to $\C$ (or $\R$ for Hermitian
$\matA$). Since we can identify each word with an index between $0$
and $N^{2}-1$ (using lexigraphical ordering), we can view $\alpha_{\matA}$
as a vector of length $N^{2}$.

We can also arrange the Pauli coefficients in the form of a hypermatrix,
using the identification $I=0,\,X=1,\,Y=2,\,Z=3$. We denote this
hypermatrix by $\tenA_{P}$. When $\matA$ is $N\times N$, the order
of $\tenA_{P}$ is $n$, and its dimension is $4\times4\times\cdots\times4$.
Since the higher order Pauli matrices form an orthogonal basis, and
the norm of each basis matrix is $\sqrt{N}$, we have $\FNormS{\matA}=N\FNormS{\tenA_{P}}$.

\begin{example}
Let's consider a simple case of a $4\times4$ Hermitian matrix. Its
Pauli decomposition involves $16$ terms. The hypermatrix, here a
matrix, $\mathbf{A}_{P}$ is $4\times4$, where each element corresponds
to a specific Pauli coefficient:
\[
\matA_{P}=\left[\begin{array}{cccc}
\alpha_{II} & \alpha_{IX} & \alpha_{IY} & \alpha_{IZ}\\
\alpha_{XI} & \alpha_{XX} & \alpha_{XY} & \alpha_{XZ}\\
\alpha_{YI} & \alpha_{YX} & \alpha_{YY} & \alpha_{YZ}\\
\alpha_{ZI} & \alpha_{ZX} & \alpha_{ZY} & \alpha_{ZZ}
\end{array}\right]
\]
\end{example}

Using the orthogonality of $\PP_{N}$, a simple formula for the Pauli
coefficients can be obtained:
\[
\alpha_{\matA}(\w)=\frac{\Trace{\sigma_{{\cal \w}}\matA}}{N}.
\]
Computing the Pauli coefficients via direct evaluation of this formula
is computationally expensive, with total cost of $O(N^{5})$. However,
by exploiting the structure of $\sigma_{\w}$ and the properties of
Pauli matrices, we can derive a substantially more efficient evaluation
method and reduce the cost to $O(N^{2}n)$. The basic idea is to use
\emph{Pauli submatrices} and \emph{Pauli subcoefficents}.
\begin{defn}
[Pauli submatrices] Given $\matA\in\C^{N\times N}$, and for each
$c\in\Sigma$, we define the Pauli submatrix $\matA_{c}\in\C^{(N/2)\times(N/2)}$
as:
\[
\matA_{c}\coloneqq\sum_{\w\in\Sigma^{n-1}}\alpha_{\matA}(c\w)\sigma_{{\cal \w}}
\]
When $\matA$ is Hermitian, $\matA_{c}$ is also Hermitian.
\end{defn}

\begin{example}
\label{exa:submatrics_example}Given $\matA$, we have
\begin{align*}
\matA_{I} & =\sum_{\w\in\Sigma^{n-1}}\alpha_{\matA}(I{\cal \w})\sigma_{{\cal \w}}\\
\matA_{X} & =\sum_{\w\in\Sigma^{n-1}}\alpha_{\matA}(X{\cal \w})\sigma_{{\cal \w}}\\
\matA_{Y} & =\sum_{\w\in\Sigma^{n-1}}\alpha_{\matA}(Y{\cal \w})\sigma_{{\cal \w}}\\
\matA_{Z} & =\sum_{\w\in\Sigma^{n-1}}\alpha_{\matA}(Z{\cal \w})\sigma_{{\cal \w}}
\end{align*}
\end{example}

\begin{defn}
[Pauli subcoefficient] Given $\matA\in\C^{N\times N}$, for each
$c\in\Sigma$ define the mapping $\alpha_{\matA,c}:\Sigma^{n-1}\to\C$
as follows:
\[
\alpha_{\matA,c}(\w)\coloneqq\alpha_{\matA}(c\w)
\]
The result is real if $\matA$ is Hermitian.
\end{defn}

\begin{lem}
\emph{Any matrix $\matA\in\C^{N\times N}$ can be expressed as: 
\begin{equation}
\matA=\frac{1}{\sqrt{2}}\left[\begin{array}{cc}
\matA_{I}+\matA_{Z} & \matA_{X}-\si\matA_{Y}\\
\matA_{X}+\si\matA_{Y} & \matA_{I}-\matA_{Z}
\end{array}\right]\label{eq:block-pauli}
\end{equation}
}
\end{lem}

\begin{proof}
Due to the bilinearity and associativity of the Kronecker product
we have {\scriptsize
\begin{align*}
\matA & =\sum_{\w\in\Sigma^{n}}\alpha_{{\cal \w}}\sigma_{{\cal \w}}\\
 & =\sum_{\w\in\Sigma^{n-1}}\alpha_{\matA,I}({\cal \w})\sigma_{I{\cal \w}}+\sum_{\w\in\Sigma^{n-1}}\alpha_{\matA,X}({\cal \w})\sigma_{X{\cal \w}}+\sum_{\w\in\Sigma^{n-1}}\alpha_{\matA,Y}({\cal \w})\sigma_{Y\w}+\sum_{\w\in\Sigma^{n-1}}\alpha_{\matA,Z}({\cal \w})\sigma_{Z\w}\\
 & =\sum_{\w\in\Sigma^{n-1}}\alpha_{\matA,I}({\cal \w})(\sigma_{I}\otimes\sigma_{\w})+\sum_{\w\in\Sigma^{n-1}}\alpha_{\matA,X}({\cal \w})(\sigma_{X}\otimes\sigma_{\w})+\sum_{\w\in\Sigma^{n-1}}\alpha_{\matA,Y}({\cal \w})(\sigma_{Y}\otimes\sigma_{\w})+\sum_{\w\in\Sigma^{n-1}}\alpha_{\matA,Z}({\cal \w})(\sigma_{Z}\otimes\sigma_{{\cal \w}})\\
 & =\sigma_{I}\otimes\left(\sum_{\calW\in\Sigma^{n-1}}\alpha_{\matA,I}({\cal \w})\sigma_{\w}\right)+\sigma_{X}\otimes\left(\sum_{\w\in\Sigma^{n-1}}\alpha_{\matA,X}({\cal \w})\sigma_{\w}\right)+\sigma_{Y}\otimes\left(\sum_{\w\in\Sigma^{n-1}}\alpha_{\matA,Y}({\cal \w})\sigma_{\w}\right)+\sigma_{Z}\otimes\left(\sum_{\w\in\Sigma^{n-1}}\alpha_{\matA,Z}({\cal \w})\sigma_{\w}\right)\\
 & =\sigma_{I}\otimes\matA_{I}+\sigma_{X}\otimes\matA_{X}+\sigma_{Y}\otimes\matA_{Y}+\sigma_{Z}\otimes\matA_{Z}\\
 & =\frac{1}{\sqrt{2}}\left[\begin{array}{cc}
\matA_{I}+\matA_{Z} & \matA_{X}-i\matA_{Y}\\
\matA_{X}+i\matA_{Y} & \matA_{I}-\matA_{Z}
\end{array}\right]
\end{align*}
}{\scriptsize\par}
\end{proof}
Given $\matA\in\HH_{N}$, we can write $\matA$ in block form as:
\[
\matA=\left[\begin{array}{cc}
\matA_{11} & \matA_{12}\\
\matA_{12}^{\H} & \matA_{22}
\end{array}\right]
\]
where $\matA_{11},\matA_{22}\in\HH_{N/2}$ and $\matA_{12}\in\C^{(N/2)\times(N/2)}$.
Then, the Pauli submatrices of $\matA$ are given by:
\begin{align*}
\matA_{I} & =\frac{1}{\sqrt{2}}(\matA_{11}+\matA_{22})\\
\matA_{X} & =\frac{1}{\sqrt{2}}(\matA_{12}+\matA_{12}^{\H})\\
\matA_{Y} & =\frac{\si}{\sqrt{2}}(\matA_{12}-\matA_{12}^{\H})\\
\matA_{Z} & =\frac{1}{\sqrt{2}}(\matA_{11}-\matA_{22})
\end{align*}
Thus, for any word $\w\in\Sigma^{n-1}$ we have:
\begin{align*}
\alpha_{\matA,I}({\cal \w}) & =\frac{1}{\sqrt{2}}(\alpha_{\matA_{11}}(\w)+\alpha_{\matA_{22}}(\w))\\
\alpha_{\matA,X}({\cal \w}) & =\frac{\alpha_{(\matA_{12}+\matA_{12}^{\H})}(\w)}{\sqrt{2}}\\
\alpha_{\matA,Y}({\cal \w}) & =\frac{\alpha_{(\si(\matA_{12}-\matA_{12}^{\H}))}(\w)}{\sqrt{2}}\\
\alpha_{\matA,Z}({\cal \w}) & =\frac{1}{\sqrt{2}}(\alpha_{\matA_{11}}(\w)-\alpha_{\matA_{22}}(\w))
\end{align*}

Using the above observation, we can develop an efficient recursive
algorithm for computing the Pauli decomposition of a Hermitian matrix.
First, we compute the Pauli decompositions of four $N/2$-by-$N/2$
matrices: $\matA_{11},\matA_{22},\matA_{12}+\matA_{12}^{\H},\si(\matA_{12}-\matA_{12}^{\H})$.
We then combine the resulting coefficeints to compute the Pauli decomposition
of $\matA$. The bottom of the recursion, which amounts to computing
the Pauli decomposition of a $2$-by-$2$ matrix, is trivial. The
procedure is summarized in Algorithm~\ref{alg:dense2pauli}, which
returns as output the coefficient tensor $\tenA_{P}$. The combination
costs $O(N^{2})$, so the running time obeys $T(N)=4T(N/2)+O(N^{2})$
which implies the cost of $O(N^{2}n)$.

To extend the method to non-Hermitian matrices, we decompose them
into two Hermitian parts ($\matA=\frac{\matA^{\conj}+\matA}{2}-\si\frac{\matA^{\conj}-\matA}{2i}$),
each independently decomposable into real Pauli coefficients. Combining
the coefficients, we obtain a Pauli decomposition for the original
matrix $\matA$.

\begin{algorithm}[t]
\begin{algorithmic}[1]

\STATE \textbf{Input: $\matA\in\HH_{N}$, where $N$ is a power of
2}

~

\STATE $n\gets\log_{2}N$

\IF {$q==1$}  

\STATE  $\matA_{P}\gets\frac{1}{2}\left[\begin{array}{cc}
a_{11}+a_{22} & \text{2Re}(a_{12})\\
-\text{2Im}(a_{12}) & a_{11}-a_{22}
\end{array}\right]$

\ELSE 

\STATE  Split $\matA=\left[\begin{array}{cc}
\matA_{11} & \matA_{12}\\
\matA_{12}^{\conj} & \matA_{22}
\end{array}\right]$ where each matrix is $(N/2)\times(N/2)$

~

\STATE ${\cal P}_{11}\gets\textsc{Dense2Pauli}(\matA_{11})$

\STATE ${\cal P}_{22}\gets\textsc{Dense2Pauli}(\matA_{22})$

~

\STATE ${\cal P}_{I}\gets{\cal P}_{11}+{\cal P}_{22}$

\STATE ${\cal P}_{X}\gets\textsc{Dense2Pauli}(\matA_{12}+\matA_{12}^{\H})$

\STATE ${\cal P}_{Y}\gets\textsc{Dense2Pauli}(i(\matA_{12}-\matA_{12}^{\H}))$

\STATE ${\cal P}_{Z}\gets{\cal P}_{11}-{\cal P}_{22}$

~

\STATE Initialize $\tenA_{P}$ to a order $n$ tensor of size $4\times4\times\cdots\times4$

\STATE $\tenA_{P}(0,:,\dots,:)\gets\frac{1}{2}{\cal P}_{I}$

\STATE $\tenA_{P}(1,:,\dots,:)\gets\frac{1}{2}{\cal P}_{X}$

\STATE $\tenA_{P}(2,:,\dots,:)\gets\frac{1}{2}{\cal P}_{Y}$

\STATE $\tenA_{P}(3,:,\dots,:)\gets\frac{1}{2}{\cal P}_{Z}$

\ENDIF 

~

\RETURN $\tenA_{P}$

\end{algorithmic}

\caption{\label{alg:dense2pauli}\textsc{Dense2Pauli}: Computing $\protect\tenA_{P}$
for a dense Hermitian matrix $\protect\matA$.}
\end{algorithm}

\paragraph{Related work on efficient computation of Pauli decompositions.}

Direct computation of the Pauli coefficients for an $N\times N$ matrix
using the naive approach costs $O(N^{5})$, making it impractical
for large matrices. To address this computational challenge, various
approaches have been proposed in the literature~\cite{romero2023paulicomposer,jones2024decomposing,koska2024tree,ying2023preparing,lapworth2022hybrid,gunlycke2020efficient,hantzko2024tensorized,gidney2023stackexchange,hamaguchi2024handbook}.
Koska et al. introduced a tree-based approach that reduces asymptotic
complexity to $O(N^{3})$ by hierarchically pruning redundant paths
in the Pauli operator tree~\cite{koska2024tree}, making it particularly
efficient for structured matrices (e.g., tridiagonal systems). Hantzko
et al. employed recursive matrix partitioning to avoid costly multiplications
with asymptotic scaling of $O(N^{2}n)$ for general matrices~\cite{hantzko2024tensorized}.
Their tensorized Pauli decomposition algorithm is very similar to
the one we describe in this section, and was developed in parallel.

Georges et al. presented a method for computing the Pauli coefficients
using the Walsh-Hadamard Transform~\cite{georges2024pauli} (however,
elements of their approach are traceable to a Stack Overflow comment
by Gidney~\cite{gidney2023stackexchange}). Their method achieves
$O(N^{2}n)$ complexity, and enables in-place computation with $O(1)$
additional memory requirements. However, their method is only described
for classical computing. However, it can be leveraged for quantum
conversion, as we discuss in the next subsection.

\subsection{\label{subsec:ua-to-uap}Quantum conversion - from matrix state preparation
to Pauli state preparation (and back)}

In this subsection, we discuss a algorithm for converting between
a description of the matrix in the standard basis (given in the form
of a matrix state preparation circuit $\cirU_{\matA}^{\text{SP}}$),
to a description in the Pauli basis (in the form of a hypermatrix
state preparation circuit $\cirU_{\tenA_{P}}^{\text{SP}}$). Our goal
is twofold. First, to show we can do this conversion very efficiently.
Second, we need this operation as a component in later algorithms.
Our algorithm is based on the Walsh-Hadamard Transform based algorithm
from~\cite{georges2024pauli}, though that paper describes only a
classical algorithm. Formally, we state the problem as follows:
\begin{problem}
\label{prob:uap-ua}Given a matrix state preparation circuit $\mathcal{U}_{\matA}^{\text{SP}}\in\MS{\matA}{\alpha}$,
where $\matA\in\mathbb{C}^{N\times N}$, compute hypermatrix state
preparation circuit $\mathcal{U}_{\tenA_{p}}^{\text{SP}}\in\MS{\tenA_{P}}{f(\alpha,N)}$,
where $\tenA_{P}$ is the Pauli coefficient hypermatrix of $\matA$
and $f(\alpha,N)$ is some function.
\end{problem}

Theorem~\ref{thm:ua-to-uap-and-back} describes an algorithm for
solving Problem\emph{~}\ref{prob:uap-ua}. To state and prove the
theorem, we need to define some ancillary circuits, and prove auxiliary
results.
\begin{prop}
\label{prop:kron-vec}Let $m\geq2$ be an integer. Let $\sigma^{(m)}$
be the following permutation on $(0,\dots,2m-1)$:
\[
\sigma^{(m)}(j)=\begin{cases}
j & \text{if }0\leq j<m-1\text{ or }j=2m-1\\
j+1 & \text{if }m-1\leq j<2m-2\\
m-1 & \text{if }j=2m-2
\end{cases}
\]
Denote ${\cal S}_{(m)}\coloneqq{\cal S}_{\sigma^{(m)}}$. This circuit
moves the qubit at position $2m-2$ to position $m-1$, shifting the
intermediate qubits forward by one position.

Then,
\begin{equation}
\mathbf{M}(\mathcal{S}_{(m)})=\mathbf{I}_{2^{m-1}}\otimes\mathbf{P}_{2,2^{m-1}}\otimes\mathbf{I}_{2}\label{eq:Sm_pauli}
\end{equation}
where $\mathbf{P}_{K,L}$ is the Kronecker permutation matrix, defined
as $\mathbf{P}_{K,L}\coloneqq\sum_{i=0}^{K-1}\sum_{j=0}^{L-1}\mathbf{E}_{ij}^{(K\times L)}\otimes\mathbf{E}_{ji}^{(L\times K)}$.
\end{prop}

\begin{proof}
We prove Eq.~(\ref{eq:Sm_pauli}) by analyzing the action on computational
basis states, which gives the columns of $\mathbf{M}(\mathcal{S}_{(m)})$
(the action of $\mathcal{S}_{(m)}$ on $\Ket j_{2m}$ gives column
$j$ in $\mathbf{M}(\mathcal{S}_{(m)})$). Any index $j\in\{0,1,\ldots,2^{2m-1}-1\}$
can be written in binary as $j=\sum_{\ell=0}^{2m-1}b_{\ell}\cdot2^{\ell}$,
where $b_{\ell}\in\{0,1\}$ are the binary digits. Similar to many
physics textbooks, we assume that the MSB of $i$ is in the lowest
index in the binary expansion, i.e., $\Ket j_{2m}=\Ket{b_{0}}_{1}\Ket{b_{1}}_{1}\cdots\Ket{b_{2m-1}}_{1}\eqqcolon\Ket{b_{0}\cdots b_{2m-1}}_{2m}$
where $b_{0}\cdots b_{2m-1}$ is the binary expansion of $j$.\footnote{We note that in QISKIT the MSB is the highest index, i.e., $\Ket j_{2m}=\Ket{b_{2m-1}\cdots b_{0}}_{2m}$.}

We group the bits as: 
\begin{align*}
\ket j_{2m} & =\ket{b_{0},\ldots,b_{m-2}}_{m-1}\otimes\ket{b_{m-1},\ldots,b_{2m-3}}_{m-1}\otimes\ket{b_{2m-2}}_{1}\otimes\ket{b_{2m-1}}_{1}\\
 & =\ket{\text{upper}}_{m-1}\otimes\ket{\text{middle}}_{m-1}\otimes\ket{b_{2m-2}}_{1}\otimes\ket{b_{2m-1}}_{1}
\end{align*}
By definition of permutation qubits in the lemma statement, when $\mathcal{S}_{(m)}$
operate on $\ket j_{2m}$ we transform
\[
\ket{\text{upper}}_{m-1}\otimes\ket{\text{middle}}_{m-1}\otimes\ket{b_{2m-2}}_{1}\otimes\ket{b_{2m-1}}_{1}\mapsto\ket{\text{upper}}_{m-1}\otimes\ket{b_{2m-2}}_{1}\otimes\ket{\text{middle}}_{m-1}\otimes\ket{b_{2m-1}}_{1}
\]
and this defines the $j$th column of $\mathbf{M}(\mathcal{S}_{(m)})$
.

Now, think of a circuit whose matrix is ${\cal \mathbf{I}}_{2^{m-1}}\otimes\mathbf{P}_{2,2^{m-1}}\otimes\mathbf{I}_{2}$
which acts on $\ket j_{2m}$ (thus giving column $j$ of ${\cal \mathbf{I}}_{2^{m-1}}\otimes\mathbf{P}_{2,2^{m-1}}\otimes\mathbf{I}_{2}$).
The operation is 
\begin{align*}
({\cal \mathbf{I}}_{2^{m-1}}\otimes\mathbf{P}_{2,2^{m-1}}\otimes\mathbf{I}_{2})\ket j_{2m} & =({\cal \mathbf{I}}_{2^{m-1}}\otimes\mathbf{P}_{2,2^{m-1}}\otimes\mathbf{I}_{2})(\ket{\text{upper}}_{m-1}\otimes\ket{\text{middle}}_{m-1}\otimes\ket{b_{2m-2}}_{1}\otimes\ket{b_{2m-1}}_{1})\\
 & =\ket{\text{upper}}_{m-1}\otimes\mathbf{P}_{2,2^{m-1}}(\ket{\text{middle}}_{m-1}\otimes\ket{b_{2m-2}}_{1})\otimes\ket{b_{2m-1}}_{1}\\
 & =\ket{\text{upper}}_{m-1}\otimes\ket{b_{2m-2}}_{1}\otimes\ket{\text{middle}}_{m-1}\otimes\ket{b_{2m-1}}_{1}
\end{align*}
where we use the fact that $\mathbf{P}_{2,2^{m-1}}=\sum_{i=0}^{1}\sum_{j=0}^{2^{m-1}-1}\mathbf{E}_{ij}^{(2\times2^{m-1})}\otimes\mathbf{E}_{ji}^{(2^{m-1}\times2)}$
so for any $\ket i_{m-1}\ket j_{1}$ we have $\mathbf{P}_{2,2^{m-1}}\ket i_{m-1}\ket j_{1}=\ket j_{1}\ket i_{m-1}$.

We see that the two results are the same. All columns are the same,
so the matrices are equal.
\end{proof}
\begin{defn}
[CNOT Comb Circuit] The \emph{CNOT comb circuit}, denoted as $\mathcal{U}_{\text{comb}}^{(q)}$,
is a quantum circuit that acts on $2q$ qubits. The circuit consists
of $q$ parallel CNOT gates, where each control qubit $i\in[q-1]$
is connected to its corresponding target qubit $i+q$. In other words,
the circuit applies CNOT gates between pairs $(0,q),(1,q+1),\dots(q-1,2q-1)$.
A circuit diagram is shown in Figure~\ref{fig:CNOT_Comb_Circuit}.
\end{defn}

\begin{figure}[t]
\begin{centering}
\[
\Qcircuit @C=1em @R=.9em {
&                  &         &          &              &           &             & 		&\\
& \lstick{r_0}     & \qw     & \ctrl{4} & \qw          & \qw       & \qw          & \qw	& \\
& \lstick{r_1}     & \qw     & \qw      & \ctrl{4}     & \qw       & \qw          & \qw	 & \\
&                  & \vdots  &          &              & \ddots    &              & & \\
& \lstick{r_{q-1}} & \qw     & \qw      & \qw          & \qw       & \ctrl{4}     &   \qw	& \\
& \lstick{c_0}     & \qw     & \targ    & \qw          & \qw       & \qw          &   \qw	& \\
& \lstick{c_1}     & \qw     & \qw      & \targ        & \qw       & \qw          &   \qw	  & \\
&                  & \vdots  &          &              & \ddots    &              &   	 & \\
& \lstick{c_{q-1}} & \qw     & \qw      & \qw          & \qw       & \targ        &   \qw	 & \\
}
\]
\par\end{centering}
\caption{\label{fig:CNOT_Comb_Circuit}The CNOT Comb circuit, $\mathcal{U}_{\mathrm{comb}}^{(q)}$,
on $2q$ qubits. The row qubits ($r_{i}$) act as controls on their
corresponding target column qubits ($c_{i}$).}
\end{figure}
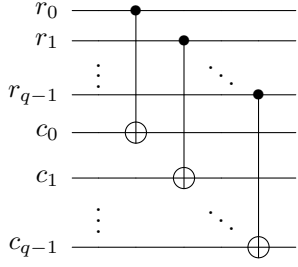

\begin{prop}
\label{prop:u-comb}The $2q$-qubit CNOT comb circuit $\mathcal{U}_{\text{comb}}^{(q)}$
implements a controlled multiplex operation, which can be expressed
as: 
\[
\mathcal{U}_{\text{comb}}^{(q)}\in{\cal MX}(\cirU_{0},\dots,\cirU_{2^{q}-1})
\]
where for each $k=\sum_{j=0}^{q-1}b(k)_{j}2^{j}$, $\matM(\mathcal{U}_{k})$
is a tensor product of $q$Pauli X gates based on the binary expansion
of $k$: 
\[
\matM(\mathcal{U}_{k})=\bigotimes_{j=0}^{q-1}\matX^{b(k)_{j}}
\]
\end{prop}

\begin{proof}
We prove by induction on $q$. For\textbf{ $q=1$.} The CNOT comb
circuit is a single CNOT gate: 
\[
\matM(\mathcal{U}_{\text{comb}}^{(1)})=\matM(\text{CX}(0,1))=(\ket 0\bra 0)_{0}\otimes\matI+(\ket 1\bra 1)_{0}\otimes\matX
\]

For the inductive step we note that $\mathcal{U}_{\text{comb}}^{(q+1)}$
circuit can be constructed recursively as: 
\[
\mathcal{U}_{\text{comb}}^{(q+1)}=\mathcal{S}_{(q+1)}\cdot\left(\mathcal{U}_{\text{comb}}^{(q)}\otimes\text{CX}(2q,2q+1)\right)\cdot\mathcal{S}_{(q+1)}^{-1}
\]
where $\mathcal{S}_{(q)}$ defined in Proposition~\ref{prop:kron-vec}
(where we set $m\gets q$). Note that $\mathcal{S}_{(q)}=\mathcal{S}_{(q,q+1)}\cdot\mathcal{S}_{(q+1,q+2)}\cdot\cdots\cdot\mathcal{S}_{(2q-1,2q)}$
where the notation $(a,b)$ denotes transposition of $a$ and $b$.

We assume the proposition holds for $q$ and show for\textbf{ $q\rightarrow q+1$}.
By the inductive hypothesis: 
\[
\matM(\mathcal{U}_{\text{comb}}^{(q)})=\sum_{k=0}^{2^{q}-1}\ket k\bra k\otimes\matM(\mathcal{U}_{k})
\]
and we have{\footnotesize
\begin{align*}
\matM(\mathcal{U}_{\text{comb}}^{(q+1)}) & =\matM(\mathcal{S}_{(q+1)})\left(\matM(\mathcal{U}_{\text{comb}}^{(q)})\otimes\text{CX}(2q,2q+1)\right)\matM(\mathcal{S}_{(q+1)}^{-1})\\
 & =\matM(\mathcal{S}_{(q+1)})\left(\sum_{k=0}^{2^{q}-1}\ket k\bra k\otimes\matM(\mathcal{U}_{k})\otimes\left[(\ket 0\bra 0)_{2q}\otimes\matI+(\ket 1\bra 1)_{2q}\otimes\matX\right]\right)\matM(\mathcal{S}_{(q+1)}^{-1})\\
 & =\matM(\mathcal{S}_{(q+1)})\left(\sum_{k=0}^{2^{q}-1}\ket k\bra k\otimes\bigotimes_{j=0}^{q-1}\matX^{b(k)_{j}}\otimes\ket 0\bra 0{}_{2q}\otimes\matI\right)\mathcal{S}_{(q)}\\
 & \,\,\,\,\,\,+\mathcal{S}_{(q)}^{-1}\left(\sum_{k=0}^{2^{q}-1}\ket k\bra k\otimes\bigotimes_{j=0}^{q-1}\matX^{b(k)_{j}}\otimes\ket 1\bra 1{}_{2q}\otimes\matX\right)\matM(\mathcal{S}_{(q+!)}^{-1})\\
 & =\sum_{k=0}^{2^{q}-1}\ket k\bra k\otimes\ket 0\bra 0{}_{2q}\otimes\bigotimes_{j=0}^{q-1}\matX^{b(k)_{j}}\otimes\matI+\sum_{k=0}^{2^{q}-1}\ket k\bra k\otimes(\ket 1\bra 1)_{2q}\otimes\bigotimes_{j=0}^{q-1}\matX^{b(k)_{j}}\otimes\matX\\
 & =\sum_{k=0}^{2^{q}-1}\ket{2k}\bra{2k}\otimes\bigotimes_{j=0}^{q}\matX^{b(2k)_{j}}+\sum_{k=0}^{2^{q}-1}\ket{2k+1}\bra{2k+1}\bigotimes_{j=0}^{q}\matX^{b(2k+1)_{j}}\\
 & =\sum_{k'=0}^{2^{q+1}-1}\ket{k'}\bra{k'}\otimes\bigotimes_{j=0}^{q}\matX^{b(k')_{j}}\\
 & =\sum_{k'=0}^{2^{q+1}-1}\ket{k'}\bra{k'}\otimes\matM(\mathcal{U}_{k'})
\end{align*}
}where we used the fact that for any set of $M$ single qubit gates
$\cirU_{0},\dots,\cirU_{M-1}$ we have: (See Figure~\ref{fig:permutation-conjugation}
for a visualization)

\[
\mathcal{S}_{\sigma}\cdot\left(\bigotimes_{i=0}^{M-1}\cirU_{i}\right)\cdot\mathcal{S}_{\sigma^{-1}}=\bigotimes_{i=0}^{M-1}\cirU_{\sigma(i)}
\]
which in terms of matrices associated with the circuit translates
to
\[
\matM(\mathcal{S}_{\sigma})\cdot\left(\bigotimes_{i=0}^{M-1}\matM(\cirU_{i})\right)\cdot\matM(\mathcal{S}_{\sigma^{-1}})=\bigotimes_{i=0}^{M-1}\matM(\cirU_{\sigma(i)})
\]
\end{proof}
\begin{figure}[t]
\begin{centering}
\begin{tabular}{ccc}
{\footnotesize\Qcircuit @C=1em @R=1em {
\lstick{\ket{\psi_{0}}}    & \multigate{4}{\mathcal{S}_{\sigma^{-1}}} & \gate{\mathcal{U}_{0}}   & \multigate{4}{\mathcal{S}_{\sigma}} & \qw & \rstick{\mathcal{U}_{\sigma(0)}\ket{\psi_{0}}}  \\
\lstick{\ket{\psi_{1}}}    & \ghost{\mathcal{S}_{\sigma^{-1}}}        & \gate{\mathcal{U}_{1}}   & \ghost{\mathcal{S}_{\sigma}}        & \qw & \rstick{\mathcal{U}_{\sigma(1)}\ket{\psi_{1}}} \\
\lstick{\ket{\psi_{2}}}    & \ghost{\mathcal{S}_{\sigma^{-1}}}        & \gate{\mathcal{U}_{2}}   & \ghost{\mathcal{S}_{\sigma}}        & \qw & \rstick{\mathcal{U}_{\sigma(2)}\ket{\psi_{2}}} \\
\raisebox{0.7em}{\vdots}     &                                          &        \raisebox{0.7em}{\vdots}          &                                   & & \\
\lstick{\ket{\psi_{M-1}}}  & \ghost{\mathcal{S}_{\sigma^{-1}}}        & \gate{\mathcal{U}_{M-1}} & \ghost{\mathcal{S}_{\sigma}}        & \qw & \rstick{\mathcal{U}_{\sigma(M-1)}\ket{\psi_{M-1}}}
}} & \hspace{1.5cm}\raisebox{-10.1ex}{=}\hspace{1.5cm} & {\footnotesize{}
\Qcircuit @C=1em @R=1em {
\lstick{\ket{\psi_{0}}}    & \gate{\mathcal{U}_{\sigma(0)}}    & \qw & \rstick{\mathcal{U}_{\sigma(0)}\ket{\psi_{0}}} \\
\lstick{\ket{\psi_{1}}}    & \gate{\mathcal{U}_{\sigma(1)}}    & \qw & \rstick{\mathcal{U}_{\sigma(1)}\ket{\psi_{1}}}\\
\lstick{\ket{\psi_{2}}}    & \gate{\mathcal{U}_{\sigma(2)}}    & \qw & \rstick{\mathcal{U}_{\sigma(2)}\ket{\psi_{2}}}\\
\raisebox{0.7em}{\vdots}     &   \raisebox{0.7em}{\vdots}                            & & \\
\lstick{\ket{\psi_{M-1}}}  & \gate{\mathcal{U}_{\sigma(M-1)}}  & \qw & \rstick{\mathcal{U}_{\sigma(M-1)}\ket{\psi_{M-1}}}
}
}\tabularnewline
\end{tabular}
\par\end{centering}
\caption{\label{fig:permutation-conjugation}Permutation conjugation property
for tensor products of single-qubit operators: $\mathcal{S}_{\sigma^{-1}}\cdot\left(\bigotimes_{i=0}^{M-1}\protect\cirU_{i}\right)\cdot\mathcal{S}_{\sigma}=\bigotimes_{i=0}^{M-1}\mathcal{U}_{\sigma(i)}$.
The permutation redistributes operators according to the mapping $\sigma$.}
\end{figure}
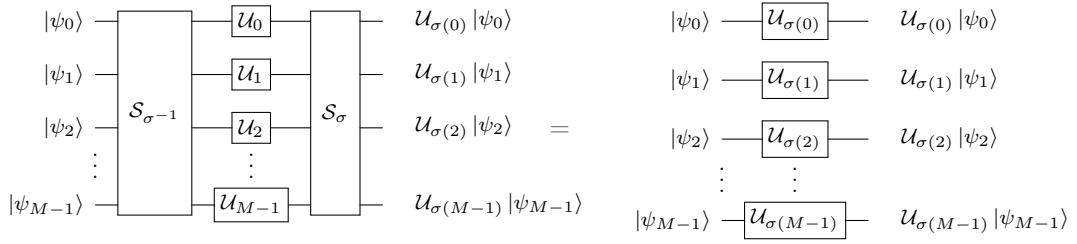

\begin{fact}
For each $\w\in\Sigma^{q}$,
\begin{align*}
N\alpha_{\matA}(\w) & =\vec{\sigma_{\w}^{\T}}^{\T}\vec{\matA}
\end{align*}
\end{fact}

\begin{proof}
We have,
\begin{align*}
N\alpha_{\matA}(\w) & =\Trace{\sigma_{{\cal \w}}\matA}\\
 & =\cvec{\sigma_{{\cal \w}}^{\T}}^{\conj}\cvec{\matA^{\T}}\\
 & =\vec{\sigma_{{\cal \w}}}^{\conj}\vec{\matA}\\
 & =\overline{\vec{\sigma_{{\cal \w}}}^{\T}}\vec{\matA}\\
 & =\vec{\overline{\sigma_{\w}}}^{\T}\vec{\matA}\\
 & =\vec{\sigma_{{\cal \w}}^{\T}}^{\T}\vec{\matA}
\end{align*}
where the last line equality follows from the fact that for an Hermitian
matrix $\matX$ we have $\overline{\matX}=\matX^{\T}$.
\end{proof}
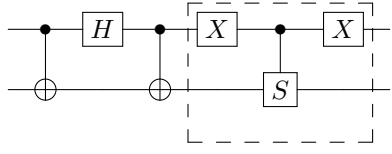
\begin{figure}[t]
\begin{centering}
\[
\Qcircuit @C=1em @R=1em {
 \lstick{} & \ctrl{1} & \gate{H} & \ctrl{1} & \gate{X} & \ctrl{1} & \gate{X} & \qw \\
 \lstick{} & \targ & \qw & \targ & \qw & \gate{S} & \qw &  \qw \gategroup{1}{5}{3}{7}{.7em}{--} \\
 &  &  &  &  & &  & \\
}
\]
\par\end{centering}
\caption{Circuit \label{fig:u1-lemma-illustration} of $\protect\cirU^{(1)}$.
The dashed box highlights the diagonal circuit ${\cal D}$.}
\end{figure}

\begin{lem}
\label{lem:rows_of_pauli}Consider the circuit,
\[
\cirU^{(q)}={\cal D}_{(q)}\cdot\mathcal{{\cal P}}_{(q)}\cdot\mathcal{U}_{\mathrm{comb}}^{(q)}\cdot\bigl({\cal H}^{\otimes q}\otimes{\cal I}_{q}\bigr)\cdot\mathcal{U}_{\mathrm{comb}}^{(q)}
\]
where ${\cal D}_{(q)}=\bigotimes_{i=0}^{q}{\cal D}$ with $\matM({\cal D})=\diag{1,1,\si,1}$,
$\mathcal{{\cal P}}_{(q)}$ is a circuit that interlaces qubits of
two size $q$ registers, that is if the original qubits are labeled
$(0,1,\dots,q-1,\;q,q+1,\dots,2q-1)$ then after ${\cal P}_{(q)}$
the qubits are ordered $(0,q,1,q+1,\dots,q-1,2q-1)$, and ${\cal H}^{\otimes q}$
denotes the $q$-qubit circuit that has an Hadamard gate on each qubit
in parallel. See Figure~\ref{fig:U_q} for a visualization of the
circuit. Then the rows of the matrix $\matM\bigl(\cirU^{(q)}\bigr)$
are exactly the vectorizations of the \emph{$q$-wise Pauli matrices}
$\sigma_{\w_{0}},\dots,\sigma_{\w_{4^{q}-1}}$, in lexicographic order:
\[
\matM\bigl(\cirU^{(q)}\bigr)=\frac{1}{\sqrt{2^{q}}}\begin{bmatrix}\vec{\sigma_{\w_{0}}^{\T}}^{\T}\\
\vec{\sigma_{\w_{1}}^{\T}}^{\T}\\
\vdots\\
\vec{\sigma_{\w_{4^{q}-1}}^{\T}}^{\T}
\end{bmatrix}.
\]
\end{lem}

\begin{proof}
We prove by induction on $q$. For $q=1$ one checks directly that
(see Figure~\ref{fig:u1-lemma-illustration}):
\begin{align*}
\cirU^{(1)} & =\frac{1}{\sqrt{2}}{\cal D}\left[\begin{array}{cc}
\matI_{2}\\
 & \matX
\end{array}\right]\left[\begin{array}{cc}
\matI_{2} & \matI_{2}\\
\matI_{2} & -\matI_{2}
\end{array}\right]\left[\begin{array}{cc}
\matI_{2}\\
 & \matX
\end{array}\right]\\
 & =\frac{1}{\sqrt{2}}\diag{1,1,\si,1}\left[\begin{array}{cc}
\matI & \matX\\
\matX & -\matI
\end{array}\right]\\
 & =\frac{1}{\sqrt{2}}\diag{1,1,\si,1}\left[\begin{array}{cccc}
1 & 0 & 0 & 1\\
0 & 1 & 1 & 0\\
0 & 1 & -1 & 0\\
1 & 0 & 0 & -1
\end{array}\right]\\
 & =\frac{1}{\sqrt{2}}\left[\begin{array}{cccc}
1 & 0 & 0 & 1\\
0 & 1 & 1 & 0\\
0 & \si & -\si & 0\\
1 & 0 & 0 & -1
\end{array}\right]\\
 & =\frac{1}{\sqrt{2}}\left[\begin{array}{c}
\vec{\matI_{2}}^{\T}\\
\vec{\matX}^{\T}\\
\vec{\matY^{\T}}^{\T}\\
\vec{\matZ}^{\T}
\end{array}\right]
\end{align*}
which agrees with the four one-qubit Paulis $(\sigma_{\w_{0}},\dots,\sigma_{\w_{3}})=(I,X,Y,Z)$
in lexicographic order.

Assume the statement holds for $q$, i.e.,
\[
\matM\bigl(\cirU^{(q)}\bigr)=\frac{1}{\sqrt{2^{q}}}\begin{bmatrix}\vec{\sigma_{w_{0}}^{\T}}^{\T}\\
\vdots\\
\vec{\sigma_{w_{4^{q}-1}}^{\T}}^{\T}
\end{bmatrix}.
\]
We first make the following observation. Let ${\cal V}^{(q)}\coloneqq\mathcal{U}_{\mathrm{comb}}^{(q)}\cdot\bigl({\cal H}^{\otimes q}\otimes{\cal I}_{q}\bigr)\cdot\mathcal{U}_{\mathrm{comb}}^{(q)}$,
so that by definition $\cirU^{(q)}={\cal D}_{(q)}\cdot\mathcal{{\cal P}}_{(q)}\cdot{\cal V}^{(q)}$.
We have the following recursive relation:
\begin{align*}
\matM({\cal V}^{(1)}) & =\matM(\mathcal{U}_{\mathrm{comb}}^{(1)}\cdot\bigl({\cal H}\otimes{\cal I}_{1})\cdot\mathcal{U}_{\mathrm{comb}}^{(1)})\\
 & =\ket 0\bra 0\otimes\matI+\ket 0\bra 1\otimes\matX+\ket 1\bra 0\otimes\matX-\ket 1\bra 1\otimes\matI
\end{align*}
and 
\[
{\cal V}^{(q+1)}=\mathcal{S}_{(q+1)}\cdot\left({\cal V}^{(q)}\otimes{\cal V}^{(1)}\right)\cdot\mathcal{S}_{(q+1)}^{-1}
\]
Indeed, we have {\footnotesize
\begin{align*}
\mathcal{V}^{(q+1)} & =\mathcal{U}_{\text{comb}}^{(q+1)}\cdot\left({\cal H}^{\otimes q+1}\otimes{\cal I}_{q+1}\right)\cdot\mathcal{U}_{\text{comb}}^{(q+1)}\\
 & =\left[\mathcal{S}_{(q+1)}\cdot(\mathcal{U}_{\text{comb}}^{(q)}\otimes\text{CNOT}_{2q,2q+1})\cdot\mathcal{S}_{(q+1)}^{-1}\right]\cdot\left[\mathcal{S}_{(q+1)}\cdot({\cal H}^{\otimes q}\otimes{\cal I}_{q}\otimes{\cal H}\otimes{\cal I}_{1})\cdot\mathcal{S}_{(q+1)}^{-1}\right]\cdot\left[\mathcal{S}_{(q+1)}\cdot(\mathcal{U}_{\text{comb}}^{(q)}\otimes\text{CNOT}_{2q,2q+1})\cdot\mathcal{S}_{(q+1)}^{-1}\right]\\
 & =\mathcal{S}_{(q+1)}\cdot(\mathcal{U}_{\text{comb}}^{(q)}\otimes\text{CNOT}_{2q,2q+1})\cdot({\cal H}^{\otimes q}\otimes{\cal I}_{q}\otimes{\cal H}\otimes{\cal I}_{1})\cdot(\mathcal{U}_{\text{comb}}^{(q)}\otimes\text{CNOT}_{2q,2q+1})\cdot\mathcal{S}_{(q+1)}^{-1}\\
 & =\mathcal{S}_{(q+1)}\cdot\left((\mathcal{U}_{\text{comb}}^{(q)}\cdot({\cal H}^{\otimes q}\otimes{\cal I}_{q})\cdot\mathcal{U}_{\text{comb}}^{(q)})\otimes(\text{CNOT}_{2q,2q+1}\cdot({\cal H}\otimes{\cal I}_{1})\cdot\text{CNOT}_{2q,2q+1})\right)\cdot\mathcal{S}_{(q+1)}^{-1}\\
 & =\mathcal{S}_{(q+1)}\cdot\left(\mathcal{V}^{(q)}\otimes\mathcal{V}^{(1)}\right)\cdot\mathcal{S}_{(q+1)}^{-1}
\end{align*}
}where we used the Proposition~\ref{prop:u-comb} and ${\cal H}^{\otimes q+1}\otimes{\cal I}_{q+1}=\mathcal{S}_{(q+1)}\cdot\left(({\cal H}^{\otimes q}\otimes{\cal I}_{q})\otimes({\cal H}\otimes{\cal I}_{1})\right)\cdot\mathcal{S}_{(q+1)}^{-1}$.

The next observation is that:
\[
\mathcal{{\cal P}}_{(q+1)}=({\cal P}_{(q)}\otimes{\cal I}_{1})\cdot\mathcal{S}_{(q+1)}^{-1}
\]
where $\mathcal{S}_{(q+1)}$ is defined in the previous proposition.
Proof of this observation is deferred to the end, in order not to
disturb the flow of the proof. The last equation implies that, 
\begin{align*}
\cirU^{(q+1)} & ={\cal D}_{(q+1)}\cdot\mathcal{{\cal P}}_{(q+1)}\cdot{\cal V}^{(q+1)}\\
 & ={\cal D}_{(q+1)}\cdot({\cal P}_{(q)}\otimes{\cal I}_{1})\cdot\mathcal{S}_{(q+1)}^{-1}\cdot\mathcal{S}_{(q+1)}\cdot\left({\cal V}^{(q)}\otimes V^{(1)}\right)\cdot\mathcal{S}_{(q+1)}^{-1}\\
 & ={\cal D}_{(q+1)}\cdot(\mathcal{{\cal P}}_{(q)}{\cal V}^{(q)}\otimes{\cal I}_{1}{\cal V}^{(1)})\cdot\mathcal{S}_{(q+1)}^{-1}\\
 & =(D_{(q)}\otimes{\cal D})\cdot\left(\mathcal{{\cal P}}_{(q)}{\cal V}^{(q)}\otimes{\cal V}^{(1)}\right)\cdot\mathcal{S}_{(q+1)}^{-1}\\
 & =(\cirU^{(q)}\otimes\cirU^{(1)})\cdot\mathcal{S}_{(q+1)}^{-1}
\end{align*}

The following is a known identity: for every $\matA\in\C^{N\times M}$
and $\matB\in^{L\times K}$ we have 
\[
\cvec{\matA\otimes\matB}=(\matI_{M}\otimes\matP_{K,N}\otimes\matI_{L})\left[\cvec{\matA}\otimes\cvec{\matB}\right]
\]
where $\cvec{\cdot}$ denotes column-major vectorization of a matrix,
and $\matP_{N,M}$ is the Kronecker permutation matrix. 
\[
\matP_{N,M}=\sum_{i=0}^{N-1}\sum_{j=0}^{M-1}\matE_{ij}^{(N\times M)}\otimes\matE_{ji}^{(M\times N)}
\]
(see \cite[Fact 7.4.29.(xii)]{bernstein2009matrix}). Translating
to row-major vectorization, we have 
\[
\vec{\matA^{\T}\otimes\matB^{\T}}=(\matI_{M}\otimes\matP_{K,N}\otimes\matI_{L})\left[\vec{\matA^{\T}}\otimes\vec{\matB^{\T}}\right]
\]

In our case, we have for any 2-by-2 Pauli matrix $\sigma_{i}$ and
for any $q$-wise Pauli matrix $\sigma_{\w}\in\C^{N\times N}$ (note
that in the following, we are using transpose and not conjugate transpose):
\[
\vec{\sigma_{\w}^{\T}\otimes\sigma_{i}^{\T}}=(\matI_{N}\otimes\matP_{2,N}\otimes\matI_{2})\left[\vec{\sigma_{\w}^{\T}}\otimes\vec{\sigma_{i}^{\T}}\right]
\]
 So, we have{\footnotesize
\begin{align*}
\matM\bigl(\cirU^{(q+1)}\bigr) & =\left(\matM\bigl(\cirU^{(q)}\bigr)\otimes\matM\bigl(\cirU^{(1)}\bigr)\right)\matM(\mathcal{S}_{(q+1)}^{-1})\\
 & =\left(\frac{1}{\sqrt{2^{q}}}\begin{bmatrix}\vec{\sigma_{w_{0}}^{\T}}^{\T}\\
\vec{\sigma_{w_{1}}^{\T}}^{\T}\\
\vdots\\
\vec{\sigma_{w_{4^{q}-1}}^{\T}}^{\T}
\end{bmatrix}\otimes\frac{1}{\sqrt{2}}\left[\begin{array}{c}
\vec{\matI^{\T}}^{\T}\\
\vec{\matX^{\T}}^{\T}\\
\vec{\matY^{\T}}^{\T}\\
\vec{\matZ^{\T}}^{\T}
\end{array}\right]\right)\matM(\mathcal{S}_{(q+1)}^{-1})\\
 & =\frac{1}{\sqrt{2^{q+1}}}\begin{bmatrix}\vec{\sigma_{\w_{0}}^{\T}}^{\T}\otimes\vec{\matI^{\T}}^{\T}\\
\vec{\sigma_{\w_{0}}^{\T}}^{\T}\otimes\vec{\matX^{\T}}^{\T}\\
\vec{\sigma_{\w_{0}}^{\T}}^{\T}\otimes\vec{\matY^{\T}}^{\T}\\
\vec{\sigma_{\w_{0}}^{\T}}^{\T}\otimes\vec{\matZ^{\T}}^{\T}\\
\vec{\sigma_{\w_{1}}^{\T}}^{\T}\otimes\vec{\matI^{\T}}^{\T}\\
\vdots
\end{bmatrix}\matM(\mathcal{S}_{(q+1)}^{-1})\\
 & =\frac{1}{\sqrt{2^{q+1}}}\left(\matM(\mathcal{S}_{(q+1)})\begin{bmatrix}\vec{\sigma_{\w_{0}}^{\T}}^{\T}\otimes\vec{\matI^{\T}}^{\T}\\
\vec{\sigma_{\w_{0}}^{\T}}^{\T}\otimes\vec{\matX^{\T}}^{\T}\\
\vec{\sigma_{\w_{0}}^{\T}}^{\T}\otimes\vec{\matY^{\T}}^{\T}\\
\vec{\sigma_{\w_{0}}^{\T}}^{\T}\otimes\vec{\matZ^{\T}}^{\T}\\
\vec{\sigma_{\w_{1}}^{\T}}^{\T}\otimes\vec{\matI^{\T}}^{\T}\\
\vdots
\end{bmatrix}^{\T}\right)^{\T}\\
 & =\frac{1}{\sqrt{2^{q+1}}}\left(\matI_{N}\otimes\matP_{2,N}\otimes\matI_{2}\begin{bmatrix}\vec{\sigma_{\w_{0}}^{\T}}^{\T}\otimes\vec{\matI^{\T}}^{\T}\\
\vec{\sigma_{\w_{0}}^{\T}}^{\T}\otimes\vec{\matX^{\T}}^{\T}\\
\vec{\sigma_{\w_{0}}^{\T}}^{\T}\otimes\vec{\matY^{\T}}^{\T}\\
\vec{\sigma_{\w_{0}}^{\T}}^{\T}\otimes\vec{\matZ^{\T}}^{\T}\\
\vec{\sigma_{\w_{1}}^{\T}}^{\T}\otimes\vec{\matI^{\T}}^{\T}\\
\vdots
\end{bmatrix}^{\T}\right)^{\T}\\
 & =\frac{1}{\sqrt{2^{q+1}}}\begin{bmatrix}\vec{(\sigma_{\w_{0}}\otimes\matI)^{\T}}^{\T}\\
\vec{(\sigma_{\w_{0}}\otimes\matX)^{\T}}^{\T}\\
\vec{(\sigma_{\w_{0}}\otimes\matY)^{\T}}^{\T}\\
\vec{(\sigma_{\w_{0}}\otimes\matZ)^{\T}}^{\T}\\
\vec{(\sigma_{\w_{1}}\otimes\matI)^{\T}}^{\T}\\
\vdots\\
\vec{(\sigma_{\w_{4^{q}-1}}^{\T}\otimes\matZ)^{\T}}^{\T}
\end{bmatrix}
\end{align*}
}where the fifth equality follows from Proposition~\ref{prop:kron-vec}.

Finally we need to complete the argument $\mathcal{{\cal P}}_{(q+1)}=({\cal P}_{(q)}\otimes{\cal I}_{1})\cdot\mathcal{S}_{(q+1)}^{-1}$.
It is easy to see that we have $\mathcal{P}_{(q)}={\cal S}_{\tau(q)}$
where $\tau_{(q)}\in S_{2q}$ is the following permutation:
\[
\tau_{(q)}(i)=\begin{cases}
2i & \text{if }0\leq i<q\\
2(i-q)+1 & \text{if }q\leq i<2q
\end{cases}
\]
\[
\]
The permutation $\tau_{(q+1)}$ can be decomposed as $\tau_{(q+1)}=(\tau_{(q)}\oplus\text{id}_{1})\circ(\sigma^{(q+1)})^{-1}$,
which implies that 
\[
\mathcal{{\cal P}}_{(q+1)}=({\cal P}_{(q)}\otimes{\cal I}_{1})\cdot\mathcal{S}_{(q+1)}^{-1}
\]
as required. To see that $\tau_{(q+1)}=(\tau_{(q)}\oplus\text{id}_{1})\circ(\sigma^{(q+1)})^{-1}$
we show that both permutations produce identical results when applied
to any index $i\in\{0,1,\ldots,2q+1\}$.
\begin{itemize}
\item For $i\in\{0,1,\ldots,q-1\}$: since $i<q$, we have $(\sigma^{(q+1)})^{-1}(i)=i$.
Thus: 
\begin{align*}
(\tau_{(q)}\oplus\text{id}_{1})(\sigma^{(q+1)})^{-1})(i) & =(\tau_{(q)}\oplus\text{id}_{1})(i)\\
 & =\tau_{(q)}(i)\quad\text{(since }i<2q)\\
 & =2i
\end{align*}
On the other hand, since $i<q<q+1$ we have, by definition: 
\[
\tau_{(q+1)}(i)=2i
\]
\item For $i=q$: we have $(\sigma^{(q+1)})^{-1}(q)=2q$, so since $2q\geq2q$
we have
\begin{align*}
(\tau_{(q)}\oplus\text{id}_{1})(\sigma^{(q+1)})^{-1})(q) & =(\tau_{(q)}\oplus\text{id}_{1})(2q)\\
 & =2q\quad\text{(since }\tau_{(q)}\oplus\text{id}_{1}\text{ fixes }2q\text{)}
\end{align*}
On the other hand, since $q<q+1$: 
\[
\tau_{(q+1)}(q)=2q
\]
\item For $i\in\{q+1,q+2,\ldots,2q\}$: we have $(\sigma^{(q+1)})^{-1}(i)=i-1$,
so since $q\leq i-1<2q$: 
\begin{align*}
(\tau_{(q)}\oplus\text{id}_{1})(\sigma^{(q+1)})^{-1})(i) & =(\tau_{(q)}\oplus\text{id}_{1})(i-1)\\
 & =\tau_{(q)}(i-1)\quad\text{(since }i-1<2q\text{)}\\
 & =2((i-1)-q)+1\quad\text{(since }i-1\geq q\text{)}\\
 & =2i-2q-1
\end{align*}
On the other hand, since $i\geq q+1$: 
\begin{align*}
\tau_{(q+1)}(i) & =2(i-(q+1))+1\\
 & =2i-2q-1
\end{align*}
\item For $i=2q+1$: since $i=2q+1>2q$ we have $(\sigma^{(q+1)})^{-1}(2q+1)=2q+1$,
and so
\begin{align*}
((\tau_{(q)}\oplus\text{id}_{1})(\sigma^{(q+1)})^{-1})(2q+1) & =(\tau_{(q)}\oplus\text{id}_{1})(2q+1)\\
 & =2q+1\quad\text{(since }(\tau_{(q)}\oplus\text{id}_{1})\text{ fixes }2q+1\text{)}
\end{align*}
On the other hand, since $2q+1\geq q+1$: 
\begin{align*}
\tau_{(q+1)}(2q+1) & =2((2q+1)-(q+1))+1\\
 & =2q+1
\end{align*}
\end{itemize}
This covers all cases. Since the two permutations are equal for all
indices, the two permutations are equal.
\end{proof}
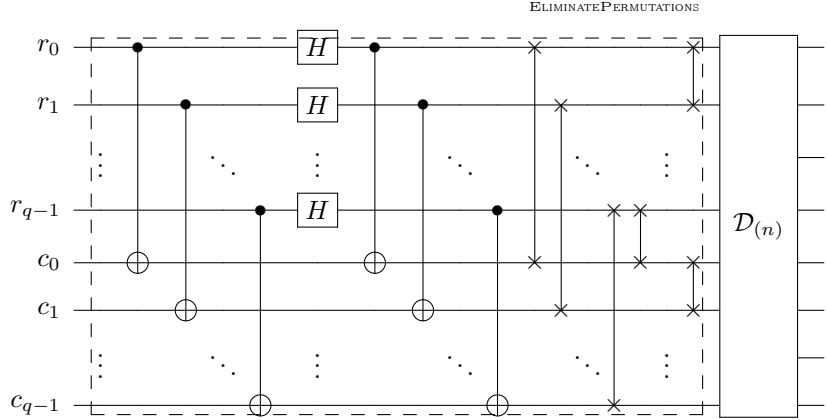
\begin{figure}[t]
\begin{centering}
\[
\Qcircuit @C=1em @R=.9em {
&                  &         &      &                    &            &             &          &            &          &          &          &                                      &      & {\hspace{2.0em}\tiny\textsc{EliminatePermutations}} 											&\\
& \lstick{r_0}     & \qw     & \ctrl{4} & \qw          & \qw       & \qw        & \gate{H} & \ctrl{4}   & \qw       & \qw       & \qw        & \qswap & \qw        & \qw       & \qw             & \qw          & \qw    & \qswap \qwx[1] 				& \multigate{7}{{\cal D}_{(n)}}  & \qw	& \\
& \lstick{r_1}     & \qw      & \qw      & \ctrl{4}     & \qw       & \qw        & \gate{H} & \qw        & \ctrl{4}  & \qw       & \qw               & \qw    & \qswap     & \qw       & \qw             & \qw          & \qw    & \qswap    			& \ghost{{\cal D}_{(n)}}         & \qw	 & \\
&                  & \vdots   &          &              & \ddots    &            & \vdots   &            &           & \ddots    &                  &       &            & \ddots    &                 &              & \vdots &           			& \nghost{{\cal D}_{(n)}}        & \qw	 & \\
& \lstick{r_{q-1}} & \qw      & \qw      & \qw          & \qw       & \ctrl{4}   & \gate{H} & \qw        & \qw       & \qw       & \ctrl{4}          & \qw    & \qw        & \qw       & \qswap          & \qswap \qwx[1] & \qw  & \qw        			& \ghost{{\cal D}_{(n)}}         &   \qw	& \\
& \lstick{c_0}     & \qw      & \targ    & \qw          & \qw       & \qw        & \qw      & \targ      & \qw       & \qw       & \qw               & \qswap \qwx[-4] & \qw & \qw       & \qw             & \qswap       & \qw    & \qswap \qwx[1] 	& \ghost{{\cal D}_{(n)}}         &   \qw	& \\
& \lstick{c_1}     & \qw      & \qw      & \targ        & \qw       & \qw        & \qw      & \qw        & \targ     & \qw       & \qw               & \qw    & \qswap \qwx[-4] & \qw   & \qw             & \qw          & \qw    & \qswap 		 	& \ghost{{\cal D}_{(n)}}         &   \qw	  & \\
&                  & \vdots   &          &              & \ddots    &            & \vdots   &            &           & \ddots    &                 &       &            & \ddots    &                 &              & \vdots &            			& \nghost{{\cal D}_{(n)}}        &   \qw	 & \\
& \lstick{c_{q-1}} & \qw      & \qw      & \qw          & \qw       & \targ      & \qw      & \qw        & \qw       & \qw       & \targ      \gategroup{2}{3}{9}{19}{.7em}{--} & \qw & \qw & \qw & \qswap \qwx[-4] & \qw & \qw & \qw      			& \ghost{{\cal D}_{(n)}}         &   \qw	 & \\
}
\]
\par\end{centering}
\caption{\label{fig:U_q}Visualization of the $\protect\cirU^{(q)}={\cal D}_{(q)}\cdot\mathcal{{\cal P}}_{(q)}\cdot\mathcal{U}_{\mathrm{comb}}^{(q)}\cdot\bigl({\cal H}^{\otimes q}\otimes{\cal I}_{q}\bigr)\cdot\mathcal{U}_{\mathrm{comb}}^{(q)}$.
In order to use the \noun{EliminatePermutations} optimization, the
final circuit must be a state preparation circuit, meaning it must
act on the ground state.}
\end{figure}

\begin{rem}
In order to create ${\cal D}$ such that $\matM({\cal D})=\diag{1,1,\si,1}$
we note that
\[
\matM({\cal D})=\left[\begin{array}{cc}
\matI\\
 & \matX\matS\matX
\end{array}\right]=\left[\begin{array}{cc}
\matX\\
 & \matX
\end{array}\right]\left[\begin{array}{cc}
\matI\\
 & \matS
\end{array}\right]\left[\begin{array}{cc}
\matX\\
 & \matX
\end{array}\right]
\]
See Figure~\ref{fig:u1-lemma-illustration} for a circuit diagram
of ${\cal D}$.

\begin{figure}[t]
\begin{centering}
\[
\Qcircuit @C=1em @R=.9em {
&                   &        & 		&      &                    &            &             &          &            &          &          &          &                                      &      & {\hspace{2.0em}\tiny\textsc{EliminatePermutations}} 											& &&&&&&&\\
& \lstick{a_0}      &  \qw   & \multigate{11}{{\cal U}^{\text{SP}}_{\matA}} 	&  \qw    &    \qw  &  \qw          &    \qw         &   \qw       &    \qw        &    \qw      &  \qw        &  \qw        &     \qw                                 & \qw     & 			 	\qw										&\qw&    \qw         &   \qw       &    \qw        &    \qw      &  \qw        &  \qw        &\\
& \lstick{a_1}      &  \qw   & \ghost{{\cal U}^{\text{SP}}_{\matA}} 	&  \qw    &    \qw                &    \qw        &      \qw       &   \qw       &      \qw      &   \qw       &    \qw      &    \qw      &      \qw                                &  \qw    & 		\qw	 											&\qw&    \qw         &   \qw       &    \qw        &    \qw      &  \qw        &  \qw        &\\
&                   & \vdots & \nghost{{\cal U}^{\text{SP}}_{\matA}} &      &                    &            &             &   {\cal V}^{(n)}        &            &          &          &          &       &      & & \mathcal{P}_{(n)} &&&&&&&&\\
& \lstick{a_{q_{N}}}&  \qw   & \ghost{{\cal U}^{\text{SP}}_{\matA}} &  \qw    &   \qw                 &    \qw        &   \qw          &  \qw        &     \qw       &    \qw      &   \qw       &  \qw        &    \qw       &  \qw    & 	\qw		&\qw&    \qw         &   \qw       &    \qw        &    \qw      &  \qw        &  \qw        &\\
& \lstick{r_0}      & \qw    & \ghost{{\cal U}^{\text{SP}}_{\matA}} & \ctrl{4} & \qw          & \qw       & \qw        & \gate{H} & \ctrl{4}   & \qw       & \qw       & \qw        & \qswap & \qw        & \qw       & \qw             & \qw          & \qw    & \qswap \qwx[1] 				& \multigate{7}{{\cal D}_{(n)}}  & \qw	& \\
& \lstick{r_1}      & \qw    & \ghost{{\cal U}^{\text{SP}}_{\matA}}         & \qw      & \ctrl{4}     & \qw       & \qw        & \gate{H} & \qw        & \ctrl{4}  & \qw       & \qw               & \qw    & \qswap     & \qw       & \qw             & \qw          & \qw    & \qswap    			& \ghost{{\cal D}_{(n)}}         & \qw	 & \\
&                   & \vdots & \nghost{{\cal U}^{\text{SP}}_{\matA}}        &          &              & \ddots    &            & \vdots   &            &           & \ddots    &                  &       &            & \ddots    &                 &              & \vdots &           			& \nghost{{\cal D}_{(n)}}        & 	 & \\
& \lstick{r_{n-1}}  & \qw    & \ghost{{\cal U}^{\text{SP}}_{\matA}}         & \qw      & \qw          & \qw       & \ctrl{4}   & \gate{H} & \qw        & \qw       & \qw       & \ctrl{4}          & \qw    & \qw        & \qw       & \qswap          & \qswap \qwx[1] & \qw  & \qw        			& \ghost{{\cal D}_{(n)}}         &   \qw	& \\
& \lstick{c_0}      & \qw    & \ghost{{\cal U}^{\text{SP}}_{\matA}}         & \targ    & \qw          & \qw       & \qw        & \qw      & \targ      & \qw       & \qw       & \qw  \gategroup{5}{14}{13}{20}{0.7em}{^\}}             & \qswap \qwx[-4] & \qw & \qw       & \qw             & \qswap       & \qw    & \qswap \qwx[1] 	& \ghost{{\cal D}_{(n)}}         &   \qw	& \\
& \lstick{c_1}      & \qw    & \ghost{{\cal U}^{\text{SP}}_{\matA}}         & \qw      & \targ        & \qw       & \qw        & \qw      & \qw        & \targ     & \qw       & \qw  \gategroup{5}{5}{13}{13}{0.7em}{^\}}             & \qw    & \qswap \qwx[-4] & \qw   & \qw             & \qw          & \qw    & \qswap 		 	& \ghost{{\cal D}_{(n)}}         &   \qw	  & \\
&                   & \vdots & \nghost{{\cal U}^{\text{SP}}_{\matA}}        &          &              & \ddots    &            & \vdots   &            &           & \ddots    &                 &       &            & \ddots    &                 &              & \vdots &            			& \nghost{{\cal D}_{(n)}}        &   	 & \\
& \lstick{c_{n-1}}  & \qw    & \ghost{{\cal U}^{\text{SP}}_{\matA}}         & \qw      & \qw          & \qw       & \targ      & \qw      & \qw        & \qw       & \qw       & \targ   \gategroup{2}{3}{13}{20}{1.1em}{--}    & \qw & \qw & \qw & \qswap \qwx[-4] & \qw & \qw & \qw      			& \ghost{{\cal D}_{(n)}}         &   \qw	 & \\
}
\]
\par\end{centering}
\caption{\label{fig:U_A_to_U_Ap}Circuit of $\mathcal{U}_{\protect\tenA_{P}}^{\text{SP}}$.
The top $q_{N}\protect\coloneqq q(\mathcal{U}_{\protect\matA}^{\text{SP}})-2n$
wires correspond to ancillary qubits on which only $\mathcal{U}_{\protect\matA}^{\text{SP}}$
acts. The bottom $2n$ wires are the data qubits, where after applying
$\mathcal{U}_{\protect\matA}^{\text{SP}}$ the sequence of circuits
${\cal V}^{(n)}$, $\mathcal{P}_{(n)}$, and $\mathcal{D}_{(n)}$
are applied.}
\end{figure}
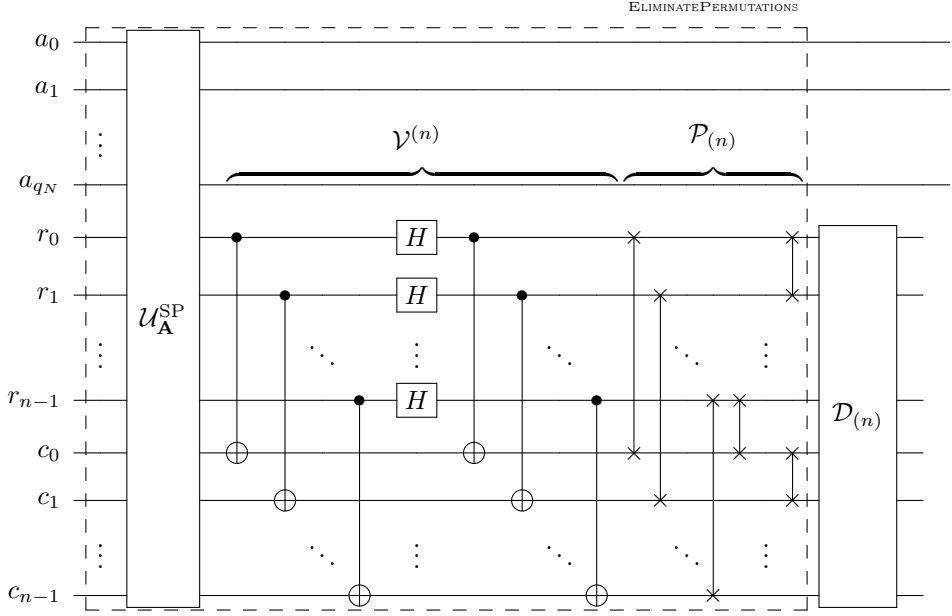
\end{rem}

We are now ready to state the main theorem of this section.
\begin{thm}
\emph{\label{thm:ua-to-uap-and-back}Let $\matA\in\C^{N\times N}$
where $N=2^{n}$. Given a classical description of a matrix state
preparation circuit ${\cal U}_{\matA}^{\text{SP}}\in\MS{\matA}{\alpha}$,
let $\mathcal{U}_{\tenA_{P}}^{\text{SP}}\coloneqq({\cal I}_{q(\cirU_{\matA})-2n}\otimes\cirU^{(n)})\cdot\cirU_{\matA}^{\text{SP}}$
(see Lemma~\ref{lem:rows_of_pauli} for $\cirU^{(n)}$'s definition;
if $q(\cirU_{\matA}^{\text{SP}})=2n$ then then identity part is omitted);
see Figure~\ref{fig:U_A_to_U_Ap} for a circuit diagram. The circuit
$\mathcal{U}_{\tenA_{P}}^{\text{SP}}$ is a hypermatrix state preparation
for the Pauli coefficient hypermatrix $\tenA_{P}$. In particular,
$\mathcal{U}_{\tenA_{P}}^{\text{SP}}\in\TS{\tenA_{P}}{\nicefrac{\alpha}{\sqrt{N}}}$.
The circuit $\mathcal{U}_{\tenA_{P}}^{\text{SP}}$ can be constructed
with classical cost $O(g({\cal U}_{\matA}^{\text{SP}})+n)$ , and
with $g({\cal U}_{\matA}^{\text{SP}})+6n$ gate complexity. This construction
yields a depth of $d({\cal U}_{\matA}^{\text{SP}})+5$, and with surplus
T-gate and T-depth (compared to the T-costs of $\cirU_{\matA}^{\text{SP}}$)
is $n$ and $1$ (respectively).}
\end{thm}

\begin{proof}
From Lemma~\ref{lem:rows_of_pauli} the matrix $\matM(\cirU^{(n)})$
consists of rows that are the vectorized (and transposed) Pauli operators:
\[
\matM(\cirU^{(n)})=\frac{1}{\sqrt{N}}\begin{bmatrix}\vec{\sigma_{\w_{0}}^{\T}}^{\T}\\
\vec{\sigma_{\w_{1}}^{\T}}^{\T}\\
\vdots\\
\vec{\sigma_{\w_{4^{n}-1}}^{\T}}^{\T}
\end{bmatrix}
\]
Since $\cirU_{\matA}^{\text{SP}}\in\MS{\matA}{\alpha}$ , there exists
a $\theta$ such that 
\[
\alpha e^{\si\theta}\matM(\cirU_{\matA}^{\text{SP}})=\begin{bmatrix}\vdots & * & \cdots & *\\[1mm]
\vec{{\cal \matA}} & \vdots &  & \vdots\\[1mm]
\vdots & \vdots &  & \vdots\\[1mm]
\vpsi & \vdots &  & \vdots\\[1mm]
\vdots & * & \cdots & *
\end{bmatrix}
\]
where $\vpsi\in\C^{2^{q(\cirU_{\matA}^{\text{SP}})-2n}}$. So,
\begin{align*}
\alpha e^{\si\theta}\matM({\cal I}_{q(\cirU_{\matA})-2n}\otimes\cirU^{(n)})\matM(\cirU_{\matA}^{\text{SP}}) & =\left[\begin{array}{ccc}
\matM(\cirU^{(n)}) & 0 & \cdots\\
0 & \ddots & 0\\
\vdots & 0 & \matM(\cirU^{(n)})
\end{array}\right]\begin{bmatrix}\vdots & * & \cdots & *\\[1mm]
\vec{{\cal \matA}} & \vdots &  & \vdots\\[1mm]
\vdots & \vdots &  & \vdots\\[1mm]
\vpsi & \vdots &  & \vdots\\[1mm]
\vdots & * & \cdots & *
\end{bmatrix}\\
 & =\sqrt{N}\begin{bmatrix}\alpha_{\matA}(\w_{0}) & * & \cdots & *\\
\vdots & \vdots &  & \vdots\\
\alpha_{\matA}(\w_{4^{q-1}}) & \vdots &  & \vdots\\*
* & \vdots &  & \vdots\\
\vdots\\*
* & * & \cdots & *
\end{bmatrix}
\end{align*}
where in the last equality we used the identity $N\alpha_{\matA}(\w)=\vec{\sigma_{\w}^{\T}}^{\T}\vec{\matA}$.
Dividing by $\sqrt{N}$ on both sides, we see that\emph{ }$\mathcal{U}_{\tenA_{P}}^{\text{SP}}\in\TS{\tenA_{P}}{\nicefrac{\alpha}{\sqrt{N}}}$
by definition of $\TS{\tenA_{P}}{\nicefrac{\alpha}{\sqrt{N}}}$.

Using \noun{EliminatePermutations} procedure for $\mathcal{{\cal P}}_{(n)}\cdot\mathcal{U}_{\mathrm{comb}}^{(n)}\cdot\bigl({\cal H}^{\otimes n}\otimes{\cal I}_{n}\bigr)\cdot\mathcal{U}_{\mathrm{comb}}^{(n)}\cdot\cirU_{\matA}^{\text{SP}}$
to eliminate $\mathcal{{\cal P}}_{(n)}$, we have that the depth of
$\mathcal{U}_{\tenA_{p}}^{\text{SP}}$ is $d({\cal U}_{\matA}^{\text{SP}})+6$
the gate complexity and the classical cost are $g({\cal U}_{\matA}^{\text{SP}})+6n$
(see Figure~\ref{fig:U_A_to_U_Ap}).
\end{proof}
By reversing the process, we can convert from a state preparation
of $\tenA_{P}$ to a state preparation of $\matA$. To see this, note
that if 
\[
\mathcal{U}_{\tenA_{P}}^{\text{SP}}=\cirU^{(n)}\cdot\cirU_{\matA}^{\text{SP}}
\]
then
\[
\mathcal{U}_{\mathbf{A}}^{\text{SP}}=(\cirU^{(n)})^{\conj}\cdot\mathcal{U}_{\tenA_{P}}^{\text{SP}}
\]
where $(\cirU^{(n)})^{\conj}$ denotes the inverse of the circuit.
The scaling factor is adjusted by $\sqrt{N}$.
\begin{rem}
Seemingly moving from a state preparation of $\matA$ to a state preperation
of $\tenA_{P}$ improves the scale (makes it smaller) by a factor
of $\nicefrac{1}{\sqrt{N}}$. However, the Frobenius norm of $\FNorm{\tenA_{P}}$
is also smaller by a factor of $\nicefrac{1}{\sqrt{N}}$. Since only
the ratio between the scale and the Frobenius norm of the prepared
hypermatrix matters, there is no gain in terms of the scale for moving
from state preparation of $\tenA_{P}$ to state preparation of $\matA$.
Similarly, when moving from $\tenA_{P}$ to $\matA$ the scale ``worsens''
by a factor of $\sqrt{N}$, however, there is no actual loss here
as well since the Frobenius norm increases by the same factor. Moving
one direction and then back around keeps the scale the same, and only
marginally increases the depth. So the two modes of state preparation
are effectively equivalent.
\end{rem}

\section{\label{sec:multiplexing-pauli-coefficient}Multiplexing higher-order
Pauli matrices}

The algorithms we propose for constructing block encodings employ
multiplexers of higher-order Pauli matrices. This section describes
how to construct such multiplexers efficiently.

\subsection{Constructing a multiplexer for arbitrary set of higher-order Pauli
matrices}

In this subsection, our goal is that given a length $K$ sequence
$\w_{0},\dots,\w_{K-1}$ of length $q$ Pauli words (i.e., $\w_{j}\in\Sigma^{q}$),
construct an efficient circuit ${\cal PMX}(\w_{0},\dots,\w_{K-1})\in{\cal MX}({\cal V}_{\w_{0}},\dots,{\cal V}_{\w_{K-1}})$,
where ${\cal V}_{\w}$ implements the higher order Pauli matrix corresponding
to the word $\w$. That is, the circuit ${\cal PMX}(\w_{0},\dots,\w_{K-1})$
is constructed such that

\begin{equation}
\matM({\cal PMX}(\w_{0},\dots,\w_{K-1}))=\bigoplus_{j=0}^{K-1}\sigma_{\mathbf{w}_{j}}.\label{eq:mpx-def}
\end{equation}

Define the rotation matrices: 
\[
\matR_{y}(\theta)\coloneqq\left[\begin{array}{cc}
\cos(\frac{\theta}{2}) & -\sin(\frac{\theta}{2})\\
\sin(\frac{\theta}{2}) & \cos(\frac{\theta}{2})
\end{array}\right],\,\,\matR_{z}(\theta)\coloneqq\left[\begin{array}{cc}
e^{-\si\frac{\theta}{2}} & 0\\
0 & e^{\si\frac{\theta}{2}}
\end{array}\right].
\]
These are the operators that correspond to the ${\cal R}_{y}$ and
${\cal R}_{z}$ rotation gates. The key to our approach is the well
known fact that we can write every Pauli matrix as the product of
specific combination of $\matR_{y}$ and $\matR_{z}$ rotation matrices:
\begin{align}
\matI & =\matR_{y}(0)\matR_{z}(0)\nonumber \\
-\si\matX & =\matR_{y}(\pi)\matR_{z}(\pi)\nonumber \\
-\si\matY & ={\cal \matR}_{y}(\pi)\matR_{z}(0)\nonumber \\
-\si\matZ & ={\cal \matR}_{y}(0)\matR_{z}(\pi)\label{eq:pauli-angels}
\end{align}
Going from Pauli matrices to higher-order Pauli matrices, given a
Pauli string $\w$ we can find a set of angles $\{\theta_{i}^{(y)}\}_{i=0}^{q-1}$
and $\{\theta_{i}^{(z)}\}_{i=0}^{q-1}$, all of them $0$ or $\pi$,
and a phase $\varphi$ such that 
\[
e^{-\si\varphi}\sigma_{\w}=\bigotimes_{i=0}^{q-1}\matR_{y}(\theta_{i}^{(y)})\matR_{z}(\theta_{i}^{(z)})
\]
The phase $\varphi$ corresponds to $\nicefrac{\pi}{2}$ times the
number number of $X,Y,Z$ in $\w$. Thus, given words $\w_{0},\dots,\w_{K-1}$,
there exists angles $\{\theta_{ij}^{(y)}\}$ and $\{\theta_{ij}^{(z)}\}$
(where $i$ runs from $0$ to $q-1$, and $j$ from $0$ to $K-1$),
and phases $\varphi_{0},\dots,\varphi_{K-1}$, such that 
\begin{equation}
\bigoplus_{j=0}^{K-1}e^{-\text{\ensuremath{\si}}\varphi_{j}}\sigma_{\w_{j}}=\bigoplus_{j=0}^{K-1}\bigotimes_{i=0}^{q-1}\matR_{y}(\theta_{ij}^{(y)})\matR_{z}(\theta_{ij}^{(z)})\label{eq:hopauli-mux-angles}
\end{equation}

Both $\matR_{y}$ and $\matR_{z}$ rotation matrices can be implemented
using the corresponding rotation gates, ${\cal R}_{y}$ for $\matR_{y}$
and ${\cal R}_{z}$ for $\matR_{z}$. This allows us to implement
a multiplexer of higher-order Pauli matrices, up to per element scaling,
using a sequence of two 2D multiplexers, first a ${\cal R}_{z}$ multiplexer,
and then a ${\cal R}_{y}$ multiplexers. The key is the following
identity: given matrices $\text{\ensuremath{\matA}}_{ij},\matB_{ij}$
for $i=0,\dots.L-1,j=0,\dots,M-1$ of compatible sizes, we have 
\begin{eqnarray*}
\bigoplus_{i=0}^{L-1}\bigotimes_{j=0}^{M-1}(\matA_{ij}\matB_{ij}) & = & \bigoplus_{i=0}^{L-1}\left[\left(\bigotimes_{j=0}^{M-1}\matA_{ij}\right)\left(\bigotimes_{j=0}^{M-1}\matB_{ij}\right)\right]\\
 & = & \left(\bigoplus_{i=0}^{L-1}\bigotimes_{j=0}^{M-1}\matA_{ij}\right)\left(\bigoplus_{i=0}^{L-1}\bigotimes_{j=0}^{M-1}\matB_{ij}\right)
\end{eqnarray*}
where the first equality follows mixed-product property of Kronecker
products, and the second one is product of block diagonal matrices.
Applying this identity to Eq.~(\ref{eq:hopauli-mux-angles}), we
have 
\begin{equation}
\bigoplus_{j=0}^{K-1}e^{-\text{i}\varphi_{j}}\sigma_{\w_{j}}=\left(\bigoplus_{j=0}^{K-1}\bigotimes_{i=0}^{q-1}\matR_{y}(\theta_{ij}^{(y)})\right)\left(\bigoplus_{j=0}^{K-1}\bigotimes_{i=0}^{q-1}\matR_{z}(\theta_{ij}^{(z)})\right)\label{eq:all-pauli-decompose}
\end{equation}

All angles in the two multiplexers are either $0$ or $\pi$. Finally,
the phases can be corrected by a diagonal gate with diagonal values
\[
\underbrace{\e^{\text{i}\varphi_{0}},\dots,\e^{\text{i}\varphi_{0}}}_{q\text{ times}},\underbrace{\e^{\text{i}\varphi_{1}},\dots,\e^{\text{i}\varphi_{1}}}_{q\text{ times}},\dots,\underbrace{\e^{\text{i}\varphi_{K-1}},\dots,\e^{\text{i}\varphi_{K-1}}}_{q\text{ times}}
\]
Such a diagonal matrix is the Kronecker product of the identity matrix
with the diagonal matrix with distinct phases, so we need only to
implement a diagonal gate for the phases themselves. Thus, the diagonal
gate for correcting the phases, using a generic implementation, has
depth complexity of $O(\nicefrac{K}{k})$ and T-gate and T-depth of
$O(Kk)$ and $O(K)$ respectively.

Using the generic rotation multiplexer construction described in Section
\ref{subsec:rotation-multiplexers} we can construct a multiplexer
for any set of Pauli words of the same length. Since for $K$ Pauli
words of length $q$ we need two 2D multiplexers of size $q$-by-$K$,
the number of single qubit rotations in the multiplexers is $O(qK)$,
number of fanout-CNOTs is $O(K)$, and depth is $O(K)$. Combined
with the costs associated with the diagonal gate, we have total depth
of $O(Kq)$. However, in subsequent sections we make use of multiplexers
for higher order Pauli matrices corresponding to words of a given
length. As we will see in the next subsection, such multiplexers can
be implemented using highly efficient circuits.

\subsection{Multiplexing all higher-order Pauli matrices}

In this section we consider the construction of a multiplexer for
all Pauli words in $\PP_{N}$. The construction in the previous section
can be used to that end, i.e., all we need is to impose some order
on the Pauli words, $\w_{0},\dots,\w_{N^{2}-1}$, and form ${\cal PMX}(\w_{0},\dots,\w_{N^{2}-1})$.
However, we shall show that if we choose a specific order, we can
construct a highly efficient circuit, ${\cal PMX}_{n}$, which implements
a multiplexer of all Pauli words.

Fix $n$, and let $N=2^{n}$. The ordering on $\PP_{N}$ that we use
to implement the multiplexer is the lexicographical ordering on the
Pauli words (most important index is the first, i.e., on the left).
This corresponds to viewing the Pauli words as integers written in
the base-4, with the mapping $0\longleftrightarrow I,1\longleftrightarrow X,2\longleftrightarrow Y,3\longleftrightarrow Z$.

For presentation purposes, it is useful to define an auxiliary matrix
$\matP_{n}\in\Sigma^{n\times N^{2}}$ as follows. Each column in $\matP_{n}$
corresponds to a Pauli word, according to the lexicographical order.
For a column corresponding to word $\w$, we map the individual letters
to rows, with left being on the top. For example, \textbf{
\[
\mathcal{\matP}_{1}=\left[\begin{array}{cccc}
I & X & Y & Z\end{array}\right]
\]
}
\[
\mathcal{\matP}_{2}=\left[\begin{array}{cccccccccccccccc}
I & I & I & I & X & X & X & X & Y & Y & Y & Y & Z & Z & Z & Z\\
I & X & Y & Z & I & X & Y & Z & I & X & Y & Z & I & X & Y & Z
\end{array}\right]
\]
\[
\mathcal{\matP}_{3}=\left[\begin{array}{ccccccccccccccccc}
I & I & I & I & I & I & I & I & I & I & I & I & I & I & I & I & \cdots\\
I & I & I & I & X & X & X & X & Y & Y & Y & Y & Z & Z & Z & Z & \cdots\\
I & X & Y & Z & I & X & Y & Z & I & X & Y & Z & I & X & Y & Z & \cdots
\end{array}\right]
\]
A recursive formula for $\matP_{n}$ is as follows:
\begin{equation}
\mathcal{\matP}_{n}=\left[\begin{array}{cccc}
\underbrace{I\,I\,\cdots\,I}_{4^{n-1}\,\textrm{times}} & \underbrace{X\,X\,\cdots\,X}_{4^{n-1}\,\textrm{times}} & \underbrace{Y\,Y\,\cdots\,Y}_{4^{n-1}\,\textrm{times}} & \underbrace{Z\,Z\,\cdots\,Z}_{4^{n-1}\,\textrm{times}}\\
\matP_{n-1} & \matP_{n-1} & \matP_{n-1} & \matP_{n-1}
\end{array}\right]\label{eq:recursive_pq}
\end{equation}

Assume that $\w_{0},\dots,\w_{N^{2}-1}$ are all the Pauli words,
ordered according to the lexicographical order. We already seen in
the previous section that there exists matrices $\mat{\Theta}_{n}^{(P,y)}=\left[\theta_{ij}^{(P,y)}\right]_{n\times N^{2}}$
and $\mat{\Theta}_{n}^{(P,z)}=\left[\theta_{ij}^{(P,z)}\right]_{n\times N^{2}}$
and angles $\{\varphi_{j}^{(P)}\}$ (the superscript $P$ denotes
the use for multiplexing all Pauli matrices) such that 
\[
\bigoplus_{j=0}^{4^{n}-1}e^{\text{-i}\varphi_{j}^{(P)}}\sigma_{\w_{j}}=\left(\bigoplus_{j=0}^{4^{n}-1}\bigotimes_{i=0}^{n-1}\matR_{y}(\theta_{ij}^{(P,y)})\right)\left(\bigoplus_{j=0}^{4^{n}-1}\bigotimes_{i=0}^{n-1}\matR_{z}(\theta_{ij}^{(P,z)})\right)
\]
The entries matrix $\mat{\Theta}^{(P,y)}$ and $\mat{\Theta}^{(P,z)}$
are easily defined by the corresponding entries of $\matP_{n}$: 
\[
\theta_{ij}^{(P,z)}=\begin{cases}
\pi & (\matP_{n})_{ij}=X\text{ or }Z\\
0 & (\matP_{n})_{ij}=I\text{ or }Y
\end{cases}\quad\quad\theta_{ij}^{(P,y)}=\begin{cases}
\pi & (\matP_{n})_{ij}=X\text{ or }Y\\
0 & (\matP_{n})_{ij}=I\text{ or }Z
\end{cases}
\]
From the recursive formula for $\matP_{n}$ (Eq.~(\ref{eq:recursive_pq})),
we have the following recursive formula for $\mat{\Theta}^{(P,y)}$
and $\mat{\Theta}^{(P,z)}$:
\[
\mat{\Theta}_{1}^{(P,z)}=\left[\begin{array}{cccc}
0 & \pi & 0 & \pi\end{array}\right]\quad\text{and}\quad\Theta_{1}^{(P,y)}=\left[\begin{array}{cccc}
0 & \pi & \pi & 0\end{array}\right]
\]
\[
\mat{\Theta}_{n}^{(P,z)}=\left[\begin{array}{cccc}
\underbrace{0\cdots0}_{4^{n-1}\,\textrm{times}} & \underbrace{\pi\cdots\pi}_{4^{n-1}\,\textrm{times}} & \underbrace{0\cdots0}_{4^{n-1}\,\textrm{times}} & \underbrace{\pi\cdots\pi}_{4^{n-1}\,\textrm{times}}\\
\mat{\Theta}_{n-1}^{(P,z)} & \mat{\Theta}_{n-1}^{(P,z)} & \mat{\Theta}_{n-1}^{(P,z)} & \mat{\Theta}_{n-1}^{(P,z)}
\end{array}\right]
\]

\[
\mat{\Theta}_{n}^{(P,y)}=\left[\begin{array}{cccc}
\underbrace{0\cdots0}_{4^{n-1}\,\textrm{times}} & \underbrace{\pi\cdots\pi}_{4^{n-1}\,\textrm{times}} & \underbrace{\pi\cdots\pi}_{4^{n-1}\,\textrm{times}} & \underbrace{0\cdots0}_{4^{n-1}\,\textrm{times}}\\
\mat{\Theta}_{n-1}^{(P,y)} & \mat{\Theta}_{n-1}^{(P,y)} & \mat{\Theta}_{n-1}^{(P,y)} & \mat{\Theta}_{n-1}^{(P,y)}
\end{array}\right]
\]
For example, 
\[
\mat{\Theta}_{2}^{(P,z)}=\left[\begin{array}{cccccccccccccccc}
0 & 0 & 0 & 0 & \pi & \pi & \pi & \pi & 0 & 0 & 0 & 0 & \pi & \pi & \pi & \pi\\
0 & \pi & 0 & \pi & 0 & \pi & 0 & \pi & 0 & \pi & 0 & \pi & 0 & \pi & 0 & \pi
\end{array}\right]
\]

\[
\mat{\Theta}_{2}^{(P,y)}=\left[\begin{array}{cccccccccccccccc}
0 & 0 & 0 & 0 & \pi & \pi & \pi & \pi & \pi & \pi & \pi & \pi & 0 & 0 & 0 & 0\\
0 & \pi & \pi & 0 & 0 & \pi & \pi & 0 & 0 & \pi & \pi & 0 & 0 & \pi & \pi & 0
\end{array}\right]
\]

\subsubsection{\label{subsec:Efficient-Rotation-Multiplexers}Ultra-sparsity of
$\hat{\protect\mat{\Theta}}_{n}^{(P,y)}$ and $\hat{\protect\mat{\Theta}}_{n}^{(P,z)}$}

Key to an efficient multiplexer for all Pauli matrices is the ultra
sparsity of $\hat{\mat{\Theta}}_{n}^{(P,y)}$ and $\hat{\mat{\Theta}}_{n}^{(P,z)}$.
Recall that when using matrix $\mat{\Theta}$ with $M$ columns in
a rotation multiplexer, we set up an array of rotation gates with
CNOTs between them, where the angles are given by $\hat{\mat{\Theta}}=\frac{1}{\sqrt{M}}\mat{\Theta}\mat H_{M}\matG_{M}$
with $\mat H_{M}$ being the $M\times M$ Walsh-Hadamard matrix, and
$\matG_{M}$ is the permutation matrix that transforms binary ordering
to Gray code ordering (see Fact~\ref{fact:GN-fact}). For multiplexing
all Pauli words, for the two multiplexers we have $\hat{\mat{\Theta}}_{n}^{(P,y)}=\frac{1}{N}\mat{\Theta}_{n}^{(P,y)}\mat H_{N^{2}}\matG_{N^{2}}$
and $\hat{\mat{\Theta}}_{n}^{(P,z)}=\frac{1}{N}\mat{\Theta}_{n}^{(P,z)}\mat H_{N^{2}}\matG_{N^{2}}$
. The following proposition shows that these two matrices are ultra
sparse and contain only $\pm\nicefrac{\pi}{2}$:
\begin{prop}
\label{prop:direct-angels}The following holds: (recall that we use
$0$ based indexing)
\end{prop}

\begin{enumerate}
\item All entries in the first column of $\hat{\mat{\Theta}}_{n}^{(P,y)}$
and $\hat{\mat{\Theta}}_{n}^{(P,z)}$ have value $\nicefrac{\pi}{2}$.
\item For $\hat{\mat{\Theta}}_{n}^{(P,z)}$, each row contains one other
non-zero. At row $i$ it is in column $2\cdot4^{n-i-1}-1$ with value
$-\nicefrac{\pi}{2}$.
\item For $\hat{\mat{\Theta}}_{n}^{(P,y)}$, each row contains one other
non-zero. At row $i$ it is in column $2\cdot4^{n-i-1}$ with value
$-\nicefrac{\pi}{2}$.
\end{enumerate}
\begin{proof}
For the first column, note that the each row in $\mat{\Theta}_{n}^{(P,z)}$
and $\mat{\Theta}_{n}^{(P.y)}$ consists of exactly $N^{2}/2$ entries
non-zero entries, each of them equal to $\pi$. The first column of
$\mat H_{N^{2}/2}\matG_{N^{2}}$ is the all ones vector divided by
$N$, since this the value of the first column in the Walsh-Hadamard
matrix and BRGC does not change the location of the first index. So,
the first column in $\hat{\mat{\Theta}}_{n}^{(P,z)}$ and $\hat{\mat{\Theta}}_{n}^{(P,y)}$
is equal to the sum of entries in the rows of $\mat{\Theta}_{n}^{(P,z)}$
and $\mat{\Theta}_{n}^{(P,y)}$ divided by $N^{2}$. Together, we
see that the value is indeed equal to $\pi/2$.

For other non-zero entries, we prove using induction. The base case
follows from simple calculations:
\begin{align*}
\hat{\mat{\Theta}}_{1}^{(P,z)} & =\frac{1}{\sqrt{4}}\left[\begin{array}{cccc}
0 & \pi & 0 & \pi\end{array}\right]\underbrace{\left(\frac{1}{\sqrt{4}}\left[\begin{array}{cccc}
1 & 1 & 1 & 1\\
1 & -1 & 1 & -1\\
1 & 1 & -1 & -1\\
1 & -1 & -1 & 1
\end{array}\right]\right)}_{\mat H_{4}}\underbrace{\left[\begin{array}{cccc}
1 & 0 & 0 & 0\\
0 & 1 & 0 & 0\\
0 & 0 & 0 & 1\\
0 & 0 & 1 & 0
\end{array}\right]}_{\matG_{4}}\\
 & =\frac{1}{4}\left[\begin{array}{cccc}
0 & \pi & 0 & \pi\end{array}\right]\left[\begin{array}{cccc}
1 & 1 & 1 & 1\\
1 & -1 & -1 & 1\\
1 & 1 & -1 & -1\\
1 & -1 & 1 & -1
\end{array}\right]\\
 & =\left[\begin{array}{cccc}
\nicefrac{\pi}{2} & -\nicefrac{\pi}{2} & 0 & 0\end{array}\right]
\end{align*}
Likewise,
\begin{align*}
\hat{\mat{\Theta}}_{1}^{(P,y)} & =\frac{1}{4}\left[\begin{array}{cccc}
0 & \pi & \pi & 0\end{array}\right]\left[\begin{array}{cccc}
1 & 1 & 1 & 1\\
1 & -1 & -1 & 1\\
1 & 1 & -1 & -1\\
1 & -1 & 1 & -1
\end{array}\right]\\
 & =\left[\begin{array}{cccc}
\nicefrac{\pi}{2} & 0 & -\nicefrac{\pi}{2} & 0\end{array}\right]
\end{align*}

For the inductive step, we use the following well known formula for
$\matG_{M}$ when $M$ is a power of $2$. Let $\matJ_{M}$ denote
the $M$-by-$M$ index reversal matrix, i.e., the $M$-by-$M$ matrix
with $1$ on the anti-diagonal. Then, (see~\cite[Section 5]{oteo2005fractal})
\[
\matG_{M}=\matG_{\nicefrac{M}{2}}\oplus\matG_{\nicefrac{M}{2}}\matJ_{\nicefrac{M}{2}}
\]
Thus, we have
\begin{align*}
\matG_{4^{n}} & =\begin{bmatrix}\matG_{\nicefrac{4^{n}}{2}} & \mat 0\\
\mat 0 & \matG_{\nicefrac{4^{n}}{2}}\matJ_{\nicefrac{4^{n}}{2}}
\end{bmatrix}\\
 & =\left[\begin{array}{cccc}
\matG_{4^{n-1}}\\
 & \matG_{4^{n-1}}\matJ_{4^{n-1}}\\
 &  & \matG_{4^{n-1}}\\
 &  &  & \matG_{4^{n-1}}\matJ_{4^{n-1}}
\end{array}\right]\left[\begin{array}{cccc}
\matI_{4^{n-1}}\\
 & \matI_{4^{n-1}}\\
 &  &  & \matJ_{4^{n-1}}\\
 &  & \matJ_{4^{n-1}}
\end{array}\right]\\
 & =\left[\begin{array}{cccc}
\matG_{4^{n-1}} & \mat 0 & \mat 0 & \mat 0\\
\mat 0 & \matG_{4^{n-1}}\matJ_{4^{n-1}} & \mat 0 & \mat 0\\
\mat 0 & \mat 0 & \mat 0 & \matG_{4^{n-1}}\matJ_{4^{n-1}}\\
\mat 0 & \mat 0 & \matG_{4^{n-1}} & \mat 0
\end{array}\right]
\end{align*}
using the fact that $\matJ_{M}^{2}=\matI_{M}$. The recursive formula
for for $\mat H_{4^{n}}$ implies that:
\[
\mat H_{4^{n}}=\frac{1}{\sqrt{2}}\begin{bmatrix}\mat H_{\nicefrac{4^{n}}{2}} & \mat H_{\nicefrac{4^{n}}{2}}\\
\mat H_{\nicefrac{4^{n}}{2}} & -\mat H_{\nicefrac{4^{n}}{2}}
\end{bmatrix}=\frac{1}{2}\left[\begin{array}{cccc}
\mat H_{4^{n-1}} & \mat H_{4^{n-1}} & \mat H_{4^{n-1}} & \mat H_{4^{n-1}}\\
\mat H_{4^{n-1}} & -\mat H_{4^{n-1}} & \mat H_{4^{n-1}} & -\mat H_{4^{n-1}}\\
\mat H_{4^{n-1}} & \mat H_{4^{n-1}} & -\mat H_{4^{n-1}} & -\mat H_{4^{n-1}}\\
\mat H_{4^{n-1}} & -\mat H_{4^{n-1}} & -\mat H_{4^{n-1}} & \mat H_{4^{n-1}}
\end{array}\right]
\]
So,{\scriptsize
\begin{align*}
\hat{\mat{\Theta}}_{n}^{(P,z)}=\frac{1}{N}\Theta_{n}^{(z)}\mat H_{4^{n}}\matG_{4^{n}} & =\frac{1}{2^{n+1}}\left[\begin{array}{cccc}
\mat 0_{1\times4^{n-1}} & \pi\e_{1\times4^{n-1}} & \mat 0_{1\times4^{n-1}} & \pi\e_{1\times4^{n-1}}\\
\mat{\Theta}_{n-1}^{(P,z)} & \mat{\Theta}_{n-1}^{(P,z)} & \mat{\Theta}_{n-1}^{(P,z)} & \mat{\Theta}_{n-1}^{(P,z)}
\end{array}\right]\\
 & \,\,\,\,\,\,\,\cdot\left[\begin{array}{cccc}
\mat H_{4^{n-1}} & \mat H_{4^{n-1}} & \mat H_{4^{n-1}} & \mat H_{4^{n-1}}\\
\mat H_{4^{n-1}} & -\mat H_{4^{n-1}} & \mat H_{4^{n-1}} & -\mat H_{4^{n-1}}\\
\mat H_{4^{n-1}} & \mat H_{4^{n-1}} & -\mat H_{4^{n-1}} & -\mat H_{4^{n-1}}\\
\mat H_{4^{n-1}} & -\mat H_{4^{n-1}} & -\mat H_{4^{n-1}} & \mat H_{4^{n-1}}
\end{array}\right]\left[\begin{array}{cccc}
\matG_{4^{n-1}} & \mat 0 & \mat 0 & \mat 0\\
\mat 0 & \matG_{4^{n-1}}\matJ_{4^{n-1}} & \mat 0 & \mat 0\\
\mat 0 & \mat 0 & \mat 0 & \matG_{4^{n-1}}\matJ_{4^{n-1}}\\
\mat 0 & \mat 0 & \matG_{4^{n-1}} & \mat 0
\end{array}\right]\\
 & =\frac{1}{2^{n+1}}\left[\begin{array}{cccc}
\mat 0_{1\times4^{n-1}} & \pi\e_{1\times4^{n-1}} & \mat 0_{1\times4^{n-1}} & \pi\e_{1\times4^{n-1}}\\
\mat{\Theta}_{n-1}^{(P,z)} & \mat{\Theta}_{n-1}^{(P,z)} & \mat{\Theta}_{n-1}^{(P,z)} & \mat{\Theta}_{n-1}^{(P,z)}
\end{array}\right]\\
 & \,\,\,\,\,\,\,\cdot\begin{bmatrix}\mat H_{4^{n-1}}\matG_{4^{n-1}} & \mat H_{4^{n-1}}\matG_{4^{n-1}}\matJ_{4^{n-1}} & \mat H_{4^{n-1}}\matG_{4^{n-1}} & \mat H_{4^{n-1}}\matG_{4^{n-1}}\matJ_{4^{n-1}}\\
\mat H_{4^{n-1}}\matG_{4^{n-1}} & -\mat H_{4^{n-1}}\matG_{4^{n-1}}\matJ_{4^{n-1}} & -\mat H_{4^{n-1}}\matG_{4^{n-1}} & \mat H_{4^{n-1}}\matG_{4^{n-1}}\matJ_{4^{n-1}}\\
\mat H_{4^{n-1}}\matG_{4^{n-1}} & \mat H_{4^{n-1}}\matG_{4^{n-1}}\matJ_{4^{n-1}} & -\mat H_{4^{n-1}}\matG_{4^{n-1}} & -\mat H_{4^{n-1}}\matG_{4^{n-1}}\matJ_{4^{n-1}}\\
\mat H_{4^{n-1}}\matG_{4^{n-1}} & -\mat H_{4^{n-1}}\matG_{4^{n-1}}\matJ_{4^{n-1}} & \mat H_{4^{n-1}}\matG_{4^{n-1}} & -\mat H_{4^{n-1}}\matG_{4^{n-1}}\matJ_{4^{n-1}}
\end{bmatrix}\\
 & =\frac{1}{2^{n+1}}\left[\begin{array}{cccc}
2\pi\e_{1\times4^{n-1}}\mat H_{4^{n-1}}\matG_{4^{n-1}} & -2\pi\e_{1\times4^{n-1}}\mat H_{4^{n-1}}\matG_{4^{n-1}}\matJ_{4^{n-1}} & \mat 0_{1\times4^{n-1}} & \mat 0_{1\times4^{n-1}}\\
4\Theta_{n-1}^{(P,z)}\mat H_{4^{n-1}} & \mat 0_{1\times4^{n-1}} & \mat 0_{1\times4^{n-1}} & \mat 0_{1\times4^{n-1}}
\end{array}\right]\\
 & =\frac{1}{2^{n+1}}\left[\begin{array}{cccc}
\underbrace{\begin{array}{cccc}
2\pi\sqrt{4^{n-1}} & 0 & \cdots & 0\end{array}}_{4^{n-1}\,\textrm{length}} & \underbrace{\begin{array}{cccc}
0 & \cdots & 0 & -2\pi\sqrt{4^{n-1}}\end{array}}_{4^{n-1}\,\textrm{length}} & \underbrace{\begin{array}{ccc}
0 & \cdots & 0\end{array}}_{4^{n-1}\,\textrm{length}} & \underbrace{\begin{array}{ccc}
0 & \cdots & 0\end{array}}_{4^{n-1}\,\textrm{length}}\\
4\mat{\Theta}_{n-1}^{(P,z)}\mat H_{4^{n-1}}\matG_{4^{n-1}} & \mat 0 & \mat 0 & \mat 0
\end{array}\right]\\
 & =\left[\begin{array}{cccc}
\underbrace{\begin{array}{cccc}
\nicefrac{\pi}{2} & 0 & \cdots & 0\end{array}}_{4^{n-1}\,\textrm{length}} & \underbrace{\begin{array}{cccc}
0 & \cdots & 0 & -\nicefrac{\pi}{2}\end{array}}_{4^{n-1}\,\textrm{length}} & \underbrace{\begin{array}{ccc}
0 & \cdots & 0\end{array}}_{4^{n-1}\,\textrm{length}} & \underbrace{\begin{array}{ccc}
0 & \cdots & 0\end{array}}_{4^{n-1}\,\textrm{length}}\\
\frac{1}{2^{n-1}}\mat{\Theta}_{q-1}^{(P,z)}\mat H_{4^{n-1}}\matG_{4^{n-1}} & \mat 0 & \mat 0 & \mat 0
\end{array}\right]
\end{align*}
}In the above, we used the fact that for any $M$ which is a power
of 2, we have $\e_{1\times M}\mat H_{M}=\left[\begin{array}{cccc}
\sqrt{M} & 0 & \cdots & 0\end{array}\right]$ (this follows from the fact that the first column of $\sqrt{M}\mat H_{M}$
is the all ones vector, while other columns are balanced in the number
of $+1$ and $-1$). This implies that $\e_{1\times M}\mat H_{M}\matG_{4^{n-1}}=\left[\begin{array}{cccc}
\sqrt{M} & 0 & \cdots & 0\end{array}\right]$ and $\e_{1\times M}\mat H_{M}\matG_{4^{n-1}}\matJ_{4^{n-1}}=\left[\begin{array}{cccc}
0 & \cdots & 0 & \sqrt{M}\end{array}\right]$ . We see that the claim holds for the first row, i.e., index $i=0$.
The claim holds for the rest of the rows by induction, that gives
the structure of $\frac{1}{2^{n-1}}\mat{\Theta}_{n-1}^{(z)}\mat H_{4^{n-1}}\matG_{4^{n-1}}$.

The proof for $\hat{\mat{\Theta}}_{n}^{(P,y)}$ follows the same logic,
as follows:

{\scriptsize
\begin{align*}
\hat{\mat{\Theta}}_{n}^{(P,y)}=\frac{1}{N}\Theta_{n}^{(y)}\mat H_{4^{n}}\matG_{4^{n}} & =\frac{1}{2^{n+1}}\left[\begin{array}{cccc}
\mat 0_{1\times4^{n-1}} & \pi\e_{1\times4^{n-1}} & \pi\e_{1\times4^{n-1}} & \mat 0_{1\times4^{n-1}}\\
\mat{\Theta}_{n-1}^{(P,y)} & \mat{\Theta}_{n-1}^{(P,y)} & \mat{\Theta}_{n-1}^{(P,y)} & \mat{\Theta}_{n-1}^{(P,y)}
\end{array}\right]\\
 & \,\,\,\,\,\,\,\,\,\cdot\begin{bmatrix}\mat H_{4^{n-1}}\matG_{4^{n-1}} & \mat H_{4^{n-1}}\matG_{4^{n-1}}\matJ_{4^{n-1}} & \mat H_{4^{n-1}}\matG_{4^{n-1}} & \mat H_{4^{n-1}}\matG_{4^{n-1}}\matJ_{4^{n-1}}\\
\mat H_{4^{n-1}}\matG_{4^{n-1}} & -\mat H_{4^{n-1}}\matG_{4^{n-1}}\matJ_{4^{n-1}} & -\mat H_{4^{n-1}}\matG_{4^{n-1}} & \mat H_{4^{n-1}}\matG_{4^{n-1}}\matJ_{4^{n-1}}\\
\mat H_{4^{n-1}}\matG_{4^{n-1}} & \mat H_{4^{n-1}}\matG_{4^{n-1}}\matJ_{4^{n-1}} & -\mat H_{4^{n-1}}\matG_{4^{n-1}} & -\mat H_{4^{n-1}}\matG_{4^{n-1}}\matJ_{4^{n-1}}\\
\mat H_{4^{n-1}}\matG_{4^{n-1}} & -\mat H_{4^{n-1}}\matG_{4^{n-1}}\matJ_{4^{n-1}} & \mat H_{4^{n-1}}\matG_{4^{n-1}} & -\mat H_{4^{n-1}}\matG_{4^{n-1}}\matJ_{4^{n-1}}
\end{bmatrix}\\
 & =\frac{1}{2^{n+1}}\left[\begin{array}{cccc}
2\pi\e_{1\times n}\mat H_{4^{n-1}}\matG_{4^{q-1}} & \mat 0_{1\times4^{n-1}} & -2\pi\e_{1\times4^{n-1}}\mat H_{4^{n-1}}\matG_{4^{n-1}} & \mat 0_{1\times4^{n-1}}\\
4\Theta_{n-1}^{(P.y)}\mat H_{4^{n-1}} & \mat 0_{1\times4^{n-1}} & \mat 0_{1\times4^{n-1}} & \mat 0_{1\times4^{n-1}}
\end{array}\right]\\
 & =\frac{1}{2^{n+1}}\left[\begin{array}{cccc}
\underbrace{\begin{array}{cccc}
2\pi\sqrt{4^{n-1}} & 0 & \cdots & 0\end{array}}_{4^{n-1}\,\textrm{length}} & \underbrace{\begin{array}{ccc}
0 & \cdots & 0\end{array}}_{4^{n-1}\,\textrm{length}} & \underbrace{\begin{array}{cccc}
-2\pi\sqrt{4^{n-1}} & 0 & \cdots & 0\end{array}} & \underbrace{\begin{array}{ccc}
0 & \cdots & 0\end{array}}_{4^{n-1}\,\textrm{length}}\\
4\mat{\Theta}_{n-1}^{(P.y)}\mat H_{4^{n-1}}\matG_{4^{n-1}} & \mat 0 & \mat 0 & \mat 0
\end{array}\right]\\
 & =\left[\begin{array}{cccc}
\underbrace{\begin{array}{cccc}
\nicefrac{\pi}{2} & 0 & \cdots & 0\end{array}}_{4^{n-1}\,\textrm{length}} & \underbrace{\begin{array}{ccc}
0 & \cdots & 0\end{array}}_{4^{n-1}\,\textrm{length}} & \underbrace{\begin{array}{cccc}
-\nicefrac{\pi}{2} & 0 & \cdots & 0\end{array}}_{4^{n-1}\,\textrm{length}} & \underbrace{\begin{array}{ccc}
0 & \cdots & 0\end{array}}_{4^{n-1}\,\textrm{length}}\\
\frac{1}{2^{n-1}}\mat{\Theta}_{n-1}^{(P.y)}\mat H_{4^{n-1}}\matG_{4^{n-1}} & \mat 0 & \mat 0 & \mat 0
\end{array}\right]
\end{align*}
}{\scriptsize\par}
\end{proof}

\subsubsection{\label{subsec:Ultra-low-depth}Ultra low depth rotation multiplexers
for $\protect\mat{\Theta}_{n}^{(P,y)}$ and $\protect\mat{\Theta}_{n}^{(P,z)}$}

By Proposition~\ref{prop:direct-angels}, the angles matrices $\hat{\mat{\Theta}}_{n}^{(P,z)}$
and $\hat{\mat{\Theta}}_{n}^{(P,y)}$ are ultra sparse (each row has
only two nonzero entries $\pm\pi/2$). In this section, we leverage
this to construct two families of efficient circuits, $2d{\cal PMX}_{y,n}$
and $2d{\cal PMX}_{z,n}$ (where $n=1,2,\dots$ is a parameter), such
that 
\begin{equation}
2d{\cal PMX}_{y,n}\in{\cal MX}_{y}(\mat{\Theta}_{n}^{(P,y)}),\quad2d{\cal PMX}_{z,n}\in{\cal MX}_{z}(\mat{\Theta}_{n}^{(P,z)})\label{eq:pmx-spec}
\end{equation}

Derivation of $2d{\cal PMX}_{y,n}$ and $2d{\cal PMX}_{z,n}$ is based
on two well known circuit rewriting rules:
\begin{fact}
Consider a circuit ${\cal U}$ which contains either a ${\cal R}_{y}$
or ${\cal R}_{z}$ gate with angle $0$. Then the circuit obtained
by erasing that gate from ${\cal U}$ is equivalent to ${\cal U}$.
\end{fact}

\begin{fact}
\label{fact:cancell-CNOTs}Consider a circuit ${\cal U}$. Suppose
that in $\cirU$ there are two CNOTs with the same target and control,
and in which between them there are only CNOTs with the same target
and different control, or gates that operate on neither the target
or control. Then the circuit $\cirU'$ obtained by removing the two
aforementioned CNOTs from $\cirU$ is equivalent to ${\cal U}$ .
\end{fact}

We now show that repeated application of these two facts starting
with $2d{\cal MR}_{y}(\mat{\Theta}_{n}^{(P,y)})\in{\cal MX}_{y}(\mat{\Theta}_{n}^{(P,y)})$
and $2d{\cal MR}_{z}(\mat{\Theta}_{n}^{(P,z)})\in{\cal MX}_{z}(\mat{\Theta}_{z}^{(P,z)})$
results in the following circuits, which are our proposed $2d{\cal PMX}_{y,n}$
and $2d{\cal PMX}_{z,n}$. The circuits $2d{\cal PMX}_{y,n}$ and
$2d{\cal PMX}_{z,n}$ have $3n$ qubits. Of these, $n$ are target
qubits, and are labeled $t_{0},\dots,t_{n-1}$ (each corresponding
to a row in $\mat{\Theta}_{n}^{(P,\cdot)}$), and the rest are control
qubits labeled $c_{0},\dots,c_{2n-1}$. On each target we operate
with rotation gates, or with CNOTs with one or two specific control
qubits. Each control qubit is connected to exactly one target qubit,
and some control qubits are empty. Thus, we can describe the circuits
by just detailing the gates on each target qubit. These description
is as follows:
\begin{itemize}
\item \textbf{Target qubit $t_{i}$ in $2d{\cal PMX}_{z,n}$} \textbf{,
$i=0,1,\ldots,n-1$:}
\begin{equation}
2d{\cal PMX}_{z,n}\Big|_{t_{i}}={\cal R}_{z,t_{i}}\left(\nicefrac{\pi}{2}\right)\cdot\text{CX}(c_{2i+1},t_{i})\cdot{\cal R}_{z,t_{i}}\left(\nicefrac{-\pi}{2}\right)\cdot\text{CX}(c_{2i+1},t_{i}).\label{eq:rz-row}
\end{equation}
\item \textbf{Target qubit $t_{i}$ in }$2d{\cal PMX}_{y,n}$\textbf{, $i=0,1,\ldots,n-1$:}
\begin{equation}
2d{\cal PMX}_{y,n}\Big|_{t_{i}}={\cal R}_{y,t_{i}}\left(\nicefrac{\pi}{2}\right)\cdot\text{CX}(c_{2i+1},t_{i})\cdot\text{CX}(c_{2i},t)\cdot{\cal R}_{y,t_{i}}\left(\nicefrac{-\pi}{2}\right)\cdot\text{CX}(c_{2i+1},t_{i})\cdot\text{CX}(c_{2i},t_{i}).\label{eq:ry-row}
\end{equation}
\end{itemize}
See Figure~\ref{fig:mux_r} for a graphical illustration. Note that
each data qubit $t_{j}$ is controlled by one or both of $\{c_{2j},c_{2j+1}\}$,
and that the total number of CNOTs is $\#\mathrm{\text{CNOT}}(2d{\cal PMX}_{z,n})=2n$
and $\#\mathrm{\text{CNOT}}(2d{\cal PMX}_{z,n})=4n$. Clearly, the
gate count is $O(n)$ (we have $2n$ rotation angles and $2n$ or
$4n$ CNOT gates for $z$ or $y$, respectively). The circuit depth
is $O(1)$ (each qubit has $2$ rotations and at most $2$ or $4$
CNOT gates, so the critical path is at most $4$ or $6$). The T-count
and T-depth are $0$ (rotations with angles $\pm\nicefrac{\pi}{2}$
are Clifford gates). The classical cost construction is $O(n)$ (no
angle computations are needed since all rotations are $\pm\nicefrac{\pi}{2}$).

We now prove that these constructed circuits indeed implement the
desired multiplexers, i.e., Eqs. (\ref{eq:pmx-spec}) hold. We start
with two auxiliary lemmas. (See Section~\ref{subsec:1d-2d-rotation-multiplexers}
for the definition of ${\cal F}_{m}^{(d)}$.)
\begin{lem}
\label{lem:zero-block} Fix $m$. For any depth $d\in\{0,\dots,m-1\}$,
\[
{\cal F}_{m}^{(d)}(\underbrace{0,\ldots,0}_{2^{m-d}\text{ angles}})\cong\text{CX}(c_{d},t).
\]
\end{lem}

\begin{proof}
By induction on $m-d$. Base case is $m-d=1$, or $d=m-1$, and we
have (Eq.~(\ref{eq:base-f_m}))
\[
{\cal F}_{m}^{(m-1)}(0,0)={\cal R}_{y}(0)\cdot\text{CX}(c_{m-1},t)\cdot{\cal R}_{y}(0)=\text{CX}(c_{m-1},t)
\]
Next, we consider the inductive case, i.e., $m-d$ is increased by
$1$ (so $d$ is decreased by $1$, since we consider $m$ fixed).
The inductive assumption translates to
\[
{\cal F}_{m}^{(d+1)}(\underbrace{0,\ldots,0}_{2^{m-(d+1)}\text{ angles}})\cong\text{CX}(c_{d+1},t)
\]
We have
\begin{align*}
{\cal F}_{m}^{(d)}(\underbrace{0,\ldots,0}_{2^{m-d}\text{ angles}}) & ={\cal F}_{m}^{(d+1)}(0,\dots,0)\cdot\text{CX}(c_{d},t)\cdot{\cal F}_{m}^{(d+1)}(0,\dots,0)\\
 & \cong\text{CX}(c_{d+1},t)\cdot\text{CX}(c_{d},t)\cdot\text{CX}(c_{d+1},t)\\
 & \cong\text{CX}(c_{d},t)
\end{align*}
where we used Eq.~(\ref{eq:recursive-f_m}).
\end{proof}
\begin{lem}
\label{prop:preffix-suffix} Fix $m$. For any depth $d\in\{0,\dots,m-1\}$,
\begin{enumerate}
\item \label{enu:preffix}${\cal F}_{m}^{(d)}\bigl(\phi_{0},\underbrace{0,\ldots,0}_{2^{m-d}-1}\bigr)\cong{\cal R}(\phi_{0})\cdot\text{CX}(c_{d},t)$.
\item \label{enu:suffix}${\cal F}_{m}^{(d)}\bigl(\underbrace{0,\ldots,0}_{2^{m-d}-1},\phi_{2^{m-d}-1}\bigr)\cong\text{CX}(c_{d},t)\cdot{\cal R}(\phi_{2^{m-d}-1}).$
\end{enumerate}
\end{lem}

\begin{proof}
By induction on $m-d$. Base case is $m-d=1$, or $d=m-1$, and we
have (Eq.~(\ref{eq:base-f_m}))
\[
{\cal F}_{m}^{(m-1)}(\phi_{0},\phi_{1})={\cal R}_{y}(\phi_{0})\cdot\text{CX}(c_{m-1},t)\cdot{\cal R}_{y}(\phi_{1})=\begin{cases}
{\cal R}_{y}(\phi_{0})\cdot\text{CX}(c_{m-1},t) & \text{if\,}\phi_{1}=0\\
\text{CX}(c_{m-1},t)\cdot{\cal R}_{y}(\phi_{1}) & \text{if\,}\phi_{0}=0
\end{cases}
\]
Now we consider the inductive case, i.e., $m-d$ is increased by $1$
(so $d$ is decreased by $1$, since we consider $m$ fixed). The
inductive assumption translates to
\begin{align*}
{\cal F}_{m}^{(d+1)}(\phi_{0},0,\ldots,0) & \cong{\cal R}(\phi_{0})\cdot\text{CX}(c_{d+1},t)\\
{\cal F}_{m}^{(d+1)}(0,\ldots,0,\phi_{2^{m-(d+1)}-1}) & \cong\text{CX}(c_{d+1},t)\cdot{\cal R}(\phi_{2^{m-(d+1)}-1})
\end{align*}
We show for the prefix case (first identity). The proof for suffix
case (second identity) follows the same logic. We have,
\begin{align*}
{\cal F}_{m}^{(d)}(\phi_{0},0,\ldots,0) & =\underbrace{{\cal F}_{m}^{(d+1)}(\phi_{0},0,\ldots,0)}_{{\displaystyle {\cal R}(\phi_{0})\cdot\text{CX}(c_{d+1},t)}}\cdot\text{CX}(c_{d},t)\cdot\underbrace{{\cal F}_{m}^{(d+1)}(0,\ldots,0)}_{\text{CNOT}(c_{d+1},t)}\\
 & \cong{\cal R}(\phi_{0})\cdot\text{CX}(c_{d+1},t)\cdot\text{CX}(c_{d},t)\cdot\text{CX}(c_{d+1},t)\\
 & \cong{\cal R}(\phi_{0})\cdot\text{CX}(c_{d},t)
\end{align*}
where we used Eq.~(\ref{eq:recursive-f_m}) and Lemma~\ref{lem:zero-block}.
\end{proof}

\begin{prop}
We have,
\[
2d{\cal PMX}_{y,n}\in{\cal MX}_{y}(\mat{\Theta}_{n}^{(P,y)}),\quad2d{\cal PMX}_{z,n}\in{\cal MX}_{z}(\mat{\Theta}_{n}^{(P,z)})
\]
Equivalently:
\[
\matM(2d{\cal PMX}_{y,n})=\bigoplus_{j=0}^{N^{2}-1}\bigotimes_{i=0}^{n-1}\matR_{y}\bigl(\theta_{ij}^{(P,y)}\bigr),\qquad\matM(2d{\cal PMX}_{z,n})=\bigoplus_{j=0}^{N^{2}-1}\bigotimes_{i=0}^{n-1}\matR_{z}\bigl(\theta_{ij}^{(P,z)}\bigr).
\]
\end{prop}

\begin{proof}
The idea is to start with $2d{\cal MXR}_{y}(\mat{\Theta}_{n}^{(P,y)})$
and $2d{\cal MXR}_{z}(\mat{\Theta}_{n}^{(P,z)})$ and via repeated
applications of Lemmas~\ref{lem:zero-block}~and~\ref{prop:preffix-suffix}
obtain smaller and smaller equivalent circuits, until reaching $2d{\cal PMX}_{y,n}$
and $2d{\cal PMX}_{z,n}$ . Put differently, we show that $2d{\cal PMX}_{y,n}\Big|_{t_{i}}\cong2d{\cal MXR}_{y}(\mat{\Theta}_{n}^{(P,y)})\Big|_{t_{i}}$
and $2d{\cal PMX}_{z,n}\Big|_{t_{i}}\cong2d{\cal MXR}_{z}(\mat{\Theta}_{n}^{(P,z)})\Big|_{t_{i}}$
for $i=0,\dots,n-1$, and this implies the result.

Consider the base case of $n=1$. We have $\hat{\mat{\Theta}}_{1}^{(P,z)}=[+\pi/2,-\pi/2,0,0]$,
and $\hat{\mat{\Theta}}_{1}^{(P,y)}=[+\pi/2,0,-\pi/2,0]$. The fact
that the claim holds is straightforward (see Figure~\ref{fig:recursive-multiplexed}).

Assume that the claim holds for $n$, i.e., $2d{\cal MXR}_{\xi}(\mat{\Theta}_{n}^{(P,\xi)})\Big|_{t_{i}}\cong2d{\cal PMX}_{\xi,n}\Big|_{t_{i}}$
for $i=0,\dots n-1$ (in the equations we use $\xi$ as a placeholder
for either $y$ or $z$). From Eqs.~(\ref{eq:rz-row}) and (\ref{eq:ry-row})
note that when going from $n$ to $n+1$, in order to construct $2d{\cal PMX}_{y,n+1}$
and $2d{\cal PMX}_{z,n+1}$ we can start from $2d{\cal PMX}_{y,n}$
and $2d{\cal PMX}_{z,n}$, shift existing controls by two indices
($c_{i}\to c_{i+2}$), add two new control qubits ($c_{0},c_{1}$),
shift existing targets by one index ($t_{i}\to t_{i+1}$), add a new
target qubit $t_{0}$, and finally we add gates on $t_{0}$ according
to Eqs.~(\ref{eq:rz-row}) and (\ref{eq:ry-row}). Now lets analyze
$2d{\cal MXR}_{y}(\mat{\Theta}_{n+1}^{(P,y)})$ and $2d{\cal MXR}_{z}(\mat{\Theta}_{n+1}^{(P,z)})$.

Keeping in mind that the goal is to show that $2d{\cal MXR}_{\xi}(\mat{\Theta}_{n+1}^{(P,\xi)})\Big|_{t_{i}}\cong2d{\cal PMX}_{\xi,n+1}\Big|_{t_{i}}$for
$i=0,\dots,n$, from Proposition~\ref{prop:direct-angels} we have
the following:

\noindent
\[
\hat{\mat{\Theta}}_{n+1}^{(P,z)}=\begin{bmatrix}\underbrace{\pi/2,0,\cdots,0}_{4^{n}} & \underbrace{0,\cdots,0,-\pi/2}_{4^{n}} & \underbrace{0,\cdots,0}_{4^{n}} & \underbrace{0,\cdots,0}_{4^{n}}\\[2pt]
\hat{\mat{\Theta}}_{n}^{(P,z)} & 0 & 0 & 0
\end{bmatrix},
\]
\[
\hat{\mat{\Theta}}_{n+1}^{(P,y)}=\begin{bmatrix}\underbrace{\pi/2,0,\cdots,0}_{4^{n}} & \underbrace{0,\cdots,0}_{4^{n}} & \underbrace{-\pi/2,0,\cdots,0}_{4^{n}} & \underbrace{0,\cdots,0}_{4^{n}}\\[2pt]
\hat{\mat{\Theta}}_{n}^{(P,y)} & 0 & 0 & 0
\end{bmatrix}.
\]

\noindent Each row in the above matrices correspond to a target qubit,
which correspond to a row in the circuit $2d{\cal MXR}_{\xi}(\mat{\Theta}_{n+1}^{(P,\xi)})$.
There are two types of rows.
\begin{enumerate}
\item ``Old rows'', i.e., rows with index $i\ge1$. These rows already
exist at level $n$. At level $n+1$ they are extended by trailing
zeros. All angles beyond column $4^{n}$ are zero, and corresponding
rotations are identities and theirs gates can be removed. Next, we
argue that $2d{\cal MXR}_{\xi}(\mat{\Theta}_{n+1}^{(P,\xi)})\Big|_{t_{i}}$
equivalent to $2d{\cal MXR}_{y}(\mat{\Theta}_{n}^{(P,\xi)})|_{t_{i-1}}$
with shifted controls $(c_{i-2}\mapsto c_{i})$. To do so, we show
that\textcolor{red}{{} }
\begin{equation}
2d{\cal MXR}_{\xi}(\mat{\Theta}_{n+1}^{(P,\xi)})\Big|_{t_{i}}\cong{\cal F}_{4^{n+1}}^{(2)}\bigl(\phi_{i,0},\ldots,\phi_{i,4^{n}-1}\bigr)\cdot\mathrm{\text{CX}}(c_{2},t_{i}),\label{eq:recursive-2-morton}
\end{equation}
for some angles $\phi_{i,0},\ldots,\phi_{i,4^{n}-1}$ where ${\cal F}_{4^{+1}}^{(2)}(\cdot)$
denotes a row in the recursive multiplexer block from depth $2$ onward,
i.e., row in the circuit $2d{\cal MXR}_{\xi}(\mat{\Theta}_{n}^{(P,\xi)})$
with relabeled wires $(c_{j}\mapsto c_{j+2},\;t_{i-1}\mapsto t_{i})$.
Then, for $i=1,\dots,n$ 
\begin{align*}
{\cal F}_{4^{n+1}}^{(2)}\bigl(\phi_{i,0},\ldots,\phi_{i,4^{n}-1}\bigr)\cdot\text{CX}(c_{2},t_{i}) & \cong2d{\cal MXR}_{\xi}(\mat{\Theta}_{n}^{(P,\xi)})|_{t_{i-1}}
\end{align*}
\\
To see Eq.~(\ref{eq:recursive-2-morton}), expand the recursion for
the block of size $4^{q+1}$ for any $t_{i}$, $i\ge1$ we have:{\small
\begin{align*}
{\cal F}_{4^{n+1}}^{(0)}(\underbrace{\phi_{i,0},\ldots,\phi_{i,4^{n}-1}}_{4^{n}\text{ zeros}},\underbrace{0,0,\ldots,0}_{3\cdot4^{q}\text{ zeros}}) & ={\cal F}_{4^{n+1}}^{(1)}(\underbrace{\phi_{i,0},\ldots,\phi_{i,4^{n}-1}}_{4^{n}\text{ zeros}},\underbrace{0,0,\ldots,0}_{4^{n}\text{ zeros}})\cdot\text{CX}(c_{0},t_{i})\cdot{\cal F}_{4^{n+1}}^{(1)}(\underbrace{0,0,\ldots,0}_{2\cdot4^{n}\text{ zeros}})\\
 & \cong{\cal F}_{4^{n+1}}^{(2)}(\underbrace{\phi_{i,0},\ldots,\phi_{i,4^{n}-1}}_{4^{n}\text{ zeros}})\cdot\text{CX}(c_{1},t_{i})\cdot{\cal F}_{4^{n+1}}^{(2)}(\underbrace{0,0,\ldots,0}_{4^{n}\text{ zeros}})\cdot\text{CX}(c_{0},t_{i})\cdot\text{CX}(c_{1},t_{i})\\
 & \cong{\cal F}_{4^{n+1}}^{(2)}(\underbrace{\phi_{i,0},\ldots,\phi_{i,4^{n}-1}}_{4^{n}\text{ zeros}})\cdot\text{CX}(c_{1},t_{i})\cdot\text{CX}(c_{2},t_{i})\cdot\text{CX}(c_{0},t_{i})\cdot\text{CX}(c_{1},t_{i})\\
 & \cong{\cal F}_{4^{n+1}}^{(2)}(\underbrace{\phi_{i,0},\ldots,\phi_{i,4^{n}-1}}_{4^{q}\text{ zeros}})\cdot\text{CX}(c_{2},t_{i})\cdot\text{CX}(c_{0},t_{i})
\end{align*}
}combining with the full circuit the top level closing CNOT (the final
circuit is $2d{\cal MXR}_{\xi}(\mat{\Theta}_{n+1}^{(P,\xi)})\Big|_{t_{i}}={\cal F}_{4^{n+1}}^{(0)}(\hat{\vtheta})\cdot\text{CX}(c_{0},t_{i})$
for some row $\hat{\vtheta}$){\small
\begin{align*}
{\cal F}_{4^{n+1}}^{(0)}(\underbrace{\phi_{i,0},\ldots,\phi_{i,4^{n}-1}}_{4^{n}\text{ zeros}},\underbrace{0,0,\ldots,0}_{3\cdot4^{n}\text{ zeros}})\cdot\text{CX}(c_{0},t_{i}) & \cong F^{(2)}(\underbrace{\phi_{i,0},\ldots,\phi_{i,4^{n}-1}}_{4^{n}\text{ zeros}})\cdot\text{CX}(c_{2},t_{i})\cdot\text{CX}(c_{0},t_{i})\cdot\text{CX}(c_{0},t_{i})\\
 & \cong{\cal F}_{4^{n+1}}^{(2)}\bigl(\phi_{i,0},\ldots,\phi_{i,4^{n}-1}\bigr)\cdot\text{CX}(c_{2},t_{i})
\end{align*}
}as desired. Now, use the inductive assumption that $2d{\cal MXR}_{\xi}(\mat{\Theta}_{n}^{(P,\xi)})|_{t_{i-1}}\cong2d{\cal PMX}_{\xi,n}\Big|_{t_{i-1}}$
to show equivalence for ``old rows''.
\item A new top row with index $i=0$. For $\hat{\mat{\Theta}}_{n+1}^{(P,z)}$
we have sparse row angles given in Proposition~\ref{prop:direct-angels},
and we need to show that
\begin{align*}
2d{\cal MXR}_{z}(\mat{\Theta}_{n+1}^{(P,z)})\Big|_{t_{0}} & \cong{\cal R}_{z,t_{0}}\left(\nicefrac{\pi}{2}\right)\cdot\text{CX}(c_{1},t_{0})\cdot{\cal R}_{z,t_{0}}\left(\nicefrac{-\pi}{2}\right)\cdot\text{CX}(c_{1},t_{0})\\
2d{\cal MXR}_{y}(\mat{\Theta}_{n+1}^{(P,y)})\Big|_{t_{0}} & \cong{\cal R}_{y,t_{0}}\left(\nicefrac{\pi}{2}\right)\cdot\text{CX}(c_{1},t_{0})\cdot\text{CX}(c_{0},t_{0})\cdot{\cal R}_{y,t_{0}}\left(\nicefrac{-\pi}{2}\right)\cdot\text{CX}(c_{1},t_{0})\cdot\text{CX}(c_{0},t_{0})
\end{align*}
 Using the induction hypothesis for $z$, Eq.~(\ref{eq:recursive-f_m}),
Lemma~\ref{lem:zero-block} and Proposition~\ref{prop:preffix-suffix},
we have{\scriptsize
\begin{align*}
{\cal F}_{4^{n+1}}^{(0)}(\underbrace{\nicefrac{\pi}{2},0,\ldots,0}_{4^{n}\text{ zeros}},\underbrace{0,0,\ldots,-\nicefrac{\pi}{2}}_{4^{n}\text{ zeros}},\underbrace{0,0,\ldots,0}_{2\cdot4^{n}\text{ zeros}}) & ={\cal F}_{4^{n+1}}^{(1)}(\underbrace{\nicefrac{\pi}{2},0,\ldots,0}_{4^{n}},\underbrace{0,\ldots0,-\nicefrac{\pi}{2}}_{4^{n}})\cdot\text{CX}(c_{0},t_{0})\cdot{\cal F}_{4^{n+1}}^{(1)}(\underbrace{0,0,\ldots,0}_{2\cdot4^{n}\text{ zeros}})\\
 & \cong F^{(1)}(\underbrace{\nicefrac{\pi}{2},0,\ldots,0}_{4^{n}})\cdot\text{CX}(c_{1},t_{0})\cdot F^{(1)}(\underbrace{0,\ldots0,-\nicefrac{\pi}{2}}_{4^{n}})\cdot\text{CX}(c_{0},t_{0})\cdot\text{CX}(c_{1},t_{0})\cdot\\
 & \cong{\cal R}_{z,t_{0}}(\nicefrac{\pi}{2})\cdot\text{CX}(c_{1},t_{0})\cdot\text{CX}(c_{1},t_{0})\cdot\text{CX}(c_{1},t_{0})\cdot{\cal R}_{z,t_{0}}\left(\nicefrac{\pi}{2}\right)\cdot\text{CX}(c_{0},t_{0})\cdot\text{CX}(c_{1},t_{0})\\
 & \cong{\cal R}_{z,t_{0}}\left(\nicefrac{\pi}{2}\right)\cdot\text{CX}(c_{1},t_{0})\cdot{\cal R}_{z,t_{0}}\left(\nicefrac{\pi}{2}\right)\cdot\text{CX}(c_{0},t_{0})\cdot\text{CX}(c_{1},t_{0})
\end{align*}
}combining with the full circuit the top level closing CNOT,
\[
{\cal F}_{4^{n+1}}^{(0)}(\underbrace{\nicefrac{\pi}{2},0,\ldots,0}_{4^{n}\text{ zeros}},\underbrace{0,0,\ldots,-\nicefrac{\pi}{2}}_{4^{n}\text{ zeros}},\underbrace{0,0,\ldots,0}_{2\cdot4^{n}\text{ zeros}})\cdot\mathrm{\text{CX}(c_{0},t)}={\cal R}_{z,t_{0}}\left(\nicefrac{\pi}{2}\right)\cdot\text{CX}(c_{1},t_{0})\cdot{\cal R}_{z,t_{0}}\left(\nicefrac{\pi}{2}\right)\cdot\text{CX}(c_{1},t_{0})
\]
The proof the $y$ case follows the same logic. For $\hat{\mat{\Theta}}_{n+1}^{(P,y)}$,
we expand 
\[
{\scriptsize {\normalsize {\cal F}_{4^{n+1}}^{(0)}(\underbrace{\nicefrac{\pi}{2},0,\ldots,0}_{4^{n}\text{ zeros}},\underbrace{0,0,\ldots,0}_{4^{n}\text{ zeros}},\underbrace{-\nicefrac{\pi}{2},0,\ldots,0}_{4^{n}\text{ zeros}},\underbrace{0,0,\ldots,0}_{4^{n}\text{ zeros}})}}
\]
 as follows, {\tiny
\begin{align*}
{\cal F}_{4^{n+1}}^{(0)}(\cdot) & ={\cal F}_{4^{n+1}}^{(1)}(\underbrace{\nicefrac{\pi}{2},0,\ldots,0}_{4^{n}\text{ zeros}},\underbrace{0,0,\ldots,0}_{4^{n}\text{ zeros}})\cdot\text{CX}(c_{0},t_{0})\cdot{\cal F}_{4^{n+1}}^{(1)}(\underbrace{-\nicefrac{\pi}{2},0,\ldots,0}_{4^{n}\text{ zeros}},\underbrace{0,0,\ldots,0}_{4^{n}\text{ zeros}})\\
 & \cong{\cal F}_{4^{n+1}}^{(2)}(\underbrace{\nicefrac{\pi}{2},0,\ldots,0}_{4^{n}\text{ zeros}})\cdot\text{CX}(c_{1},t_{0})\cdot{\cal F}_{4^{n+1}}^{(2)}(\underbrace{0,0,\ldots,0}_{4^{n}\text{ zeros}})\cdot\text{CX}(c_{0},t_{0})\cdot{\cal F}_{4^{n+1}}^{(2)}(\underbrace{-\nicefrac{\pi}{2},0,\ldots,0}_{4^{n}\text{ zeros}})\cdot\text{CX}(c_{1},t_{0})\cdot{\cal F}_{4^{n+1}}^{(2)}(\underbrace{0,0,\ldots,0}_{4^{n}\text{ zeros}})\\
 & \cong{\cal F}_{4^{n+1}}^{(2)}(\underbrace{\nicefrac{\pi}{2},0,\ldots,0}_{4^{n}\text{ zeros}})\cdot\text{CX}(c_{1},t_{0})\cdot\text{CX}(c_{2},t_{0})\cdot\text{CX}(c_{0},t_{0})\cdot{\cal F}_{4^{n+1}}^{(2)}(\underbrace{-\nicefrac{\pi}{2},0,\ldots,0}_{4^{n}\text{ zeros}})\cdot\text{CX}(c_{1},t_{0})\cdot\text{CX}(c_{2},t_{0})\\
 & \cong{\cal R}_{y,t_{0}}\left(\nicefrac{\pi}{2}\right)\cdot\text{CX}(c_{2},t_{0})\cdot\text{CX}(c_{1},t_{0})\cdot\text{CX}(c_{2},t_{0})\cdot\text{CX}(c_{0},t_{0})\cdot{\cal R}_{y,t_{0}}\left(\nicefrac{\pi}{2}\right)\cdot\text{CX}(c_{2},t_{0})\cdot\text{CX}(c_{1},t_{0})\cdot\text{CX}(c_{2},t_{0})\\
 & \cong{\cal R}_{y,t_{0}}\left(\nicefrac{\pi}{2}\right)\cdot\text{CX}(c_{1},t_{0})\cdot\text{CX}(c_{0},t_{0})\cdot{\cal R}_{y,t_{0}}\left(\nicefrac{\pi}{2}\right)\cdot\text{CX}(c_{1},t_{0})\cdot
\end{align*}
}combining with the full circuit the top level closing CNOT we have,
\[
2d{\cal MXR}_{y}^{(n+1)}(\mat{\Theta}^{(y)})\Big|_{t_{0}}={\cal R}_{y,t_{0}}\left(\nicefrac{\pi}{2}\right)\cdot\text{CX}(c_{1},t_{0})\cdot\text{CX}(c_{0},t_{0})\cdot{\cal R}_{y,t_{0}}\left(\nicefrac{\pi}{2}\right)\cdot\text{CX}(c_{1},t_{0})\cdot\text{CX}(c_{0},t_{0})
\]
\end{enumerate}
Combining the above, we have for any $i=0,\dots,n$ that
\[
2d{\cal MXR}_{\xi}(\mat{\Theta}_{n+1}^{(P,\xi)})\Big|_{t_{i}}\cong2d{\cal PMX}_{\xi,n+1}\Big|_{t_{i}}
\]
as desired.
\end{proof}
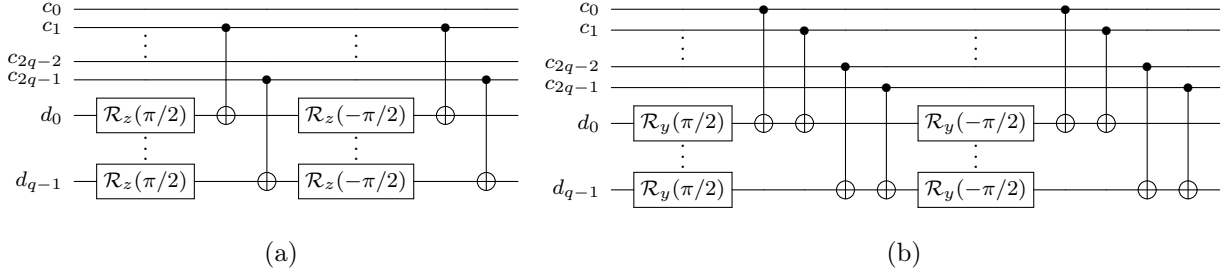
\begin{figure}[t]
\centering{}%
\begin{tabular}{ccc}
{\footnotesize{}
\Qcircuit @C=1em @R=.7em {
& \lstick{c_{0}}    & \qw               		  & \qw      & \qw          & \qw      					     & \qw      & \qw          &  \qw \\
& \lstick{c_{1}}    & \qw               		  & \ctrl{4} & \qw          & \qw      					     & \ctrl{4} & \qw          &  \qw \\
&   			    & \raisebox{0.7em} \vdots     &  		 &      	    & \raisebox{0.7em} \vdots        &  		&      	       &   \\
& \lstick{c_{2q-2}} & \qw              			  & \qw      & \qw          & \qw      					     & \qw      & \qw          &  \qw \\
& \lstick{c_{2q-1}} & \qw               		  & \qw      & \ctrl{3}     & \qw      					     & \qw      & \ctrl{3}     &  \qw \\
& \lstick{d_0}      & \gate{{\cal R}_{z}(\pi/2)}  & \targ    & \qw          & \gate{{\cal R}_{z}(-\pi/2)}  	 & \targ    & \qw          &  \qw \\
&   				& \raisebox{0.7em} \vdots     & 		 &     	        & \raisebox{0.7em} \vdots      	 & 			&     	       &   \\
& \lstick{d_{q-1}}  & \gate{{\cal R}_{z}(\pi/2)}  & \qw      & \targ        & \gate{{\cal R}_{z}(-\pi/2)}  	 & \qw      & \targ        &  \qw \\
}

} & {\footnotesize ~} & {\footnotesize\Qcircuit @C=1em @R=.7em {
& \lstick{c_{0}}    & \qw               		  & \ctrl{5} & \qw      & \qw      & \qw        & \qw      					     & \ctrl{5} & \qw      & \qw      & \qw          &  \qw \\
& \lstick{c_{1}}    & \qw               		  & \qw      & \ctrl{4} & \qw      & \qw        & \qw      					     & \qw      & \ctrl{4} & \qw      & \qw          &  \qw \\
&   			    & \raisebox{0.7em} \vdots     &  		 &  		&      	   &      	    & \raisebox{0.7em} \vdots        &  		&  		   &      	  &      	     &   \\
& \lstick{c_{2q-2}} & \qw              			  & \qw      & \qw      & \ctrl{4} & \qw        & \qw      					     & \qw      & \qw      & \ctrl{4} & \qw          &  \qw \\
& \lstick{c_{2q-1}} & \qw               		  & \qw      & \qw      & \qw      & \ctrl{3}   & \qw      					     & \qw      & \qw      & \qw      & \ctrl{3}     &  \qw \\
& \lstick{d_0}      & \gate{{\cal R}_{y}(\pi/2)}  & \targ    & \targ    & \qw      & \qw        & \gate{{\cal R}_{y}(-\pi/2)}  	 & \targ    & \targ    & \qw      & \qw          &  \qw \\
&   				& \raisebox{0.7em} \vdots     & 		 & 		    &     	   &     	    & \raisebox{0.7em} \vdots      	 & 			& 		   &     	  &     	     &   \\
& \lstick{d_{q-1}}  & \gate{{\cal R}_{y}(\pi/2)}  & \qw      & \qw      & \targ    & \targ      & \gate{{\cal R}_{y}(-\pi/2)}  	 & \qw      & \qw      & \targ    & \targ        &  \qw \\
}
}\tabularnewline
~ &  & ~\tabularnewline
(a) &  & (b)\tabularnewline
\end{tabular}\caption{\label{fig:mux_r}Quantum 2D rotation multiplexers for all Paulis.
(a) Visualization of the $2d{\cal PMX}_{z,n}$. (b) Visualization
of the $2d{\cal PMX}_{z,n}$.}
\end{figure}

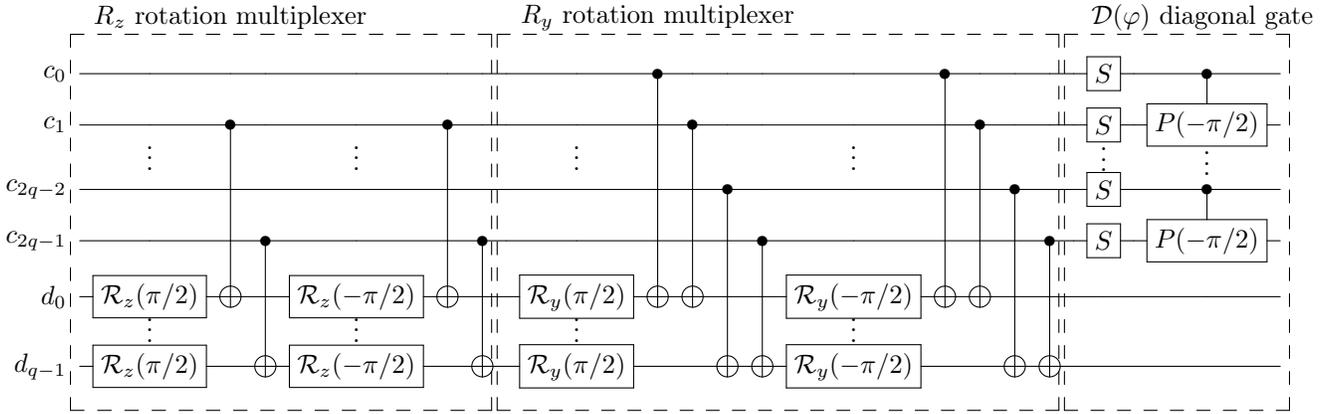
\begin{figure}[t]
\begin{centering}
\[
\Qcircuit @C=0.5em @R=.5em {
& \lstick{}   				&    & 	 \text{$R_{z}$ rotation multiplexer}	 &     	        &       	 & 			&     	       &      &      & 	\text{$R_{y}$ rotation multiplexer}	 & 		    &     	   &     	    &       	 & 			& 		   &     	  &     	     				 & 	 	&		   &	&				\text{${\cal D}(\varphi)$ diagonal gate } 	&	  \\
& \lstick{}   				&    & 	 	 &     	        &       	 & 			&     	       &      &      & 		 & 		    &     	   &     	    &       	 & 			& 		   &     	  &     	     				 & 	 	&		   &	&				 	&	  \\
& \lstick{}   				&    & 	 	 &     	        &       	 & 			&     	       &      &      & 		 & 		    &     	   &     	    &       	 & 			& 		   &     	  &     	     				 & 	 	&		   &	&				 	&	  \\
& \lstick{c_{0}}    & \qw                		  & \qw      & \qw          & \qw      					     & \qw      & \qw           &  \qw & \qw               		  & \ctrl{5} & \qw      & \qw      & \qw        & \qw      					     & \ctrl{5} & \qw      & \qw      & \qw          & 	\qw & \gate{S} & \qw & \ctrl{1}         &	\qw  \gategroup{3}{2}{12}{8}{.7em}{--} \\
& \lstick{c_{1}}    & \qw               		  & \ctrl{4} & \qw          & \qw      					     & \ctrl{4} & \qw          &  \qw & \qw               		  & \qw      & \ctrl{4} & \qw      & \qw        & \qw      					     & \qw      & \ctrl{4} & \qw      & \qw          & 	\qw & \gate{S} & \qw & \gate{P(-\pi/2)} & \qw  \gategroup{3}{9}{12}{19}{.7em}{--} \\
& \lstick{}  			    & \raisebox{0.7em}{\vdots}     &  		 &      	    & \raisebox{0.7em}{\vdots}        &  		&      	       &      & \raisebox{0.7em}{\vdots}     &  		 &  		&      	   &      	    & \raisebox{0.7em}{\vdots}        &  		&  		   &      	  &      	     &     		&	\raisebox{0.6em}{\vdots}	   &	 &	\raisebox{0.3em}{\vdots}				& \\
& \lstick{c_{2q-2}} & \qw              			  & \qw      & \qw          & \qw      					     & \qw      & \qw          &  \qw & \qw              			  & \qw      & \qw      & \ctrl{4} & \qw        & \qw      					     & \qw      & \qw      & \ctrl{4} & \qw          & 	\qw & \gate{S} & \qw & \ctrl{1}         &	 \qw \gategroup{3}{20}{12}{24}{.7em}{--} \\
& \lstick{c_{2q-1}} & \qw               		  & \qw      & \ctrl{3}     & \qw      					     & \qw      & \ctrl{3}     &  \qw & \qw               		  & \qw      & \qw      & \qw      & \ctrl{3}   & \qw      					     & \qw      & \qw      & \qw      & \ctrl{3}     & 	\qw & \gate{S} & \qw & \gate{P(-\pi/2)} &	 \qw  \\
& \lstick{d_0}      & \gate{{\cal R}_{z}(\pi/2)}  & \targ    & \qw          & \gate{{\cal R}_{z}(-\pi/2)}  	 & \targ    & \qw          &  \qw & \gate{{\cal R}_{y}(\pi/2)}  & \targ    & \targ    & \qw      & \qw        & \gate{{\cal R}_{y}(-\pi/2)}  	 & \targ    & \targ    & \qw      & \qw          & 	\qw & \qw	   & \qw & \qw				&			 \qw \\
& \lstick{}  				& \raisebox{0.7em}{\vdots}     & 		 &     	        & \raisebox{0.7em}{\vdots}      	 & 			&     	       &      & \raisebox{0.7em}{\vdots}     & 		 & 		    &     	   &     	    & \raisebox{0.7em}{\vdots}      	 & 			& 		   &     	  &     	     &     	 	&		   &	 &					&   \\
& \lstick{d_{q-1}}  & \gate{{\cal R}_{z}(\pi/2)}  & \qw      & \targ        & \gate{{\cal R}_{z}(-\pi/2)}  	 & \qw      & \targ        &  \qw & \gate{{\cal R}_{y}(\pi/2)}  & \qw      & \qw      & \targ    & \targ      & \gate{{\cal R}_{y}(-\pi/2)}  	 & \qw      & \qw      & \targ    & \targ        & 	\qw & \qw	   & \qw & \qw				&\qw  \\
& \lstick{}   				&     & 		 &     	        &       	 & 			&     	       &      &      & 		 & 		    &     	   &     	    &       	 & 			& 		   &     	  &     	     &     	 	&		   &	 &					&   \\									
}
\]
\par\end{centering}
\caption{\label{fig:pmx_block_diagram}Block circuit diagram of the ${\cal PMX}_{n}$
implementation showing the three main components: (1) ultra-sparse
${\cal R}_{z}$ rotation multiplexer (2) ultra-sparse ${\cal R}_{y}$
rotation multiplexer , and (3) efficient diagonal gate ${\cal D}(\protect\mat{\varphi})$
for phase correction. The dashed boxes highlight each functional block
within the overall circuit architecture.}
\end{figure}

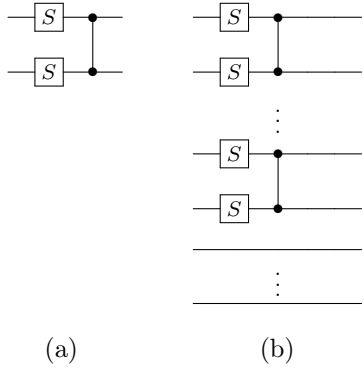
\begin{figure}[tph]
\centering{}%
\begin{tabular}{ccc}
{\footnotesize{}
\Qcircuit @C=1.2em @R=1.2em {
  & \gate{S} & \ctrl{1} & \qw \\
  & \gate{S} & \control \qw & \qw
}
} & {\footnotesize ~} & {\footnotesize\Qcircuit @C=1.2em @R=1.2em {
  & \gate{S} & \ctrl{1} & \qw & \qw & \qw \\
  & \gate{S} & \control \qw & \qw & \qw & \qw \\
  &  & \vdots  &  &  &  \\
  & \gate{S} & \ctrl{1} & \qw & \qw & \qw \\
  & \gate{S} & \control \qw & \qw & \qw & \qw \\
  & \qw      & \qw        & \qw & \qw & \qw \\
  &       &  \vdots       &  &  &  \\
  & \qw      & \qw        & \qw & \qw & \qw
}}\tabularnewline
~ &  & ~\tabularnewline
(a) &  & (b)\tabularnewline
\end{tabular}\caption{\label{fig:D-start}(a) Circuit for $\mathcal{D}^{*}=\mathrm{diag}(1,\protect\si,\protect\si,\protect\si)$.
(b) Circuit for $\mathcal{D}_{(n)}^{*}=(\mathcal{D}^{*})^{\otimes n}\otimes\mathcal{I}_{n}$.
It consists of $n$ stacked copies of the two-qubit diagonal gate
$\mathcal{D}^{*}$ acting on disjoint pairs of qubits, followed by
$n$ lower wires on which the identity acts.}
\end{figure}

\subsubsection{\label{subsec:Multiplexing-all-higher-order}Ultra low depth phase
correction}

Next, we consider how to do the phase correction required for implementing
${\cal PMX}_{n}$ via Eq.~(\ref{eq:all-pauli-decompose}). Recall,
the phases can be corrected by a diagonal gate for diagonal values
\[
\underbrace{\e^{\text{i}\varphi_{0}},\dots,\e^{\text{i}\varphi_{0}}}_{n\text{ times}},\underbrace{\e^{\text{i}\varphi_{1}},\dots,\e^{\text{i}\varphi_{1}}}_{n\text{ times}},\dots,\underbrace{\e^{\text{i}\varphi_{4^{n}-1}},\dots,\e^{\text{i}\varphi_{4^{n}-1}}}_{n\text{ times}}
\]
where $\varphi_{j}$ is equal to $\nicefrac{\pi}{2}$ times number
of $X$,$Y$ and $Z$ in $\w_{j}$. Let $\mathcal{D}_{(q)}^{*}=(\mathcal{D}^{*})^{\otimes n}\otimes{\cal I}_{n}$
where ${\cal D}^{*}$ is the circuit given in Figure~\ref{fig:D-start}
(a). Note that $\matM(\mathcal{D}^{*})=\diag{1,\si,\si,\si}$. In
Figure~\ref{fig:D-start} (b) we have a circuit diagram of $\mathcal{D}_{(n)}^{*}$.
We show that this diagonal gate corrects the phases.

For each index $j$, we have $e^{\si\varphi_{j}}=\si^{n-c(j,q)}$
where $c(j,q)$ is the number of $I$ letters in the Pauli word corresponding
to index $j$. The key observation in the following recursion formula
\begin{equation}
c(j,n)=c(j',n-1)+\mathbf{1}_{r=0}\label{eq:recursive-count}
\end{equation}
where $j=4j'+r$ with $j'=\lfloor j/4\rfloor$ and $r=j\bmod4$, and
$\mathbf{1}_{r=0}=1$ if $r=0$ and $0$ otherwise. Therefore, $c(j,n)=c(\lfloor j/4\rfloor,n-1)+\mathbf{1}_{j\bmod4=0}$.

We now show by induction that
\begin{align}
\matM((\mathcal{D}^{*})^{\otimes n}) & =\diag{\e^{\text{i}\varphi_{0}},\e^{\text{i}\varphi_{1}},\dots,\e^{\text{i}\varphi_{4^{n}-1}}}\nonumber \\
 & =\diag{\si^{n-c(0,n)},\dots,\si^{n-c(4^{n}-1,n)}}\label{eq:pmx-diag}
\end{align}
For $n=1$, the claim follows since $c(0,1)=1$ and $c(1,1)=c(2,1)=c(3,1)=0$,
this equals
\[
\diag{\si^{1-c(0,1)},\si^{1-c(1,1)},\si^{1-c(2,1)},\si^{1-c(3,1)}}.
\]
 As for $n>1$, for any index $j\in\{0,1,\dots,4^{n}-1\}$, we can
write $j=4j'+r$ where $j'=\lfloor j/4\rfloor$ and $r=j\bmod4$.
The $(j,j)$-th diagonal entry of $(\mathcal{D}^{*})^{\otimes n}$
is: 
\[
(j',j')\text{-th entry of }(\mathcal{D}^{*})^{\otimes n-1}\times(r,r)\text{-th entry of }\mathcal{D}^{*}
\]
By the inductive hypothesis and the definition of $\mathcal{D}^{*}$,
we find that the $(j,j)$-th diagonal entry of $\matM((\mathcal{D}^{*})^{\otimes n})$
is 
\[
\si^{n-c(j',n-1)-\mathbf{1}_{r=0}}
\]
By Eq.~(\ref{eq:recursive-count}), $c(j,n)=c(j',n-1)+\mathbf{1}_{r=0}$,
therefore: 
\[
\si^{n-c(j',n-1)-\mathbf{1}_{r=0}}=\si^{n-c(j,n)}
\]
and we conclude by induction that Eq.~(\ref{eq:pmx-diag}) holds.

Now, we have that 
\begin{align*}
\matM(\mathcal{D}_{(n)}^{*})\cdot\bigoplus_{j=0}^{4^{n}-1}e^{\text{-i}\varphi_{j}}\sigma_{\w_{j}} & =(\matM((\mathcal{D}^{*})^{\otimes n})\otimes\matI_{2^{n}})\bigoplus_{j=0}^{4^{n}-1}e^{\text{-i}\varphi_{j}}\sigma_{\w_{j}}\\
 & =\left(\bigoplus_{j=0}^{4^{n}-1}e^{\text{i}\varphi_{j}}\matI_{2^{n}}\right)\cdot\left(\bigoplus_{j=0}^{4^{n}-1}e^{\text{-i}\varphi_{j}}\sigma_{\w_{j}}\right)\\
 & =\bigoplus_{j=0}^{4^{n}-1}\sigma_{\w_{j}}
\end{align*}

\subsubsection{Bringing ${\cal PMX}_{n}$ together}

We now combine the ingredients of the previous subsections. Subsection~\ref{subsec:Ultra-low-depth}
introduced ultra low depth rotation multiplexers for the specific
ultra sparse angles relevant for the all Pauli multiplexer (see Subsection~\ref{subsec:Efficient-Rotation-Multiplexers}).
Subsection~\ref{subsec:Multiplexing-all-higher-order} described
the diagonal correction circuit $\mathcal{D}_{(n)}^{*}$, ensuring
the correct phases. Putting these pieces together on after the other,
the full multiplexer for all Pauli matrices is 
\[
\mathcal{PMX}_{n}\coloneqq\mathcal{D}_{(n)}^{*}\cdot2d{\cal PMX}_{y,n}\cdot2d{\cal PMX}_{z,n}.
\]
A block diagram appears in Figure~\ref{fig:pmx_block_diagram}.

\paragraph{Complexity of ${\cal PMX}_{n}$:}
\begin{itemize}
\item \textbf{Gate count:} $4n$ rotation gates, $6n$ CNOTS, $2n$ S gates,
and $n$ controlled phase gate (non-Clifford gate, T-depth $1$ and
T-count $1$ for phase $-\nicefrac{\pi}{2}$).
\item \textbf{Depth:} $10$.
\item \textbf{T-count:} all come from the controlled phase gates. There
are $n$ such gates, so T-count is $n$.
\item \textbf{T-depth: }all controlled phase gates are in parallel, so T-depth
is $1$.
\end{itemize}
Remarkably, depth is constant regardless of how many matrices are
multiplexed, and the number of gates is logarithmic in the number
multiplexed Pauli matrices.

\section{\label{sec:From-dense-classical}From dense classical matrices to
block encodings circuits}

This section presents algorithms for directly constructing block encoding
circuits from classical matrix representations. Other types of conversions
are discussed in the next section. We address both Hermitian and non-Hermitian
cases and consider matrices specified in either the standard or Pauli
basis. The core methodology is based on the Linear Combination of
Unitaries (LCU) approach and on the building blocks introduced in
the preceding sections.

\subsection{LCU block encodings framework}

A variety of approaches have been proposed in the literature for general
purpose block encoding~\cite{kothari2014efficient,berry2015hamiltonian},
many of them based on the LCU approach. In this subsection, we review
how LCUs can be used to construct block encoding.

To construct a block encoding via LCU a classical matrix $\matA\in\C^{N\times N}$
is expressed as a weighted sum where each summand $\mathbf{U}_{i}\in\C^{N\times N}$
is a unitary operator and the weights are given by a $\boldsymbol{\nu}=(\nu_{0},\dots,\nu_{K-1})$
coefficient vector, where $K$ is a power of $2$, i.e., 
\begin{equation}
\mathbf{A}=\sum_{i=0}^{K-1}\nu_{i}\mathbf{U}_{i},\quad\nu_{i}\in\mathbb{C},\label{eq:lcu-be}
\end{equation}
With such a decomposition, one can construct a block encoding as:
\begin{equation}
\mathcal{U}_{\mathbf{A}}^{\text{BE}}=\Bigl(\cirU_{\sqrt{\boldsymbol{\nu}}}^{\text{SP}}\otimes\mathbf{{\cal I}}_{n}\Bigr)\cdot\mathcal{MX}(\mathcal{U}_{1},\dots,\mathcal{U}_{K-1})\cdot\Bigl((\cirU_{\sqrt{\boldsymbol{\nu}}}^{\text{SP}})^{\T}\otimes\mathbf{{\cal I}}_{n}\Bigr)\label{eq:lcu-be-1}
\end{equation}
where $\cirU_{j}$ is a circuit that implements the unitary $\matU_{j}$
(i.e., $\matM(\cirU_{j})=\matU_{j}$), and $\cirU_{\sqrt{\boldsymbol{\nu}}}^{\text{SP}}\in\SP{\sqrt{\boldsymbol{\nu}}}{\FNorm{\sqrt{\boldsymbol{\nu}}}}$
is a exact state preparation with $\sqrt{\boldsymbol{\nu}}$ denotes
taking entry-wise square-root of $\mat{\nu}$. We have $\mathcal{U}_{\mathbf{A}}^{\text{BE}}\in\BE{\matA}{\FNormS{\sqrt{\boldsymbol{\nu}}}}$.

\begin{figure}[t]
\begin{centering}
\begin{tabular}{ccc}
{\footnotesize\Qcircuit @C=0.3em @R=0.5em {
& \lstick{\ket{0}} & \multigate{2}{\mathcal{U}^{\text{SP}}_{\sqrt{\boldsymbol{\nu}}}} & \multigate{5}{\mathcal{MX}(\mathcal{U}_1,\ldots,\mathcal{U}_K)} & \multigate{2}{(\mathcal{U}_{\sqrt{\boldsymbol{\nu}}}^{\text{SP}})^{-1}} & \qw \\
& \raisebox{0.7em}{\vdots} & \nghost{\mathcal{U}^{\text{SP}}_{\sqrt{\boldsymbol{\nu}}}} & \nghost{\mathcal{MX}(\mathcal{U}_1,\ldots,\mathcal{U}_K)} & \nghost{(\mathcal{U}_{\sqrt{\boldsymbol{\nu}}}^{\text{SP}})^{-1}} & \\
& \lstick{\ket{0}} & \ghost{\mathcal{U}^{\text{SP}}_{\sqrt{\boldsymbol{\nu}}}} & \ghost{\mathcal{MX}(\mathcal{U}_1,\ldots,\mathcal{U}_K)} & \ghost{(\mathcal{U}_{\sqrt{\boldsymbol{\nu}}}^{\text{SP}})^{-1}} & \qw \\
& \lstick{\ket{0}} & \qw & \ghost{\mathcal{MX}(\mathcal{U}_1,\ldots,\mathcal{U}_K)} & \qw & \qw \\
& \raisebox{0.7em}{\vdots} & & \nghost{\mathcal{MX}(\mathcal{U}_1,\ldots,\mathcal{U}_K)} &  & \\
& \lstick{\ket{0}} & \qw & \ghost{\mathcal{MX}(\mathcal{U}_1,\ldots,\mathcal{U}_K)} & \qw & \qw \\
}} &  & {\footnotesize\Qcircuit @C=0.3em @R=0.5em {
& \lstick{\ket{0}} & \multigate{2}{\mathcal{U}^{\text{SP}}_{\sqrt{|\boldsymbol{\nu}|}}} & \multigate{5}{\mathcal{MX}(\mathcal{U}_1,\ldots,\mathcal{U}_K)} & \multigate{2}{\mathcal{D}(\mat{\psi})} & \multigate{2}{(\mathcal{U}_{\sqrt{|\boldsymbol{\nu}|}}^{\text{SP}})^{-1}} & \qw \\
& \raisebox{0.7em}{\vdots} & \nghost{\mathcal{U}^{\text{SP}}_{\sqrt{|\boldsymbol{\nu}|}}} & \nghost{\mathcal{MX}(\mathcal{U}_1,\ldots,\mathcal{U}_K)} & \nghost{\mathcal{D}(\mat{\psi})} & \nghost{(\mathcal{U}_{\sqrt{|\boldsymbol{\nu}|}}^{\text{SP}})^{-1}} & \\
& \lstick{\ket{0}} & \ghost{\mathcal{U}^{\text{SP}}_{\sqrt{|\boldsymbol{\nu}|}}} & \ghost{\mathcal{MX}(\mathcal{U}_1,\ldots,\mathcal{U}_K)} & \ghost{\mathcal{D}(\mat{\psi})} & \ghost{(\mathcal{U}_{\sqrt{|\boldsymbol{\nu}|}}^{\text{SP}})^{-1}} & \qw \\
& \lstick{\ket{0}} & \qw & \ghost{\mathcal{MX}(\mathcal{U}_1,\ldots,\mathcal{U}_K)} & \qw & \qw & \qw \\
& \raisebox{0.7em}{\vdots} & & \nghost{\mathcal{MX}(\mathcal{U}_1,\ldots,\mathcal{U}_K)} &  &  & \\
& \lstick{\ket{0}} & \qw & \ghost{\mathcal{MX}(\mathcal{U}_1,\ldots,\mathcal{U}_K)} & \qw & \qw & \qw \\
}}\tabularnewline
~ &  & ~\tabularnewline
(a) &  & (b)\tabularnewline
 & {\footnotesize{}
\Qcircuit @C=0.3em @R=0.5em {
& \lstick{\ket{0}} & \multigate{5}{\mathcal{U}^{\text{SP}}_{\boldsymbol{\nu}}} & \qw & \qw & \qw \\
& \raisebox{0.7em}{\vdots} & \nghost{\mathcal{U}^{\text{SP}}_{\boldsymbol{\nu}}}  &  &  &  \\
& \lstick{\ket{0}} & \ghost{\mathcal{U}^{\text{SP}}_{\boldsymbol{\nu}}} & \qw & \qw & \qw \\
& \lstick{\ket{0}} & \ghost{\mathcal{U}^{\text{SP}}_{\boldsymbol{\nu}}} & \multigate{5}{\mathcal{MX}(\mathcal{U}_1,\ldots,\mathcal{U}_K)} & \gate{\mathcal{H}} & \qw \\
& \raisebox{0.7em}{\vdots} & \nghost{\mathcal{U}^{\text{SP}}_{\boldsymbol{\nu}}} & \nghost{\mathcal{MX}(\mathcal{U}_1,\ldots,\mathcal{U}_K)} & \raisebox{0.7em}{\vdots} & \\
& \lstick{\ket{0}} & \ghost{\mathcal{U}^{\text{SP}}_{\boldsymbol{\nu}}} & \ghost{\mathcal{MX}(\mathcal{U}_1,\ldots,\mathcal{U}_K)} & \gate{\mathcal{H}} & \qw \\
& \lstick{\ket{0}} & \qw & \ghost{\mathcal{MX}(\mathcal{U}_1,\ldots,\mathcal{U}_K)} & \qw & \qw \\
& \raisebox{0.7em}{\vdots} & & \nghost{\mathcal{MX}(\mathcal{U}_1,\ldots,\mathcal{U}_K)} &  & \\
& \lstick{\ket{0}} & \qw & \ghost{\mathcal{MX}(\mathcal{U}_1,\ldots,\mathcal{U}_K)} & \qw & \qw \\
}} & \tabularnewline
~ & ~ & \tabularnewline
 & (c) & \tabularnewline
\end{tabular}
\par\end{centering}
\caption{\label{fig:lcu-block-encoding}Linear Combination of Unitaries (LCU)
block encoding constructions. (a) Standard LCU block encoding with
complex coefficients using exact state preparation of $\sqrt{\boldsymbol{\nu}}$.
(b) Hermitian block encoding (HBE) with real coefficients using additional
diagonal gate $\mathcal{D}(\protect\vpsi)$ for sign correction. (c)
Alternative construction using full coefficient state preparation
$\mathcal{U}_{\boldsymbol{\nu}}^{\text{SP}}$ followed by Hadamard
gates on \emph{$2k$} coefficients qubits. The circuits in (a) and
(b) use $k$ coefficients qubits (top) and $n$ data qubits (bottom),
while the circuit in (c) accommodates non-exact state preparation
and uses $p$ state preparation ancillary qubits (top), $k$ coefficient
qubits (middle) and $n$ data qubits (bottom).}
\end{figure}
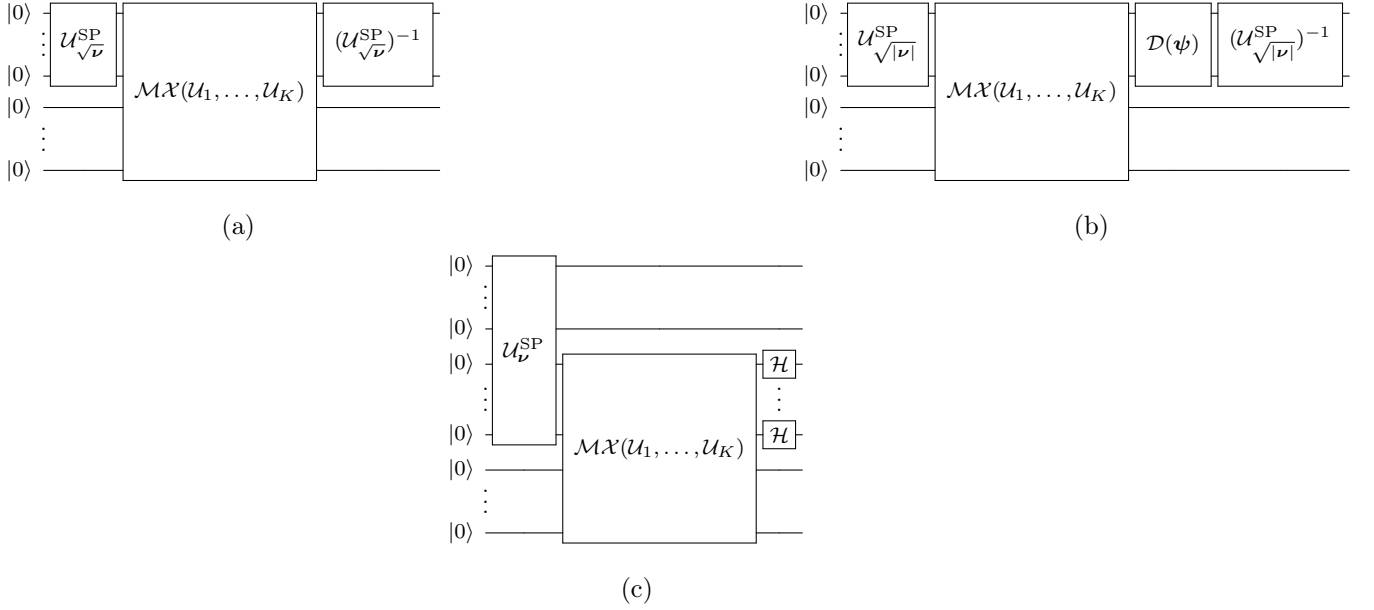

When $\matA$ is Hermitian it is desirable to build a Hermitian block
encoding instead of a regular block encoding. The construction in
Eq.~(\ref{eq:lcu-be-1}) will yield a Hermitian block encoding only
if all coefficients are nonnegative real numbers ($\nu_{i}\in\R_{\geq0}$).
For general real coefficients (possibly negative, i.e., $\nu_{i}\in\R$),
we can construct a Hermitian block encoding by incorporating sign
information through a diagonal gate:

\begin{equation}
\mathcal{U}_{\mathbf{A}}^{\text{HBE}}=\Bigl(\cirU_{\sqrt{|\boldsymbol{\nu}|}}^{\text{SP}}\otimes\mathbf{{\cal I}}_{n}\Bigr)\cdot\mathcal{MX}(\mathcal{U}_{0},\dots,\mathcal{U}_{K-1})\cdot({\cal D}(\vpsi)\otimes\mathbf{{\cal I}}_{n})\cdot\Bigl((\cirU_{\sqrt{|\boldsymbol{\nu}|}}^{\text{SP}})^{\conj}\otimes\mathbf{{\cal I}}_{n}\Bigr)\label{eq:hbe-diag-issue}
\end{equation}
where $\vpsi=\text{sign}(\boldsymbol{\nu})$ and $\cirU_{\sqrt{|\boldsymbol{\nu}|}}^{\text{SP}}\in\SP{\sqrt{|\boldsymbol{\nu}|}}{\FNorm{\sqrt{|\boldsymbol{\nu}|}}}$is
a exact state preparation with $\sqrt{|\boldsymbol{\nu}|}$ denotes
taking entry-wise square-root of the absolute value of $\mat{\nu}$.
We have $\mathcal{U}_{\mathbf{A}}^{\text{HBE}}\in\HBE{\matA}{\FNormS{\boldsymbol{\nu}}}$.
However, this construction involves additional complexity due to the
diagonal gate and must be applied for any sequence of signs (arbitrary
angels), as it depends on the input coefficients yielding a depth
complexity of $O(K)$. Even if we have an highly efficient multiplexer
and state preparation of $\sqrt{\boldsymbol{\nu}}$ , the diagonal
gate might be a bottleneck. We have $\mathcal{U}_{\mathbf{A}}^{\text{HBE}}\in\HBE{\matA}{\OneNorm{\boldsymbol{\nu}}}$.
\begin{rem}
While the LCU method is typically described in the literature with
positive, real coefficients, it is straightforward to extend the decomposition
to include complex coefficients~\cite[Lemma 52]{gilyen2019quantum}.
When complex coefficients are used, the inverse operation (conjugate
transpose) typically employed for the positive, real case is replaced
by the transpose (without the conjugate). This works because the transpose
puts the amplitudes of $\sqrt{\mat{\nu}}$ in the first of the associated
unitary matrix.
\end{rem}

In scenarios where the coefficients are provided via a state preparation
circuit $\cirU_{\boldsymbol{\nu}}^{\text{SP}}\in\SP{\boldsymbol{\nu}}{\alpha}$
for some scale $\alpha$ with $p$ ancillary qubits, rather than as
a classical vector, an analogous construction can be used:
\begin{equation}
\mathcal{U}_{\mathbf{A}}^{\text{BE}}=\Bigl(\cirU_{\boldsymbol{\nu}}^{\text{SP}}\otimes\mathbf{{\cal I}}_{n}\Bigr)\cdot({\cal I}_{p}\otimes\mathcal{MX}(\mathcal{U}_{0},\dots,\mathcal{U}_{K-1}))\cdot\Bigl({\cal I}_{p}\otimes\calH_{k}\otimes\mathbf{{\cal I}}_{n}\Bigr)\label{eq:be-h}
\end{equation}
where $\calH_{k}$ denotes applying $k=\log_{2}K$ Hadamard gates
to different qubits (we assume $K$ is a power of 2). We have $\mathcal{U}_{\mathbf{A}}^{\text{BE}}\in\BE{\matA}{\sqrt{2^{k}}\alpha}$,
so this construction has a worse scale factor.

Figure~\ref{fig:lcu-block-encoding} presents the block diagram implementations
for the three LCU block encoding approaches: (a) the standard construction
with complex coefficients, (b) the Hermitian block encoding with sign
correction, and (c) the alternative Hadamard-based construction.
\begin{rem}
\label{rem:exact-lcu}We note that the constructions in Eqs.~(\ref{eq:hbe-diag-issue})
and (\ref{eq:lcu-be-1}) cannot accommodate inexact state preparation
with ancillary qubits, due to the fact that the ``garbage'' elements
would be accounted for in the final LCU, causing the unitary to fail
to encode $\matA$. However, the Hadamard based construction in Eq.~(\ref{eq:be-h})
is designed to handle this, as the ${\cal I}_{p}\otimes\calH_{k}$
operation prepares a state that effectively zero-pads the subspace
where the garbage terms would appear. This ensures the garbage elements
from the state preparation are multiplied by zero and are therefore
eliminated from the final block-encoding.
\end{rem}

\subsection{\label{subsubsec:PauliMatrix2Be}From Pauli coefficient hypermatrix
to block encoding}

Block encodings of a matrix can be constructed using the matrix's
Pauli decomposition~\cite{dong2021efficient}. Via our highly efficient
multiplexer of all Pauli matrices (Section~\ref{sec:multiplexing-pauli-coefficient})
we obtain an efficient block encoding. Consider a Pauli coefficient
hypermatrix $\mathcal{A}_{P}$ associated with a matrix $\mathbf{A}\in\C^{N\times N}$.
Our algorithms treat $\mathcal{A}_{P}$ as dense. A block encoding
$\mathcal{U}_{\mathbf{A}}^{\text{BE}}$ can be constructed as follows:
\begin{equation}
\mathcal{U}_{\mathbf{A}}^{\text{BE}}=\Bigl(\mathcal{U}_{\sqrt{\vec{\mathcal{A}_{P}}}}^{\text{SP}}\otimes\mathbf{{\cal I}}_{\log N}\Bigr)\cdot\mathcal{PMX}_{n}\cdot\Bigl((\mathcal{U}_{\sqrt{\vec{\mathcal{A}_{P}}}}^{\text{SP}})^{\T}\otimes\mathbf{{\cal I}}_{\log N}\Bigr).\label{eq:uap_to_be}
\end{equation}
This construction works because $\matM(\mathcal{PMX}_{n})$ is a block
diagonal matrix with all higher-order Pauli matrices as diagonal blocks,
and since $\matA$ is the weighted sum of these blocks.

When $\matA$ is Hermitian, all the coefficients are real, and we
can construct an Hermitian block encoding via:

\begin{equation}
\mathcal{U}_{\mathbf{A}}^{\text{HBE}}=\Bigl(\mathcal{U}_{\sqrt{|\vec{\mathcal{A}_{P}}|}}^{\text{SP}}\otimes\mathbf{{\cal I}}_{\log N}\Bigr)\cdot\mathcal{PMX}_{n}\cdot(\mathcal{D}(\text{sgn}(\vec{\mathcal{A}_{P}})\otimes\mathbf{{\cal I}}_{\log N})\cdot\Bigl((\mathcal{U}_{\sqrt{|\vec{\mathcal{A}_{P}}|}}^{\text{SP}})^{\conj}\otimes\mathbf{{\cal I}}_{\log N}\Bigr),\label{eq:uap_to_hbe}
\end{equation}
The scale factor for both block encodings is $\OneNorm{\vec{\tenA_{P}}}$(see
Remark~\ref{rem:exact-lcu}). The classical cost of constructing
the circuit is $O(N^{2}\log N)$ due to the cost of constructing the
state preparation component. The circuit has $3\log N$ qubits, and
the depth is $2d\left(\mathcal{U}_{\sqrt{|\vec{\mathcal{A}_{P}}|}}^{\text{SP}}\right)+d\left(\mathcal{D}(\text{sgn}(\vec{\mathcal{A}_{P}}))\right)+12=O(\nicefrac{N^{2}}{n})$
due to the state preparation circuit. The T-cost is equal to $d_{T}(\mathcal{U}_{\sqrt{|\vec{\mathcal{A}_{P}}|}}^{\text{SP}})+d_{T}(\mathcal{D}(\text{sgn}(\vec{\mathcal{A}_{P}})))+1=O(\frac{N^{2}}{n}\log(N))$.\footnote{See Remark~\ref{rem:depth-tdepth-paradox} for discussion on how
is it possible for the T-depth to be larger than the regular depth.}

Algorithm~\ref{alg:PauliTensor2BE} provides a formal description
for this block encoding construction based on Eqs.~(\ref{eq:uap_to_be})
and (\ref{eq:uap_to_hbe}). Figure~\ref{fig:lcu-block-encoding-pmx}
illustrates the block encoding construction for the both cases.

We can also use Eq.~(\ref{eq:be-h}) to build a block encoding:
\begin{equation}
\mathcal{U}_{\mathbf{A}}^{\text{BE}}=\Bigl(\mathcal{U}_{\vec{{\cal A}_{P}}}^{\text{SP}}\otimes\mathbf{{\cal I}}_{n}\Bigr)\cdot({\cal I}_{p}\otimes\mathcal{PMX}_{n})\cdot\Bigl({\cal I}_{p}\otimes\calH_{2n}\otimes\mathbf{{\cal I}}_{n}\Bigr).\label{eq:uap-to-be}
\end{equation}
 In the above, we can use \emph{any} $\mathcal{U}_{\vec{{\cal A}_{P}}}^{\text{SP}}\in\SP{\vec{{\cal A}_{P}}}{\alpha}$,
and $p$ denotes the number of ancillary qubits in $\mathcal{U}_{\vec{{\cal A}_{P}}}^{\text{SP}}$
(so $\mathcal{U}_{\vec{{\cal A}_{P}}}^{\text{SP}}$ has $2n+p$ qubits).
The classical cost of constructing the circuit is $O(N^{2}\log N)$
due to the cost of constructing the state preparation component. The
circuit has $3n+p$ qubits, and the depth is $d(\mathcal{U}_{\vec{{\cal A}_{P}}}^{\text{SP}})+12$.
The T-cost is equal to the T-cost of the state preparation circuit
plus one. See Figure~\ref{fig:lcu-block-encoding-pmx} (c) for a
block diagram.

The scale of $\mathcal{U}_{\mathbf{A}}^{\text{BE}}$ is $N\alpha$.
When using a exact state preparation for $\vec{{\cal A}_{P}}$ we
have $\alpha=\FNorm{\tenA_{P}}$ and we get a scale factor of $\sqrt{N}\cdot\FNorm{\matA}$
which is always larger or equal to $\OneNorm{\vec{\tenA_{P}}}$.However,
the construction in Eq.~(\ref{eq:uap-to-be}) has an advantage in
that any state preparation for $\vec{\tenA_{P}}$ can be used, and
this can allow us to trade qubits for depth (we see this in the depth
calculation $d(\mathcal{U}_{\vec{{\cal A}_{P}}}^{\text{SP}})+12$).
Indeed, we can utilize various tradeoffs that have been explored in
the literature between number of qubits and depth in state preparation
circuits~\cite{sun2023asymptotically}. This can result in circuits
that have lower depth at the cost of additional qubits. For example,
using constructions similar to the ones described in \cite{clader2022quantum},
we can reduce depth to $O(\log(N))$ if the number of qubits increase
to $O(N^{2})$. We leave detailed exploration of these tradeoffs to
future research.

\begin{algorithm}[t]
\begin{algorithmic}[1]

\STATE  \textbf{Input:} $\tenA_{P}$ of Hermitian matrix $\matA\in\HH_{N}$

~

\STATE  $\vpsi\leftarrow\text{sign}(\vec{\tenA_{P}})$

\STATE  $\mathcal{D}(\vpsi)\leftarrow\text{DiagonalGate}(\vpsi)$
\COMMENT{Construct diagonal gate, e.g. using Qiskit's \texttt{DiagonalGate}.}

\STATE  Construct state preparation circuit $\mathcal{U}_{\sqrt{|\vec{\mathcal{A}_{P}}|}}$

\STATE  $\cirU\leftarrow\Bigl(\mathcal{U}_{\sqrt{|\vec{\mathcal{A}_{P}}|}}\otimes\mathbf{{\cal I}}_{\log N}\Bigr)\cdot\mathcal{PMX}_{n}\cdot(\mathcal{D}(\vpsi)\otimes\mathbf{{\cal I}}_{\log N})\cdot\Bigl(\mathcal{U}_{\sqrt{|\vec{\mathcal{A}_{P}}|}}^{\conj}\otimes\mathbf{{\cal I}}_{\log N}\Bigr)$

\STATE  $\mathcal{U}_{\mathbf{A}}^{\text{HBE}}\leftarrow\textrm{BlockEncoding(U=\ensuremath{\cirU}, shape=\ensuremath{N\times N}, scale=\ensuremath{\OneNorm{\vec{\tenA_{P}}}})}$

\RETURN $\mathcal{U}_{\mathbf{A}}^{\text{HBE}}$

~

\end{algorithmic}

\caption{\label{alg:PauliTensor2BE}\textsc{Pauli2HBE}: Create a Hermitian
block encoding of Hermitian matrix $\protect\matA$ from Pauli coefficients.}
\end{algorithm}

\begin{figure}[t]
\begin{centering}
\begin{tabular}{ccc}
{\footnotesize\Qcircuit @C=0.5em @R=0.5em {
& \lstick{\ket{0}} & \multigate{2}{\mathcal{U}^{\text{SP}}_{\sqrt{\boldsymbol{\nu}}}} & \multigate{5}{\mathcal{PMX}_{n}} & \multigate{2}{(\mathcal{U}_{\sqrt{\boldsymbol{\nu}}}^{\text{SP}})^{-1}} & \qw \\
& \raisebox{0.7em}{\vdots} & \nghost{\mathcal{U}^{\text{SP}}_{\sqrt{\boldsymbol{\nu}}}} & \nghost{\mathcal{PMX}_{n}} & \nghost{(\mathcal{U}_{\sqrt{\boldsymbol{\nu}}}^{\text{SP}})^{-1}} & \\
& \lstick{\ket{0}} & \ghost{\mathcal{U}^{\text{SP}}_{\sqrt{\boldsymbol{\nu}}}} & \ghost{\mathcal{PMX}_{n}} & \ghost{(\mathcal{U}_{\sqrt{\boldsymbol{\nu}}}^{\text{SP}})^{-1}} & \qw \\
& \lstick{\ket{0}} & \qw & \ghost{\mathcal{PMX}_{n}} & \qw & \qw \\
& \raisebox{0.7em}{\vdots} & & \nghost{\mathcal{PMX}_{n}} &  & \\
& \lstick{\ket{0}} & \qw & \ghost{\mathcal{PMX}_{n}} & \qw & \qw \\
}} &  & {\footnotesize\Qcircuit @C=0.5em @R=0.5em {
& \lstick{\ket{0}} & \multigate{2}{\mathcal{U}^{\text{SP}}_{\sqrt{|\boldsymbol{\nu}|}}} & \multigate{5}{\mathcal{PMX}_{n}} & \multigate{2}{\mathcal{D}(\mat{\psi})} & \multigate{2}{(\mathcal{U}_{\sqrt{|\boldsymbol{\nu}|}}^{\text{SP}})^{-1}} & \qw \\
& \raisebox{0.7em}{\vdots} & \nghost{\mathcal{U}^{\text{SP}}_{\sqrt{|\boldsymbol{\nu}|}}} & \nghost{\mathcal{PMX}_{n}} & \nghost{\mathcal{D}(\mat{\psi})} & \nghost{(\mathcal{U}_{\sqrt{|\boldsymbol{\nu}|}}^{\text{SP}})^{-1}} & \\
& \lstick{\ket{0}} & \ghost{\mathcal{U}^{\text{SP}}_{\sqrt{|\boldsymbol{\nu}|}}} & \ghost{\mathcal{PMX}_{n}} & \ghost{\mathcal{D}(\mat{\psi})} & \ghost{(\mathcal{U}_{\sqrt{|\boldsymbol{\nu}|}}^{\text{SP}})^{-1}} & \qw \\
& \lstick{\ket{0}} & \qw & \ghost{\mathcal{PMX}_{n}} & \qw & \qw & \qw \\
& \raisebox{0.7em}{\vdots} & & \nghost{\mathcal{PMX}_{n}} &  &  & \\
& \lstick{\ket{0}} & \qw & \ghost{\mathcal{PMX}_{n}} & \qw & \qw & \qw \\
}}\tabularnewline
~ &  & ~\tabularnewline
(a) &  & (b)\tabularnewline
 & {\footnotesize{}
\Qcircuit @C=0.5em @R=0.5em {
& \lstick{\ket{0}} & \multigate{5}{\mathcal{U}^{\text{SP}}_{\boldsymbol{\nu}}} & \qw & \qw & \qw \\
& \raisebox{0.7em}{\vdots} & \nghost{\mathcal{U}^{\text{SP}}_{\boldsymbol{\nu}}}  &  &  &  \\
& \lstick{\ket{0}} & \ghost{\mathcal{U}^{\text{SP}}_{\boldsymbol{\nu}}} & \qw & \qw & \qw \\
& \lstick{\ket{0}} & \ghost{\mathcal{U}^{\text{SP}}_{\boldsymbol{\nu}}} & \multigate{5}{\mathcal{PMX}_{n}} & \gate{\mathcal{H}} & \qw \\
& \raisebox{0.7em}{\vdots} & \nghost{\mathcal{U}^{\text{SP}}_{\boldsymbol{\nu}}} & \nghost{\mathcal{PMX}_{n}} & \raisebox{0.7em}{\vdots} & \\
& \lstick{\ket{0}} & \ghost{\mathcal{U}^{\text{SP}}_{\boldsymbol{\nu}}} & \ghost{\mathcal{PMX}_{n}} & \gate{\mathcal{H}} & \qw \\
& \lstick{\ket{0}} & \qw & \ghost{\mathcal{PMX}_{n}} & \qw & \qw \\
& \raisebox{0.7em}{\vdots} & & \nghost{\mathcal{PMX}_{n}} &  & \\
& \lstick{\ket{0}} & \qw & \ghost{\mathcal{PMX}_{n}} & \qw & \qw \\
}} & {\footnotesize\begin{tabular}{cc}
     &
    $\begin{aligned}
         &  \\
         &  \\
         &  \\
        \boldsymbol{\nu} & \coloneqq \operatorname{vec}(A_P) \\
        \boldsymbol{\psi} & \coloneqq \operatorname{sgn}(\boldsymbol{\nu})
    \end{aligned}$ \\ 
 \end{tabular}}\tabularnewline
~ & ~ & \tabularnewline
 & (c) & \tabularnewline
\end{tabular}
\par\end{centering}
\caption{\label{fig:lcu-block-encoding-pmx}Pauli-based block encoding constructions.
(a) Standard block encoding with complex coefficients using state
preparation of $\sqrt{\protect\vec{\mathcal{A}_{P}}}$. (b) Hermitian
block encoding (HBE) with real coefficients using additional diagonal
gate $\mathcal{D}(\text{sgn}(\protect\vec{\mathcal{A}_{P}})$ for
sign correction. (c) Alternative construction using full coefficient
state preparation $\mathcal{U}_{\protect\vec{{\cal A}_{P}}}^{\text{SP}}$
followed by Hadamard gates on \emph{$2n$} qubits. The circuits in
(a) and (b) use $2n$ coefficients qubits (top) and $n$ data qubits
(bottom), while the circuit in (c) accommodates non-exact state preparation
and uses $p$ state preparation ancillary qubits (top), $2n$ coefficient
qubits (middle) and $n$ data qubits (bottom).}
\end{figure}
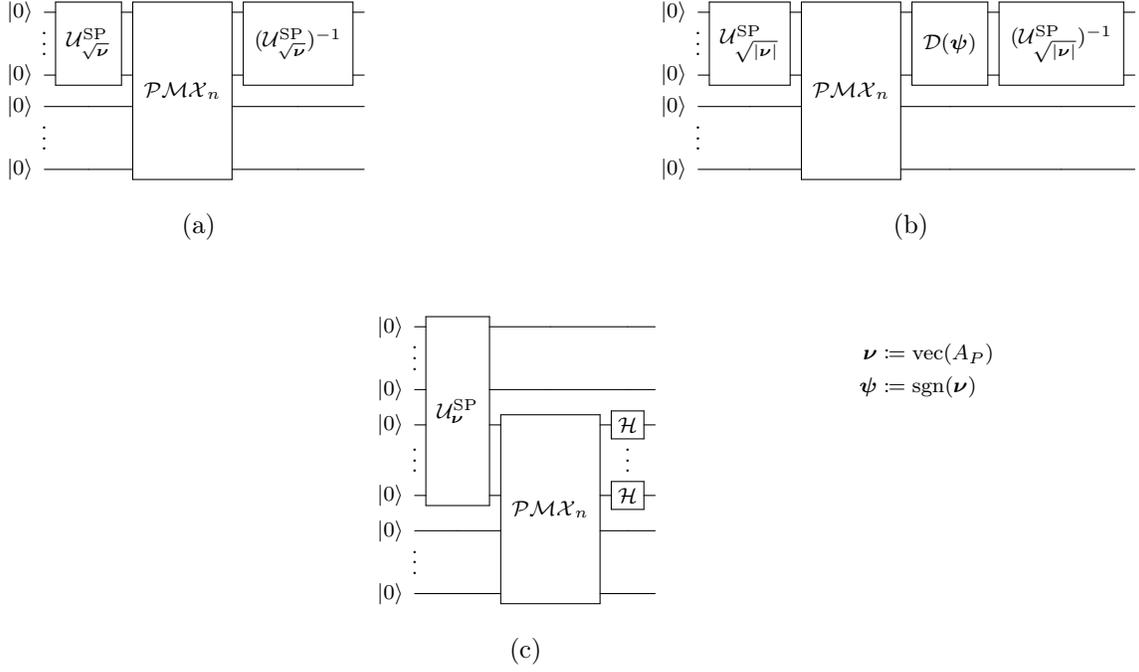

\subsection{\label{subsubsec:Matrix2Be}Constructing a block encoding from a
(dense) classical matrix}

Suppose we are given the matrix $\matA$ explicitly as entries in
classical memory (in the standard basis). A block encoding can be
constructed by converting the matrix to the Pauli basis using Algorithm~\ref{alg:dense2pauli}
(\textsc{Dense2Pauli}), and then proceeding as in the previous section.
We call this algorithm \textsc{Matrix2BE;} see Algorithm~\ref{alg:matrix2BE}
for a pseudo-code description. The classical cost of constructing
the circuit is $O(N^{2}n)$, due to the cost of computing the Pauli
coefficients and of constructing their exact state preparation (same
asymptotic cost for both). A variant of the algorithm, called \noun{Matrix2HBE},
constructs an Hermitian block encoding (for Hermitian matrices), with
similar costs; Algorithm~\ref{alg:matrix2BE} also specifies this
variant. Given such a classical $\matA$, \noun{Matrix2BE} constructs
a $\mathcal{U}_{\mathbf{A}}^{\text{BE}}\in\BE{\matA}{\OneNorm{\vec{\tenA_{P}}}}$,
and \noun{Matrix2HBE} constructs a $\mathcal{U}_{\mathbf{A}}^{\text{BE}}\in\HBE{\matA}{\OneNorm{\vec{\tenA_{P}}}}$.

For both \noun{Matrix2BE} and \noun{Matrix2HBE}, the number of qubits
in the resulting circuit is $3n$, and the depth is $O(\nicefrac{N^{2}}{n})$.
To see that this is the depth for the Hermitian case, from the circuit
description (Eq. (\ref{eq:uap_to_hbe})), we see the cost is $2d(\mathcal{U}_{\mathcal{A}_{P}}^{\text{SP}})+d(\mathcal{D}(\vpsi))+10$
where $10$ is due to the depth of ${\cal PMX}_{n}$. The depth of
$\mathcal{U}_{\mathcal{A}_{P}}^{\text{SP}}$ is $O(\nicefrac{N^{2}}{n})$,
while the diagonal gate can be built with depth $O(\nicefrac{N}{n})$,
so the overall depth is $O(\nicefrac{N^{2}}{n})$. A similar analysis
applies to the non-Hermitian circuit. Asymptotically, the T-depth
is also the same as the T-depth of the state preparation circuit,
which is $O(N^{2})$. 

\paragraph{Other methods for constructing a block-encoding for a dense classical
matrix.}

There is scarce literature on constructing block encodings given a
dense classical matrix $\matA$ (most works focused on the sparse
access model~\cite{ambainis2012variable,childs2017quantum}). Table
\ref{tab:scale_comparison} lists various algorithms, and compares
them to \noun{Matrix2BE} in terms of several parameters.

The FABLE algorithm~\cite{camps2022fable} was the first work to
describe an algorithm for this tasks. It uses a multiplexer structure
that results in a circuit depth of $O(N^{2})$ and requires $2n+1$
qubits. It's scaling factor is $N\max_{i,j}|a_{ij}|$where the $N$
component results from the use of two Hadamard layers. Clader et al.~\cite{clader2022quantum}
studied various tradeoffs for building block encoding using a general
formula from \cite{gilyen2019quantum} which was previously used only
for constructing block encodings in the sparse access model. They
offer two algorithms, one with minimal T-depth and the other with
minimal T-count. The construction is based on circuits implementing
QRAM. While all their methods result in circuits with a scaling factor
of $\FNorm{\matA}$, along with excellent T-complexities for the target
complexity, all his algorithms result in circuits with $\Omega(N)$
qubits (compared to $O(n)$ for all other algorithms, i.e., an exponential
gap). Li et al.~\cite{li2025binary} introduced BITBLE, method for
constructing block encoding quantum circuits with scaling factor of
$\FNorm{\matA}$ and circuit depth $O(N^{2})$~\cite{li2025binary}.
Our method achieves a superior circuit depth compared to other $O(n)$-qubit
algorithms like FABLE and BITBLE: $O(N^{2}/n)$ versus $O(N^{2})$
As for the scale factor, we have
\begin{equation}
\frac{\FNorm{\matA}}{\sqrt{N}}\leq\OneNorm{\vec{\tenA_{P}}}\leq\sqrt{N}\cdot\FNorm{\matA}\label{eq:subnorm-inequal}
\end{equation}
where the lower bound is realized with, for example, the identity
matrix, and the upper bound is realized, for example, when $\tenA_{P}$
is the all-ones tensor. Eq.~(\ref{eq:subnorm-inequal}) shows that
our method obtains a scale factor that is not worse by factor more
than $\sqrt{N}$ than the algorithms in \cite{camps2022fable,clader2022quantum,li2025binary},
but may also be better by at most the same factor. Both \cite{clader2022quantum,li2025binary}
also consider variants of their algorithm whose scaling factor is
the so-called $q$-norm of a matrix. An interesting question is what
is the minimal scale factor achievable with a generic block encoding
methods for classical dense matrices.

\begin{algorithm}[t]
\begin{algorithmic}[1]

\STATE  \textbf{Input:} $\matA\in\mathbb{C}^{N\times N}$

~

\STATE  $\tenA_{P}\leftarrow$ \textsc{Dense2Pauli}($\matA$)

\STATE  $\cirU_{\matA}^{\textrm{BE}}$$\leftarrow$ \textsc{Pauli2BE}($\tenA_{P}$)
(for Hermitian: $\cirU_{\matA}^{\textrm{HBE}}$$\leftarrow$ \textsc{Pauli2HBE}($\tenA_{P}$))

\RETURN $\cirU_{\matA}^{\textrm{BE}}$ (for Hermitian: $\cirU_{\matA}^{\textrm{HBE}}$)

~

\end{algorithmic}

\caption{\label{alg:matrix2BE}Matrix2BE: Create a block encoding from a classical
matrix}
\end{algorithm}

\begin{table}[t]
\centering{}%
\begin{tabular}{|l|l|l|l|l|l|l|}
\hline 
\textbf{Algorithm} & \textbf{\#qubits} & \textbf{Scale factor} & \textbf{Depth} & \textbf{T-Depth} & \textbf{T-Count} & \textbf{Classical Complexity}\tabularnewline
\hline 
\hline 
FABLE~\cite{camps2022fable} & $2n+1$ & $N\max_{i,j}|a_{ij}|$ & $O(N^{2})$ & $O(nN^{2})$ &  & $O(N^{2}n)$\tabularnewline
\hline 
Minimal depth~\cite{clader2022quantum} & $O(N^{2})$ & $\FNorm{\matA}$ &  & $O(n)$ & $O(N^{2})$ & \tabularnewline
\hline 
Minimal count~~\cite{clader2022quantum} & $O(N)$ & $\FNorm{\matA}$ &  & $O(N)$ & $O(N)$ & \tabularnewline
\hline 
BITBLE Algorithm~\cite{li2025binary} & $2n$ & $\FNorm{\matA}$ & $O(N^{2})$ & $O(nN^{2})$ &  & $O(N^{2}n)$\tabularnewline
\hline 
\texttt{\textsc{Matrix2Be (This work)}} & $3n$ & $\OneNorm{\vec{\tenA_{P}}}$ & $O(N^{2}/n)$ & $O(N^{2})$ & $O(nN^{2})$ & $O(N^{2}n)$\tabularnewline
\hline 
\end{tabular}\caption{\label{tab:scale_comparison}Comparison of of various algorithms for
building a block encoding given dense $N\times N$ matrix in classical
memory. Empty entires are ones in which the corresponding paper does
not specify this complexity, and is not easily inferable from the
description.}
\end{table}

\section{\label{sec:algorithms-encoding}Algorithms for converting between
representations}

Consider a dense matrix $\matA\in\C^{N\times N}$ with no additional
structure. Extending beyond the scope of the previous section, we
identify five possible representations of $\matA$: two classical
and three based on quantum circuits\footnote{Representing in a quantum circuit is technically also classical: it
is a classical description of a quantum circuit.}:
\begin{enumerate}
\item The entries of $\matA$ in classical memory (in, e.g., row-major ordering).
\item The entries of $\tenA_{P}$ in classical memory (e.g., $\vec{\tenA_{P}}$).
\item As a matrix state preparation circuit of $\matA$: $\cirU_{\matA}^{\text{SP}}\in\MS{\matA}{\alpha}$
for some $\alpha$.
\item As a hypermatrix state preparation circuit of $\tenA_{P}$: $\cirU_{\tenA_{P}}^{\text{SP}}\in\TS{\tenA_{P}}{\alpha}$
for some $\alpha$.
\item As block encoding circuit of $\matA$: $\cirU_{\matA}^{\text{BE}}\in\BE{\matA}{\alpha}$
for some $\alpha$.
\end{enumerate}
In this section, we discuss algorithms for moving between representations,
both from classical to quantum and from quantum to quantum (of different
types).

In the previous sections, we developed the main tools for encoding
matrices in quantum circuits, and introduced several key algorithms
for specific settings. Since much of the heavy lifting has already
been done, the descriptions in this section are concise and refer
to the methodologies established earlier in the paper. Table~\ref{tab:algorithm-comparison}
summarizes the algorithms, comparing their inputs, outputs, ancillary
qubit requirements, gate complexity, scaling factor, and T-costs.
Figure~\ref{fig:Illustration-of-all-algs} also visually summarizes
the various representations, and algorithm for converting between
them. All algorithms have classical inputs and outputs, producing
quantum circuit descriptions that can be executed on quantum hardware.
Some conversions are already well-established in the literature, and
our goal is simply to summarize the results and describe their computational
complexity as part of a comprehensive description of algorithms to
convert between representations. Other conversions, in particular
the conversion between dense classical to block encoding and the conversion
between state preparation to block encoding, are novel.

\begin{table}[t]
\centering{}{\scriptsize{}%
\begin{tabular}{|l|V{\linewidth}|l|c|c|c|c|}
\hline 
{\scriptsize\textbf{Algorithm}} & {\scriptsize\textbf{Input}} & {\scriptsize\textbf{Output}} & {\scriptsize\textbf{Qubits}} & {\scriptsize\textbf{Depth}} & {\scriptsize\textbf{T-Depth}} & {\scriptsize\textbf{Subsec.}}\tabularnewline
\hline 
\hline 
{\scriptsize\texttt{\textsc{\cite{sun2023asymptotically}}}} & {\scriptsize$\matA\in\mathbb{C}^{N\times N}$} & {\scriptsize$\mathcal{U_{\matA}^{\text{SP}}}\in\MS{\matA}{\FNorm{\matA}}$} & {\scriptsize$2n$} & {\scriptsize$O(\nicefrac{N^{2}}{n})$} & {\scriptsize$O(N^{2})$} & {\scriptsize\ref{subsec:Matrix2Msp_Matrix2PauliMsp}}\tabularnewline
\hline 
{\scriptsize\texttt{\textsc{Matrix2PauliMsp}}} & {\scriptsize$\matA\in\mathbb{C}^{N\times N}$} & {\scriptsize$\mathcal{U_{\tenA_{P}}^{\text{SP}}}\in\TS{\tenA_{P}}{\FNorm{\tenA_{P}}}$} & {\scriptsize$2n$} & {\scriptsize$O(\nicefrac{N^{2}}{n})$} & {\scriptsize$O(N^{2})$} & {\scriptsize\ref{subsec:Matrix2Msp_Matrix2PauliMsp}}\tabularnewline
\hline 
{\scriptsize\texttt{\textsc{Matrix2Be}}} & {\scriptsize$\matA\in\mathbb{C}^{N\times N}$} & {\scriptsize$\mathcal{U_{\matA}^{\text{BE}}}\in\BE{\matA}{\OneNorm{\vec{\tenA_{P}}}}$} & {\scriptsize$3n$} & {\scriptsize$O(\nicefrac{N^{2}}{n})$} & {\scriptsize$O(N^{2})$} & {\scriptsize\ref{subsubsec:Matrix2Be}}\tabularnewline
\hline 
{\scriptsize\texttt{\textsc{Matrix2HBe}}} & \begin{cellvarwidth}[t]
{\scriptsize$\matA\in\mathbb{C}^{N\times N}$, }\\
{\scriptsize Hermitian}
\end{cellvarwidth} & {\scriptsize$\mathcal{U_{\matA}^{\text{BE}}}\in\HBE{\matA}{\OneNorm{\vec{\tenA_{P}}}}$} & {\scriptsize$3n$} & {\scriptsize$O(\nicefrac{N^{2}}{n})$} & {\scriptsize$O(N^{2})$} & {\scriptsize\ref{subsubsec:Matrix2Be}}\tabularnewline
\hline 
\hline 
{\scriptsize\texttt{\textsc{PauliMatrix2Msp}}} & {\scriptsize$\tenA_{P}\in\C^{4\times4\times\cdots\times4}$} & {\scriptsize$\mathcal{U_{\matA}^{\text{SP}}}\in\MS{\matA}{\FNorm{\matA}}$} & {\scriptsize$2n$} & {\scriptsize$O(\nicefrac{N^{2}}{n})$} & {\scriptsize$O(N^{2})$} & {\scriptsize\ref{subsec:PauliMatrix2Msp_PauliMatrix2PauliMsp}}\tabularnewline
\hline 
{\scriptsize\texttt{\textsc{PauliMatrix2PauliMsp}}} & {\scriptsize$\tenA_{P}\in\C^{4\times4\times\cdots\times4}$} & {\scriptsize$\mathcal{U_{\tenA_{P}}^{\text{SP}}}\in\TS{\tenA_{P}}{\FNorm{\tenA_{P}}}$} & {\scriptsize$2n$} & {\scriptsize$O(\nicefrac{N^{2}}{n})$} & {\scriptsize$O(N^{2})$} & {\scriptsize\ref{subsec:PauliMatrix2Msp_PauliMatrix2PauliMsp}}\tabularnewline
\hline 
{\scriptsize\texttt{\textsc{PauliMatrix2Be}}} & {\scriptsize$\tenA_{P}\in\C^{4\times4\times\cdots\times4}$} & {\scriptsize$\mathcal{U_{\matA}^{\text{BE}}}\in\BE{\matA}{\OneNorm{\vec{\tenA_{P}}}}$} & {\scriptsize$3n$} & {\scriptsize$O(\nicefrac{N^{2}}{n})$} & {\scriptsize$O(N^{2})$} & {\scriptsize\ref{subsubsec:PauliMatrix2Be}}\tabularnewline
\hline 
\hline 
{\scriptsize\texttt{\textsc{PauliMsp2Msp}}} & {\scriptsize$\mathcal{U_{\tenA_{P}}^{\text{SP}}}\in\MS{\tenA_{P}}{\alpha}$} & {\scriptsize$\mathcal{U_{\matA}^{\text{SP}}}\in\MS{\matA}{\sqrt{N}\alpha}$} & {\scriptsize$+0$} & {\scriptsize$+6$} & {\scriptsize$+0$} & {\scriptsize\ref{subsec:ua-to-uap}}\tabularnewline
\hline 
{\scriptsize\texttt{\textsc{PauliMsp2Be}}} & {\scriptsize$\mathcal{U_{\tenA_{P}}^{\text{SP}}}\in\MS{\tenA_{P}}{\alpha}$} & {\scriptsize$\mathcal{U_{\matA}^{\text{BE}}}\in\BE{\matA}{N\alpha}$} & {\scriptsize$+n$} & {\scriptsize$+12$} & {\scriptsize$+1$} & {\scriptsize\ref{subsubsec:PauliMsp2Be_Msp2Be}}\tabularnewline
\hline 
\hline 
{\scriptsize\texttt{\textsc{Msp2PauliMsp}}} & {\scriptsize$\mathcal{U_{\matA}^{\text{SP}}}\in\MS{\matA}{\alpha}$} & {\scriptsize$\mathcal{U_{\tenA_{P}}^{\text{SP}}}\in\TS{{\cal A}_{P}}{\nicefrac{\alpha}{\sqrt{N}}}$} & {\scriptsize$+0$} & {\scriptsize$+6$} & {\scriptsize$+0$} & {\scriptsize\ref{subsec:ua-to-uap}}\tabularnewline
\hline 
{\scriptsize\texttt{\textsc{Msp2Be}}} & {\scriptsize$\mathcal{U_{\matA}^{\text{SP}}}\in\MS{\matA}{\alpha}$} & {\scriptsize$\mathcal{U_{\matA}^{\text{BE}}}\in\BE{\matA}{N\alpha}$} & {\scriptsize$+n$} & {\scriptsize$+18$} & {\scriptsize$+1$} & {\scriptsize\ref{subsubsec:PauliMsp2Be_Msp2Be}}\tabularnewline
\hline 
\hline 
{\scriptsize\texttt{\textsc{Be2Msp}}} & {\scriptsize$\mathcal{U_{\matA}^{\text{BE}}}\in\BE{\matA}{\alpha}$} & {\scriptsize$\mathcal{U_{\matA}^{\text{SP}}}\in\MS{\matA}{\sqrt{N}\alpha}$} & {\tiny$+q_{N}$} & {\scriptsize$+2$} & {\scriptsize$+0$} & {\scriptsize\ref{subsubsec:Be2Msp_Be2PauliMsp}}\tabularnewline
\hline 
{\scriptsize\texttt{\textsc{Be2PauliMsp}}} & {\scriptsize$\mathcal{U_{\matA}^{\text{BE}}}\in\BE{\matA}{\alpha}$} & {\scriptsize$\mathcal{U_{\tenA_{P}}^{\text{SP}}}\in\TS{{\cal A}_{P}}{\alpha}$} & {\tiny$+q_{N}$} & {\scriptsize$+8$} & {\scriptsize$+0$} & {\scriptsize\ref{subsubsec:Be2Msp_Be2PauliMsp}}\tabularnewline
\hline 
\end{tabular}}\caption{\label{tab:algorithm-comparison}Summary of encoding algorithms and
their resource requirements. For each algorithm, we list the input,
output, number of \emph{ancillary} qubits, circuit depth, $T$-depth,
and the subsection where the construction is described. When the input
is a circuit (bottom half of the table), columns \textquotedblleft Qubits\textquotedblright ,
\textquotedblleft Depth\textquotedblright{} and \textquotedblleft T-Depth\textquotedblright{}
refer to \emph{surplus} qubits or depth, i.e., qubits/depth on top
of the qubit/depth of the input circuit. In this cases, we prefix
with a plus symbol. We denote $q_{N}\protect\coloneqq q(\mathcal{U_{\protect\matA}^{\text{BE}}})-\log N$.}
\end{table}

\begin{figure}[t]
\begin{centering}
\begin{tabular}{c}
\includegraphics[scale=0.4]{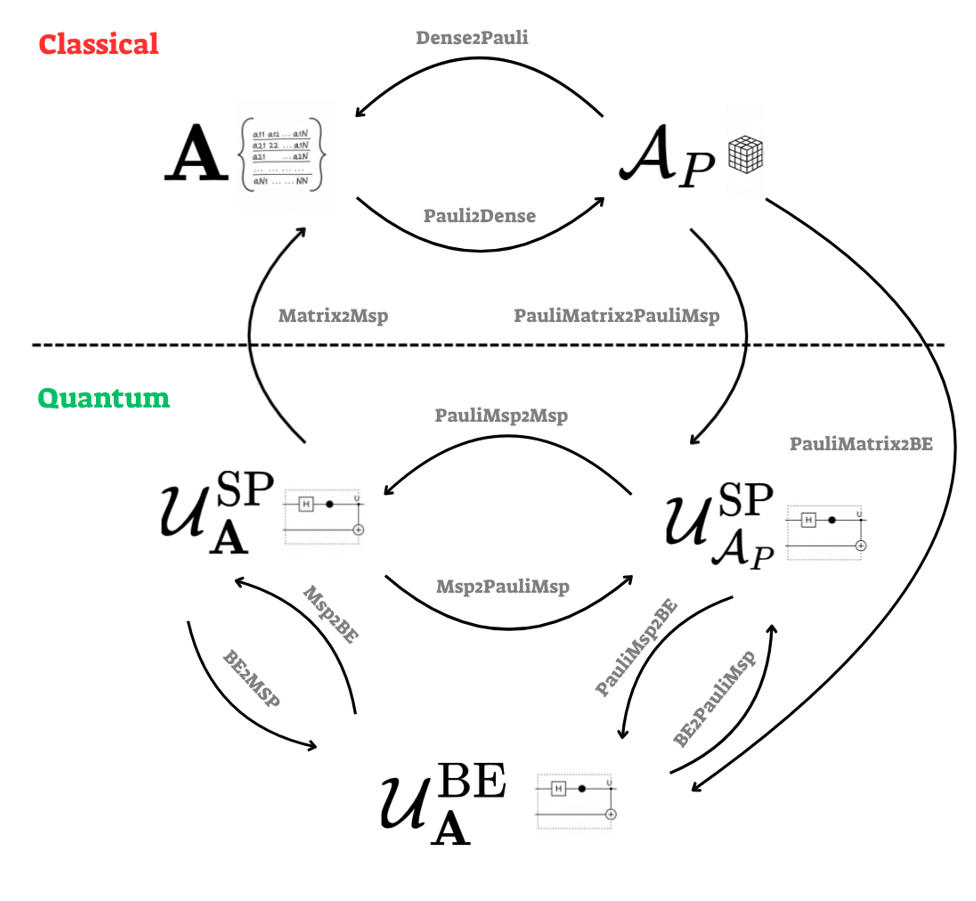}\tabularnewline
\end{tabular}
\par\end{centering}
\caption{\label{fig:Illustration-of-all-algs}Illustration of the paper\textquoteright s
algorithms. Top: classical representations $\protect\matA$ and $\protect\tenA_{P}$
. Bottom: quantum circuit representations ($\protect\cirU_{\protect\matA}^{\text{SP}}$,
$\protect\cirU_{\protect\tenA_{P}}^{\text{SP}},$ and $\protect\cirU_{\protect\matA}^{\text{BE}}$)}
\end{figure}

\subsection{\label{subsec:classical-to-msp}From classical to matrix state preparation}

\subsubsection{\label{subsec:Matrix2Msp_Matrix2PauliMsp}From coefficients in the
standard basis}

The conversion from a classical matrix representation in the standard
basis to a matrix state preparation circuit is a well-established
procedure in quantum computing literature~\cite{mottonen2004transformation,shende2005synthesis,plesch2011quantum}.
Given a classical matrix $\mathbf{A}\in\mathbb{C}^{N\times N}$, we
can construct a matrix state preparation circuit $\mathcal{U}_{\matA}^{\text{SP}}\in\MS{\matA}{\FNorm{\matA}}$
by vectorizing the matrix according to our definition in Section~\ref{sec:preliminaries},
and using standard state preparation techniques to create a circuit
that prepares $\kket{\matA}$.

Seminal works that describe the process of synthesizing arbitrary
state preparation circuits are described in~\cite{mottonen2004transformation,shende2005synthesis}.
These methods are based on rotation multiplexers (see Section~\ref{subsec:rotation-multiplexers}),
require $O(N^{2}\log N)$ classical cost, produce circuits of depth
and T-depth of $O(N^{2})$, and use $O(N^{2})$ gates. Various improvements
were described over the years. The isometry method of \cite{iten2016quantum}
improves constant factors. Plesch and Brukner~\cite{plesch2011quantum}
showed how to reduce costs, specifically by improving the constant
factors in the CNOT gate count and circuit depth for arbitrary state
preparation using universal gate decompositions. Zhang et al. showed
how depth and qubit count can be traded \cite{zhang2022quantum}.
They show that $O(\log N)$ depth suffice if we allow the number of
ancilla qubits to be $O(N)$. They also discuss preparing sparse vectors
with $D$ non zeroes entries, and prove near optimality of their costs
(their proposed algorithm for sparse states achieves an exponentially
faster depth of $O(\log ND)$ with $O(ND\log D)$ ancillary qubits).
More recently, Sun et al.~~\cite{sun2023asymptotically} proved
an asymptotically optimal bound for general state preparation without
ancillary qubits, achieving circuit depth $O(N^{2}/n)$ and T-depth~\cite[Table 1]{low2024trading}
of $O(N^{2})$.

The conversion from a classical matrix representation in the standard
basis to a hypermatrix state preparation circuit of $\tenA_{P}$ can
be achieved through two methods. The straightforward way is to classically
compute all coefficients of the Pauli hypermatrix $\tenA_{P}$ directly
from the classical matrix $\mathbf{A}$, and then create the hypermatrix
state preparation circuit $\cirU_{\tenA_{P}}^{\text{SP}}$ using the
same procedure described above. (This is essentially the same procedure
as standard state preparation techniques once the vectorized form
of the hypermatrix is obtained.) The second method is to build a matrix
state preparation circuit for $\mathbf{A}$, and then use the change
of basis techniques that we described in Subsection~\ref{subsec:ua-to-uap}.
Both methods have the same asymptotic cost (both classical and quantum).

\subsubsection{\label{subsec:PauliMatrix2Msp_PauliMatrix2PauliMsp}From coefficients
in the Pauli basis}

Given a hypermatrix $\tenA_{P}$ representing the Pauli decomposition
of matrix $\mathbf{A}$, we can create the matrix state preparation
circuit $\mathcal{U}_{\mathbf{A}}^{\text{SP}}$ through two approaches.
The first approach reconstructs the standard basis matrix $\mathbf{A}$
from $\tenA_{P}$ using the recursive Pauli decomposition relationship
in Eq.~(\ref{eq:block-pauli}), and then applies the standard matrix
state preparation techniques described in the previous subsection.
The second approach directly constructs the hypermatrix state preparation
circuit $\mathcal{U}_{\tenA_{P}}^{\text{SP}}$ and then applies a
change of basis transformation to obtain $\mathcal{U}_{\mathbf{A}}^{\text{SP}}$.
The change of basis from $\mathcal{U}_{\mathbf{A}}^{\text{SP}}$ to
$\mathcal{U}_{\tenA_{P}}^{\text{SP}}$ is described in Subsection~\ref{subsec:ua-to-uap}.
Both methods have the same asymptotic cost (both classical and quantum).

\subsection{\label{subsec:interplay-msp-be}The interplay between matrix state
preparation and block encoding}

In this subsection, we explore the relationship between the two circuit
representations of matrices discussed in this paper: matrix state
preparation circuits and block encodings. A fundamental question is
whether these two representations are equivalent, where we define
equivalence as the ability to efficiently construct an efficient circuit
for one representation given a circuit for the other representation
with (roughly) the same complexities.

The conversion from a block encoding representation to a matrix state
preparation representation was partially discussed in our previous
work~\cite[Section 3.2.2]{yosef2024multivariate}. It was shown that
an exact block encoding can be converted to a matrix state preparation
by attaching to it a very simple circuit. Here we complete the discussion,
and show how convert from an arbitrary block encoding to a matrix
state preparation. We also discuss the opposite conversion, from a
state preparation to a block encoding. This conversion is much less
straightforward, and requires a more elaborate circuit.

A surprising aspect of the results in this section is that an efficient
bidirectional conversion between the two representations is possible.
In particular, the conversion requires a surplus depth of $O(1)$,
and the surplus T-depth is exactly 1. These costs are typically much
lower than the ones required for the circuit themselves. However,
the conversions do increase the number of qubits and worsen the scale
of the representation.

We highlight one interesting consequence of the efficient bidirectional
conversion. A fundamental tool in many quantum numerical linear algebra
is the Quantum Singular Value Transformation (QSVT) algorithm \cite{gilyen2019quantum}.
However, this algorithm has only been defined for block encodings.
Due to the results of this section, it can also be applied to matrix
state preparation circuits.

\subsubsection{\label{subsubsec:Be2Msp_Be2PauliMsp}From block encoding to matrix
state preparation}

In this subsection, the goal is to construct $\mathcal{U}_{\matA}^{\textrm{SP}}$
and $\mathcal{U}_{\tenA_{P}}^{\text{SP}}$ starting from $\mathcal{U_{\matA}^{\text{BE}}}$.
Due to the construction $\mathcal{U}_{\tenA_{P}}^{\text{SP}}=\cirU^{(q)}\cdot\cirU_{\matA}^{\text{SP}}$
in Subsection~\ref{subsec:ua-to-uap}, it suffices to establish the
construction of $\mathcal{U}_{\matA}^{\textrm{SP}}$ from $\mathcal{U_{\matA}^{\text{BE}}}$.

Let $\matA\in\C^{N\times N}$, and suppose we are given a block encoding
circuit $\mathcal{U_{\matA}^{\text{BE}}}\in\BE{\matA}{\alpha}$ with
$q(\mathcal{U_{\matA}^{\text{BE}}})$ qubits. We show that we can
construct $\mathcal{U}_{\matA}^{\textrm{SP}}\in\MS{\matA}{\sqrt{N}\alpha}$
with $q(\mathcal{U_{\matA}^{\text{BE}}})+q_{N}$ qubits, where $q_{N}=q(\mathcal{U_{\matA}^{\text{BE}}})-\log_{2}N$
is the number of ancillary qubits in the block encoding. The key idea
is to use an extended version of ${\cal U}_{\matI_{N}}^{\text{SP}}$
which is obtained by padding it with zero rows (see our previous work~\cite{yosef2024multivariate}
for how to build ${\cal U}_{\matI_{N}}^{\text{SP}}$ and pad it with
zero rows). Let ${\cal U}_{\hat{\matI}_{N}}^{\text{SP}}$ be a matrix
state preparation for
\[
\hat{\matI}_{N}=\left[\begin{array}{c}
\matI_{N}\\
\mat 0_{(2^{q_{N}}-1)N}
\end{array}\right].
\]
 Then,
\begin{align*}
\mathcal{U_{\matA}^{\text{BE}}}\cdotp_{1}{\cal U}_{\hat{\matI}_{N}}^{\text{SP}}\Ket 0_{q(\mathcal{U_{\matA}^{\text{BE}}})+q_{N}} & =\mathcal{U_{\matA}^{\text{BE}}}\cdotp_{1}\kket{\left[\begin{array}{c}
\frac{1}{\sqrt{N}}\matI_{N}\\
\mat 0_{(2^{q_{N}}-1)N}
\end{array}\right]}\\
 & =\kket{\matM(\mathcal{U_{\matA}^{\text{BE}}})\left[\begin{array}{c}
\frac{1}{\sqrt{N}}\matI_{N}\\
\mat 0_{(2^{q_{N}}-1)N}
\end{array}\right]}\\
 & =\kket{\left[\begin{array}{cc}
\alpha^{-1}\matA & *\\*
* & *
\end{array}\right]\left[\begin{array}{c}
\frac{1}{\sqrt{N}}\matI_{N}\\
\mat 0_{(2^{q_{N}}-1)N}
\end{array}\right]}\\
 & =\kket{\left[\begin{array}{c}
(\sqrt{N}\alpha)^{-1}\matA\\*
*
\end{array}\right]}
\end{align*}
and we have $\mathcal{U}_{\matA}^{\textrm{SP}}\in\MS{\matA}{\sqrt{N}\alpha}$,
i.e., a state preparation for $\matA$ with $q(\mathcal{U_{\matA}^{\text{BE}}})+q_{N}$
qubits.

\subsubsection{\label{subsubsec:PauliMsp2Be_Msp2Be}From matrix state preparation
to block encoding}

In this subsection, the goal is to construct $\mathcal{U_{\matA}^{\text{BE}}}$
starting from $\mathcal{U}_{\matA}^{\textrm{SP}}$ or $\mathcal{U}_{\tenA_{P}}^{\text{SP}}$.
Again, due to the construction $\mathcal{U}_{\tenA_{P}}^{\text{SP}}=\cirU^{(q)}\cirU_{\matA}^{\text{SP}}$
in Subsection~\ref{subsec:ua-to-uap}, it suffices to establish the
conversion from $\mathcal{U}_{\tenA_{P}}^{\text{SP}}$ to $\mathcal{U_{\matA}^{\text{BE}}}$.

In Section~\ref{subsubsec:PauliMatrix2Be} we show that given a classical
Pauli coefficient hypermatrix we can construct block encoding. If
we inspect the method, we see in Eq.~(\ref{eq:uap-to-be}) that the
final circuit uses a state preparation circuit for $\tenA_{P}$. Thus,
if we are given a state preparation for $\tenA_{P}$ we can just proceed
to building the block encoding circuit. The final circuit is:
\[
\mathcal{U}_{\mathbf{A}}^{\text{BE}}=\Bigl(\mathcal{U}_{\mathcal{A}_{P}}^{\text{SP}}\otimes\mathbf{{\cal I}}_{\log N}\Bigr)\cdot({\cal I}_{p}\otimes\mathcal{PMX}_{N})\cdot\Bigl({\cal I}_{p}\otimes\calH_{2\log N}\otimes\mathbf{{\cal I}}_{\log N}\Bigr)
\]
In the case where $\mathcal{U}_{{\cal A}_{P}}^{\text{SP}}\in\TS{{\cal A}_{P}}{\alpha}$
, the scale of $\mathcal{U}_{\mathbf{A}}^{\text{BE}}$ is $N\alpha$
($\sqrt{N}\FNorm{\mathbf{A}}$ for exact state preparation). The classical
cost of constructing the circuit is $O(g(\mathcal{U}_{\mathcal{A}_{P}}^{\text{SP}}))$
dominated by the classical description of the state preparation component.
The resulting circuit acts on $3n+p$ qubits, has depth of $d(\mathcal{U}_{\mathcal{A}_{P}}^{\text{SP}})+12$,
and its T-cost equals the T-cost of the state preparation plus one.

\subsection*{Acknowledgments.}

This research was supported by the US-Israel Binational Science Foundation
(Grant no. 2017698), Israel Science Foundation (Grant no. 1524/23)
and IBM Faculty Award. Liron Mor-Yosef acknowledges support by the
Milner Foundation and the Israel Council for Higher Education.

\bibliographystyle{plain}
\bibliography{references}

\appendix

\end{document}

%% file: lyxmacros.tex
\global\long\def\R{\mathbb{R}}%

\global\long\def\C{\mathbb{C}}%

\global\long\def\N{\mathbb{N}}%

\global\long\def\HH{\mathbb{H}}%

\global\long\def\PP{\mathbb{P}}%

\global\long\def\e{{\mathbf{e}}}%

\global\long\def\et#1{{\e(#1)}}%

\global\long\def\ef{{\mathbf{\et{\cdot}}}}%

\global\long\def\a{{\mathbf{a}}}%

\global\long\def\x{{\mathbf{x}}}%

\global\long\def\xt#1{{\x(#1)}}%

\global\long\def\xf{{\mathbf{\xt{\cdot}}}}%

\global\long\def\d{{\mathbf{d}}}%

\global\long\def\w{{\mathbf{w}}}%

\global\long\def\b{{\mathbf{b}}}%

\global\long\def\u{{\mathbf{u}}}%

\global\long\def\y{{\mathbf{y}}}%

\global\long\def\k{{\mathbf{k}}}%

\global\long\def\yt#1{{\y(#1)}}%

\global\long\def\yf{{\mathbf{\yt{\cdot}}}}%

\global\long\def\z{{\mathbf{z}}}%

\global\long\def\v{{\mathbf{v}}}%

\global\long\def\h{{\mathbf{h}}}%

\global\long\def\s{{\mathbf{s}}}%

\global\long\def\c{{\mathbf{c}}}%

\global\long\def\p{{\mathbf{p}}}%

\global\long\def\f{{\mathbf{f}}}%

\global\long\def\t{{\mathbf{t}}}%

\global\long\def\rb{{\mathbf{r}}}%

\global\long\def\rt#1{{\rb(#1)}}%

\global\long\def\rf{{\mathbf{\rt{\cdot}}}}%

\global\long\def\strs{{\mathbf{\textsc{s}}}}%

\global\long\def\strt{{\mathbf{\textsc{t}}}}%

\global\long\def\strb{{\mathbf{\textsc{b}}}}%

\global\long\def\mat#1{{\ensuremath{\bm{\mathrm{#1}}}}}%

\global\long\def\matN{\ensuremath{{\bm{\mathrm{N}}}}}%

\global\long\def\matX{\ensuremath{{\bm{\mathrm{X}}}}}%

\global\long\def\matK{\ensuremath{{\bm{\mathrm{K}}}}}%

\global\long\def\matA{\ensuremath{{\bm{\mathrm{A}}}}}%

\global\long\def\matB{\ensuremath{{\bm{\mathrm{B}}}}}%

\global\long\def\matC{\ensuremath{{\bm{\mathrm{C}}}}}%

\global\long\def\matD{\ensuremath{{\bm{\mathrm{D}}}}}%

\global\long\def\matE{\ensuremath{{\bm{\mathrm{E}}}}}%

\global\long\def\matF{\ensuremath{{\bm{\mathrm{F}}}}}%

\global\long\def\matO{\ensuremath{{\bm{\mathrm{O}}}}}%

\global\long\def\matQ{\ensuremath{{\bm{\mathrm{Q}}}}}%

\global\long\def\matP{\ensuremath{{\bm{\mathrm{P}}}}}%

\global\long\def\matU{\ensuremath{{\bm{\mathrm{U}}}}}%

\global\long\def\matV{\ensuremath{{\bm{\mathrm{V}}}}}%

\global\long\def\matM{\ensuremath{{\bm{\mathrm{M}}}}}%

\global\long\def\matG{\ensuremath{{\bm{\mathrm{G}}}}}%

\global\long\def\cirU{{\cal U}}%

\global\long\def\cirV{{\cal V}}%

\global\long\def\cirW{{\cal W}}%

\global\long\def\tenA{{\cal A}}%

\global\long\def\calX{{\cal X}}%

\global\long\def\calY{{\cal Y}}%

\global\long\def\calH{{\cal H}}%

\global\long\def\calW{{\cal W}}%

\global\long\def\calS{{\cal S}}%

\global\long\def\calT{{\cal T}}%

\global\long\def\matR{\mat R}%

\global\long\def\matS{\mat S}%

\global\long\def\matO{\mat O}%

\global\long\def\matT{\mat T}%

\global\long\def\matY{\mat Y}%

\global\long\def\matI{\mat I}%

\global\long\def\matJ{\mat J}%

\global\long\def\matZ{\mat Z}%

\global\long\def\matW{\mat W}%

\global\long\def\tmatK{\widetilde{\matK}}%

\global\long\def\matL{\mat L}%

\global\long\def\matZero{\mat 0}%

\global\long\def\S#1{{\mathbb{S}_{N}[#1]}}%

\global\long\def\IS#1{{\mathbb{S}_{N}^{-1}[#1]}}%

\global\long\def\PN{\mathbb{P}_{N}}%

\global\long\def\Norm#1{\|#1\|}%

\global\long\def\NormS#1{\|#1\|^{2}}%

\global\long\def\ONorm#1{\|#1\|_{op}}%

\global\long\def\ONormS#1{\|#1\|_{op}^{2}}%

\global\long\def\OneNorm#1{\|#1\|_{1}}%

\global\long\def\OneNormS#1{\|#1\|_{1}^{2}}%

\global\long\def\TNormS#1{\|#1\|_{2}^{2}}%

\global\long\def\TNorm#1{\|#1\|_{2}}%

\global\long\def\InfNorm#1{\|#1\|_{\infty}}%

\global\long\def\InfNormS#1{\|#1\|_{\infty}^{2}}%

\global\long\def\FNorm#1{\|#1\|_{F}}%

\global\long\def\FNormK#1#2{\|#1\|_{F}^{#2}}%

\global\long\def\FNormS#1{\|#1\|_{F}^{2}}%

\global\long\def\UNorm#1{\|#1\|_{\matU}}%

\global\long\def\NucNorm#1{\|#1\|_{\star}}%

\global\long\def\UNormS#1{\|#1\|_{\matU}^{2}}%

\global\long\def\UINormS#1{\|#1\|_{\matU^{-1}}^{2}}%

\global\long\def\ANorm#1{\|#1\|_{\matA}}%

\global\long\def\BNorm#1{\|#1\|_{\mat B}}%

\global\long\def\ANormS#1{\|#1\|_{\matA}^{2}}%

\global\long\def\AINormS#1{\|#1\|_{\matA^{-1}}^{2}}%

\global\long\def\T{\textsc{T}}%

\global\long\def\H{\textsc{H}}%

\global\long\def\conj{\textsc{*}}%

\global\long\def\pinv{\textsc{+}}%

\global\long\def\Var#1{{\mathbb{V}}\left[#1\right]}%

\global\long\def\Expect#1{{\mathbb{E}}\left[#1\right]}%

\global\long\def\inner#1#2{\left\langle #1,#2\right\rangle }%

\global\long\def\ExpectC#1#2{{\mathbb{E}}_{#1}\left[#2\right]}%

\global\long\def\dotprod#1#2#3{(#1,#2)_{#3}}%

\global\long\def\dotprodN#1#2{(#1,#2)_{{\cal N}}}%

\global\long\def\dotprodH#1#2{(#1,#2)_{{\cal {\cal H}}}}%

\global\long\def\dotprodsqr#1#2#3{(#1,#2)_{#3}^{2}}%

\global\long\def\Trace#1{{\bf Tr}\left(#1\right)}%

\global\long\def\STrace#1{{\bf Tr}(#1)}%

\global\long\def\MTrace#1#2{{\bf MTr}_{#2}\left(#1\right)}%

\global\long\def\SMTrace#1#2{{\bf MTr}_{#2}(#1)}%

\global\long\def\nnz#1{{\bf nnz}\left(#1\right)}%

\global\long\def\vec#1{{\bf vec}\left(#1\right)}%

\global\long\def\rvec#1{{\bf rvec}\left(#1\right)}%

\global\long\def\cvec#1{{\bf cvec}\left(#1\right)}%

\global\long\def\Mat#1#2#3{{\bf mat}_{#1\times#2}\left(#3\right)}%

\global\long\def\rank#1{{\bf rank}\left(#1\right)}%

\global\long\def\vol#1{{\bf vol}\left(#1\right)}%

\global\long\def\cirmat#1{\matM\left(#1\right)}%

\global\long\def\range#1{{\bf range}\left(#1\right)}%

\global\long\def\order#1{{\bf order}\left(#1\right)}%

\global\long\def\sr#1{{\bf sr}\left(#1\right)}%

\global\long\def\poly#1{{\bf poly}\left(#1\right)}%

\global\long\def\sign#1{{\bf sign}\left(#1\right)}%

\global\long\def\amplitude#1#2{{\bf Amplitude}_{#2}\left(#1\right)}%

\global\long\def\gap#1#2{{\bf gap}_{#2}\left(#1\right)}%

\global\long\def\gapS#1#2{{\bf gap}_{#2}^{2}\left(#1\right)}%

\global\long\def\bit#1#2{{\bf bit}_{#2}\left(#1\right)}%

\global\long\def\diag#1{{\bf diag}\left(#1\right)}%

\global\long\def\blkdiag#1{{\bf blkdiag}\left(#1\right)}%

\global\long\def\Re#1{{\bf Re}\left(#1\right)}%

\global\long\def\Img#1{{\bf Im}\left(#1\right)}%

\global\long\def\Gr#1#2{{\bf Gr}(#1,#2)}%

\global\long\def\Ket#1{\left\vert #1\right\rangle }%

\global\long\def\Bra#1{\left\langle #1\right\vert }%

\global\long\def\ket#1{\left\vert #1\right\rangle }%

\global\long\def\bra#1{\left\langle #1\right\vert }%

\global\long\def\kket#1{\left\vert \left\vert #1\right\rangle \right\rangle }%

\global\long\def\bbra#1{\left\langle \left\langle #1\right\vert \right\vert }%

\global\long\def\Braket#1#2{\left\langle #1|#2\right\rangle }%

\global\long\def\Ptrace#1#2{{\bf {\bf Tr}}_{#2}\left(#1\right)}%

\global\long\def\SPtrace#1#2{{\bf {\bf Tr}}_{#2}(#1)}%

\global\long\def\vtheta{\mat{\theta}}%

\global\long\def\valpha{\mat{\alpha}}%

\global\long\def\vbeta{\mat{\beta}}%

\global\long\def\vpsi{\mat{\psi}}%

\global\long\def\vphi{\mat{\phi}}%

\global\long\def\vxi{\mat{\xi}}%

\global\long\def\abs#1{\lvert#1\rvert}%

\global\long\def\HT{\operatorname{HT}}%

\global\long\def\ST{\operatorname{ST}}%

\global\long\def\BE#1#2{{\bf BE}_{#2}\left(#1\right)}%

\global\long\def\SP#1#2{{\bf SP}_{#2}\left(#1\right)}%

\global\long\def\HBE#1#2{{\bf HBE}_{#2}\left(#1\right)}%

\global\long\def\MS#1#2{{\bf MS}_{#2}\left(#1\right)}%

\global\long\def\TS#1#2{{\bf HS}_{#2}\left(#1\right)}%

\global\long\def\CirMat#1{\matM\left(#1\right)}%

\global\long\def\MUX#1#2{{\bf MUX}\left(#1,#2\right)}%

\global\long\def\be{{\bf (BE)}}%

\global\long\def\sp{{\bf (SP)}}%

\global\long\def\si{\text{i}}%